\PassOptionsToPackage{psdextra}{hyperref}
\documentclass[acmsmall,screen]{acmart}

\AtBeginDocument{%
  }

\setcopyright{rightsretained}
\acmDOI{10.1145/3656450}
\acmYear{2024}
\copyrightyear{2024}
\acmSubmissionID{pldi24main-p621-p}
\acmJournal{PACMPL}
\acmVolume{8}
\acmNumber{PLDI}
\acmArticle{220}
\acmMonth{6}
\received{2023-11-16}
\received[accepted]{2024-03-31}

\usepackage{graphicx} %
\usepackage{xcolor}
\usepackage{soul}
\usepackage{amsthm}
\usepackage{amsmath}
\usepackage{xspace}
\usepackage[shortlabels]{enumitem}
\usepackage{proof}
\usepackage{mathpartir}
\usepackage{alltt}
\usepackage{fancyvrb}
\usepackage{tikz}
\usepackage{subcaption}
\usepackage{multirow}
\usetikzlibrary{positioning}
\usetikzlibrary{arrows.meta}
\usepackage{float}
\usepackage{caption}
\usepackage{wrapfig}

\usepackage{framed}

\newcommand{\cody}[1]{}
\newcommand{\amcomment}[1]{}
\newcommand{\madhu}[1]{}

\sethlcolor{red!10}

\newcommand{\inrevA}[1]{#1}

\newcommand{\Ss}{\mathcal{S}}

\newcommand{\Ff}{\mathcal{F}}
\newcommand{\Gg}{\mathcal{G}}
\newcommand{\Ll}{\mathcal{L}}

\newcommand{\soexists}{{\mathbf{\exists}\!\!\!\exists}}

\newcommand{\sortedll}{\mathit{sortedll}}
\newcommand{\sorted}{\mathit{sorted}}
\newcommand{\revsorted}{\mathit{rev\_sorted}}
\newcommand{\nil}{\mathit{nil}}
\newcommand{\nxt}{\mathit{next}}
\newcommand{\key}{\mathit{key}}
\newcommand{\prev}{\mathit{prev}}

\newcommand{\tree}{\mathit{tree}}
\newcommand{\rank}{\mathit{rank}}
\newcommand{\htree}{\mathit{htree}}
\newcommand{\len}{\mathit{length}}
\newcommand{\revlen}
{\mathit{rev\_length}}
\newcommand{\keys}{\mathit{keys}}
\newcommand{\old}{\mathit{old}}
\newcommand{\hslist}{\mathit{hslist}}
\newcommand{\lst}
{\mathit{last}}

\newcommand{\Dafny}{{\sc Dafny}\xspace}
\newcommand{\Boogie}{{\sc Boogie}\xspace}
\newcommand{\bfBoogie}{{\sc \textbf{Boogie}}\xspace}

\newcommand{\LC}{\mathit{LC}}
\newcommand{\Br}{\mathit{Br}}
\newcommand{\wb}{\vdash_{\mathrm{WB}}}
\newcommand{\ite}{\mathit{ite}}

\definecolor{nicepink}{rgb}{0.8,0.0,0.8}
\definecolor{darkpastelgreen}{rgb}{0.01, 0.75, 0.24}
\definecolor{chocolate}{rgb}{0.48, 0.25, 0.0}
\definecolor{darkviolet}{rgb}{0.58, 0.0, 0.83}
\definecolor{keywordcolor}{rgb}{0.9, 0.4, 0.0}
\newcommand{\ghost}[1]{\textcolor{blue}{#1}}
\newcommand{\autobr}[1]{\textcolor{nicepink}{#1}}

\newcommand{\codecomment}[1]{\textcolor{gray}{#1}}
\newcommand{\type}[1]{\textcolor{darkpastelgreen}{#1}}
\newcommand{\funcname}[1]{\textcolor{chocolate}{#1}}
\newcommand{\annotation}[1]{\textcolor{darkviolet}{#1}}
\newcommand{\keyword}[1]{\textcolor{keywordcolor}{#1}}

\newtheorem{theorem}{Theorem}[section]
\newtheorem{lemma}[theorem]{Lemma}

\newtheorem{proposition}[theorem]{Proposition}

\newcommand{\mypara}[1]{\smallskip\noindent\emph{\textbf{#1}.\ }}

\newcommand{\Alloc}{\mathit{Alloc}}
\newcommand{\Mod}{\mathit{Mod}}

\newif\iflong
\longfalse

\begin{document}

\title{Predictable Verification using Intrinsic Definitions}

\author{Adithya Murali}
\orcid{0000-0002-6311-1467}
\affiliation{%
  \institution{University of Illinois at Urbana-Champaign}
  \city{Urbana}
  \country{USA}
}
\email{adithya5@illinois.edu}

\author{Cody Rivera}
\orcid{0000-0001-7824-4054}
\affiliation{%
  \institution{University of Illinois at Urbana-Champaign}
  \city{Urbana}
  \country{USA}
}
\email{codyjr3@illinois.edu}

\author{P. Madhusudan}
\orcid{0000-0002-9782-721X}
\affiliation{%
  \institution{University of Illinois at Urbana-Champaign}
  \city{Urbana}
  \country{USA}
}
\email{madhu@illinois.edu}

\begin{abstract}
We propose a novel mechanism of defining data structures using \emph{intrinsic definitions} that avoids recursion and instead utilizes \emph{monadic maps satisfying local conditions}. We show that intrinsic definitions are a powerful mechanism that can capture a variety of data structures naturally. We show that they also enable a predictable verification methodology that allows engineers to write ghost code to update monadic maps and perform verification using reduction to decidable logics. We evaluate our methodology using {\sc Boogie} and prove a suite of data structure manipulating programs correct. 
\end{abstract}

\begin{CCSXML}
<ccs2012>
<concept>
<concept_id>10011007.10011074.10011099.10011692</concept_id>
<concept_desc>Software and its engineering~Formal software verification</concept_desc>
<concept_significance>500</concept_significance>
</concept>
<concept>
<concept_id>10003752.10003790.10002990</concept_id>
<concept_desc>Theory of computation~Logic and verification</concept_desc>
<concept_significance>500</concept_significance>
</concept>
<concept>
<concept_id>10003752.10003790.10003794</concept_id>
<concept_desc>Theory of computation~Automated reasoning</concept_desc>
<concept_significance>500</concept_significance>
</concept>
</ccs2012>
\end{CCSXML}

\ccsdesc[500]{Software and its engineering~Formal software verification}
\ccsdesc[500]{Theory of computation~Logic and verification}
\ccsdesc[500]{Theory of computation~Automated reasoning}

\keywords{Predictable Verification, Intrinsic Definitions, Verification of Linked Data Structures, Decidability, Ghost-Code Annotations}

\maketitle

\section{Introduction}
\label{sec:intro}

In computer science in general, and program verification in particular, classes of finite structures (such as data structures) are commonly defined using \emph{recursive definitions (aka inductive definitions)}. Proving that a set of structures is in such a class or proving that structures in the class have a property is naturally performed using \emph{induction}, typically mirroring the recursive structure in its definition. For example, trees in pointer-based heaps can be defined using the following recursive definition in first-order logic (FOL) with least fixpoint semantics for definitions:
\begin{align}
\begin{split}
\textit{tree}(x)  :: =_\textit{lfp} & ~ x=\nil  \vee \big(x \not = \nil \wedge \textit{tree}(l(x)) \wedge \textit{tree}(r(x)) \\
& ~~\wedge 
x \not \in \htree(l(x)) \wedge 
x \not \in \htree(r(x)) \wedge 
\htree(l(x)) \cap \htree(r(x)) = \emptyset\big) \\
\htree(x) :: =_\textit{lfp} &~ 
ite \left(x=\nil, ~\emptyset, 
 ~\htree\left(l(x)\right) \cup \htree\left(r(x)\right) \cup \{x\}\right)
\end{split}
\end{align}

In the above, $\htree$ maps each location $x$ in the heap to the set of all locations reachable from $x$ using $l$ and $r$ pointers, and the definition of $\tree$ uses this to ensure that the left and right trees are disjoint from each other and the root. Definitions in separation logic are similar (with heaplets being implicitly defined, and disjointness expressed using the separating conjunction '$\star$'~\cite{reynolds02,seplogicprimer,ohearn01}).

When performing imperative program verification, we annotate programs with loop invariants and contracts for methods, and reduce verification
to validation of Hoare triples of the form $\{\alpha\}s \{\beta\}$, where $s$ is a straight-line program (potentially with calls to other methods encoded using their contracts). The validity of each Hoare triple is translated to a pure logical validity question, called the \emph{verification condition} (VC). When $\alpha$ and $\beta$ refer to data structure properties, the resulting VCs are typically proved using induction on the structure of the recursive definitions. Automation of program verification reduces to automating validity of the logic the VCs are expressed in. 

Logics that are powerful enough to express rich properties of data structures are invariably {incomplete}, not just undecidable, i.e., they do not admit any automated procedure that is complete (guaranteed to eventually prove any valid theorem, but need not terminate on invalid theorems). 
For instance, validity is incomplete for both first-order logic with least fixpoints and separation logic. Consequently, though verification frameworks like {\sc Dafny}~\cite{dafny} support rich specification languages, validation of verification conditions can fail even for valid Hoare triples. Automated verification engines hence support several heuristics resulting in sound but incomplete verification. 

When proofs succeed in such systems, the verification engineer is happy that automation has taken the proof through. However, when proofs \emph{fail}, as they often do, the verification engineer is stuck and perplexed. First, they would crosscheck to see whether their annotations are strong enough and that the Hoare triples are indeed valid. If they believe they are, they do not have clear guidelines to help the tool overcome its incompleteness. Engineers are instead required to know the \emph{underlying proof mechanisms/heuristics} the verification system uses in order to figure out why the system is unable to succeed, and figure out how to help the system. For instance, for data structures with recursive definitions, the proof system may just unfold definitions a few times, and the engineer must be able to see why this heuristic will not be able to prove the theorem and formulate new inductively provable lemmas or quantification triggers that can help. 
Such \emph{unpredictable} verification systems that require engineers to know their internal heuristics and proof mechanisms frustrate verification experience. 

\mypara{Predictable Verification} 
In this paper, we seek an entirely new paradigm of \emph{predictable} verification. We want a technique where:
\begin{description}
\item [(a)] the verification engineer is asked to provide upfront a set of annotations that help prove programs correct, where these annotations are entirely \emph{independent} of the verification mechanisms/tools, and 
\item [(b)] the program verification problem, given these annotations, is guaranteed to be \emph{decidable} (and preferably decidable using efficient engines such as SMT solvers).
\end{description}

The upfront agreement on the information
that the verification engineer is required to provide makes their task crystal clear. The fact that the verification is decidable given these annotations ensures that the verification engine, given enough resources of time and space (of course) will eventually return proving the program correct or showing that the program or annotations are incorrect. There is no second-guessing by the engineer as the verification will never fail on valid theorems, and hence they need not worry about knowing how the verification engine works, or give further help. Note that the verification \emph{without annotations} can (and typically will be) undecidable.

\mypara{Intrinsic Definitions of Data Structures}
In this paper, we propose an entirely new way of defining data structures, called \emph{intrinsic definitions}, that facilitates a predictable verification paradigm 
for proving their maintenance. 
Rather than defining data structures using recursion, like in equation (1) above (which naturally calls for inductive proofs and invariably entails incompleteness), we define data structures by augmenting each location with additional information using \emph{ghost maps} and demanding that certain \emph{local conditions} hold between each location and its neighbors. 

Intrinsic definitions formally require a set of monadic maps (maps of arity one) that associate values to each location in a structure (we can think of these as ghost fields associated with each location/object). We demand that the monadic maps on local neighborhoods of every location satisfy certain logical conditions. The existence of maps that satisfy the local logical conditions ensures that the structure is a valid data structure.  

For example, we can capture trees in pointer-based heaps in the following way. Let us introduce maps $\tree: Loc \rightarrow Bool$, $\rank: Loc \rightarrow \mathbb{Q}^+$ (non-negative rationals), and $p: Loc \rightarrow Loc$ (for ``parent''),  
and demand the following local property:
\begin{align*}
\begin{split}
 \forall x::Loc. (  \tree(x) \Rightarrow & 
  (~(l(x) \not = \nil \Rightarrow (\tree(l(x)) \wedge p(l(x))=x \wedge 
         \rank(l(x)) < \rank(x))) \\
 & \wedge (r(x) \not = \nil \Rightarrow (\tree(r(x)) \wedge p(r(x))=x \wedge 
         \rank(r(x)) < \rank(x))) \\
 & \wedge ((l(x) \not = \nil \land r(x) \neq \nil) \Rightarrow l(x) \neq r(x))\\
 & \wedge (p(x) \not = \nil \Rightarrow (r(p(x))=x \vee l(p(x))=x)) )~~)
\end{split}
\end{align*}

The above demands that ranks become smaller as one descends the tree, that a node is the parent of its children, and that a node is either the left or right child of its parent. 

Given a \emph{finite} heap, it is easy to see that if there exist maps $\tree$, $\rank$ and $p$ that satisfy the above property, and
if $\textit{tree}(l)$ is true for a location $l$, then $l$ must point to a tree (strictly decreasing ranks ensure that there are no cycles and existence of a unique parent  ensures that there are no ``merges''). Furthermore, in any heap, if $T$ is the subset of locations that are roots of  trees, then there are maps that satisfy the above property and have precisely $tree(l)$ to be true for locations in $T$. 

Note that the above intrinsic definition \emph{does not use recursion} or least fixpoint semantics. It simply requires maps such that each location satisfies the local neighborhood condition.

\mypara{Fix-What-You-Break Program Verification Methodology}

Intrinsic definitions are particularly attractive for proving \emph{maintenance} of structures when structures undergo mutation. 
When a program mutates a heap $H$ to a heap $H'$, we start with monadic maps that satisfy local conditions in the pre-state. As the heap $H$ is modified, we ask the verification engineer to also \emph{repair} the monadic maps, using ghost map updates, so that the local conditions on all locations are met in the heap in the post-state $H'$.

For instance, consider a program that walks down a tree from its root to a node $x$ and introduces a newly allocated node $n$ between $x$ and $x$'s right child $r$. Then we would assume in the precondition that the monadic maps $\tree$, $\rank$, and $p$ exist satisfying the local condition (2) above. After the mutation, we would simply update these maps so that $\tree(n)$ is true, $p(r)=n$, $p(n)=x$, and $rank(n)$ is, say, $(\textit{rank(x)}+\textit{rank}(r))/2$. 

The annotations required of the user, therefore, are ghost map updates to locations such that the local conditions are valid for each location. We will guarantee that checking whether the local conditions holds for each location, after the repairs, is expressible in decidable logics.

We propose a modular verification approach for verifying data structure maintenance that asks the programmer to fix what they break. 
Given a program that we want to verify, we instead verify an \emph{augmented program} that keeps track of a ghost set of \emph{broken locations} $\Br$. Broken locations are those that (potentially) do not satisfy the local condition. 
When the program destructively modifies the fields of an object/location, it and some of its neighbors (accessible using pointers from the object) may not satisfy the local condition anymore, and hence will get added to the broken set. The verification engineer must repair the monadic maps on these broken locations and ensure (through an assertion) that the local condition holds on them before removing them from the broken set $Br$. However, even while repairing monadic maps on a location, the local condition on \emph{its neighboring} locations may fail and get added to the broken set. 

We develop a \emph{fix-what-you-break (FWYB)} program verification paradigm, giving formal rules of how to augment programs with broken sets, how users can modify monadic maps, and fixed recipes of how broken sets are maintained in any program. 
In order to verify that a method $m$ maintains a data structure, we need to prove that if $m$ starts with the broken set being empty, it returns with the empty broken set. We prove this methodology sound, i.e., if the program augmented with broken sets and ghost updates is correct, then the original program maintains the data structure properties mentioned in its contracts.

\mypara{Decidable Verification of Annotated Programs} 
The general idea of using local conditions to capture global properties has been explored in the literature to reduce the complexity of proofs (e.g., iterated separation in separation logic~\cite{ReynoldsSepLogic02}; see Section~\ref{sec:relwork}). Intrinsic definitions of data structures and the fix-what-you-break program verification methodology are more specifically designed to ensure the key property of \emph{decidable verification of annotated programs} by avoiding both recursion/least-fixpoint definitions and avoiding even quantified reasoning.

The verification conditions for Hoare triples involving basic blocks of our annotated programs have the following structure. First, the precondition can be captured using \emph{uninterpreted monadic functions} that are \emph{implicitly} assumed to satisfy the local condition on each location that is not in the broken set $Br$ (avoiding universal quantification). The monadic map updates (repairs) that the verification engineer makes can be captured using map updates. The postcondition of the ghost-code augmented program can, in addition to properties of variables, assert properties of the broken set $Br$ using logics over sets. Finally, we show that capturing the modified heap after function calls can be captured using  \emph{parameterized map update} theories, that are decidable~\cite{pointwisearrays}. 
Consequently, the entire verification condition is captured in quantifier-free logics involving maps, parametric map updates, and sets over combined theories. These verification conditions are hence decidable and efficiently handled by modern SMT solvers\footnote{Assuming of course that the underlying quantifier-free theories are decidable; for example, integer multiplication in the program or in local conditions would make verification undecidable, of course.}.

\mypara{Intrinsic Definitions for Representative Data Structures and Verification in {\scshape \bfseries Boogie}}
Intrinsic definitions of data structures is a novel paradigm and capturing data structures requires thinking anew in order to formulate monadic maps and local conditions that characterize them. 

We give intrinsic definitions for several classic data structures such as linked lists, sorted lists, circular lists, trees, binary search trees, AVL trees, and red-black trees. These require novel definitions of monadic maps and local conditions. We also show how standard methods on these data structures (insertions, deletions, concatenations, rotations, balancing, etc.) can be verified using the fix-what-you-break strategy and standard loop invariant/contract annotations. 
We also consider \emph{overlaid data structures} 
consisting of multiple data structures overlapping and
sharing locations. In particular, we model the core of an overlaid data structure that is used in an I/O scheduler in Linux that has a linked list (modeling a FIFO queue) overlaid on a binary search tree (for efficient search over a key field). Intrinsic definitions beautifully capture such structures by compositionally combining the instrinsic definitions for each structure and a local condition linking them together. We show methods to modify this structure are provable using fix-what-you-break verification.

We model the above data structures and the annotated methods in the low-level programming language {\sc Boogie}. {\sc Boogie} \inrevA{is} an intermediate programming language with verification support that several high-level programming languages compile to for verification (e.g., C~\cite{vcc,havoc}, {\sc Dafny}~\cite{dafny}, {\sc Civl}~\cite{civl}, Move~\cite{move}). These annotated programs do not use quantifiers or recursive definitions, and {\sc Boogie} is able to verify them automatically using decidable verification in negligible time, without further user-help. 

\mypara{Contributions}
The paper makes the following contributions:
\begin{itemize}
    \item A new paradigm of \emph{predictable verification} that asks upfront for programmatic annotations and ensures annotated program verification is decidable, without reliance on users to give heuristics and tactics.
    \item A novel notion of intrinsic definitions of  data structures based on ghost monadic maps and local conditions.
    \item A predictable verification methodology for programs that manipulate data structures with intrinsic definitions following a fix-what-you-break (FWYB) methodology.
    \item Intrinsic definitions for several classic data structures, and fix-what-you-break annotations
    for programs that manipulate such structures, with realization of these programs and their verification using {\sc Boogie}.
\end{itemize}

\section{Intrinsic Definitions of Data Structures: The Framework}
\label{sec:datastructures}

In this section we present the first main contribution of our paper, the framework of  intrinsically defined data structures. We first define the notion of a data structure in a pointer-based heap.

\subsection{Data Structures}
\label{sec:ds}

In this paper, we think of data structures defined using a \emph{class} $C$ of objects. The class $C$ can coexist with other classes, heaps, and data structures, potentially modeled and reasoned with using other mechanisms. For technical exposition and simplicity, we restrict the technical definitions to a single class of data structures over a class $C$. 

A class $C$ has a signature $(\Ss, \Ff)$ consisting of a finite set of sorts $\Ss = \{\sigma_0, \sigma_1 \ldots, \sigma_n\}$ and a finite set of fields $\Ff = \{f_1, f_2 \ldots, f_m\}$. We assume without loss of generality that the sort $\sigma_0$ represents the sort of objects of the class $C$, and we denote this sort by $C$ itself. We use $C$ to model objects in the heap. %
The other ``background'' sorts, e.g., integers, are used to model the values of the objects' fields. Each field $f_i: C \rightarrow \sigma$ is a unary function symbol and is used to model pointer and data fields of heap locations/objects. We model $\nil$ as a non-object value and denote the sort $C \uplus \{\nil\}$ consisting of objects as well as the $\nil$ value by $C?$.

A $C$-heap $H$ is a \emph{finite}  first-order model of the signature of $C$. More formally, it is a pair $(O, I)$ where $O$ is a finite set of \emph{objects} interpreting the foreground sort $C$ and $I$ is an interpretation of every field in $\Ff$ for every object in $O$.

\begin{minipage}{0.65\textwidth}
\begin{example}[$C$-Heap]
\label{ex:c-heap}
Let $C$ be the class consisting of a pointer field $\nxt: C \rightarrow C?$ and a data field $\key: C \rightarrow \mathit{Int}$. The figure on the right represents a $C$-heap consisting of objects $O= \{o_1,o_2\}$ and the illustrated interpretation $I$ for $\nxt$ and $\key$.\qed
\end{example}
\end{minipage}
\hspace{5pt}
\begin{minipage}{0.3\textwidth}
\centering
\begin{tikzpicture}[node distance={16mm}, main/.style = {draw, circle}]
\node[main] (o1) {$o_1$};
\node[main] (o2) [right of=o1] {$o_2$};
\node (nil) [right of=o2] {$\nil$};
\node (o1k) [below=6mm of o1] {$1$};
\node (o2k) [below=6mm of o2] {$2$};
\draw[->] (o1) -- (o2) node[above,midway] {$\nxt$};
\draw[->] (o2) -- (nil) node[above,midway] {$\nxt$};
\draw[->] (o1) -- (o1k) node[right,midway] {$\key$};
\draw[->] (o2) -- (o2k) node[right,midway] {$\key$};
\end{tikzpicture}
\end{minipage}

\smallskip
\noindent We now define a data structure. We fix a class $C$. 

\begin{definition}[Data Structure]
\label{defn:ds}
A data structure $D$ of arity $k$ is a set of triples of the form $(O, I, \overline{o})$ such that $(O,I)$ is a $C$-heap and $\overline{o}$ is a $k$-tuple of objects from $O$.\qed
\end{definition}

Informally, a data structure is a particular subset of $C$-heaps along with a distinguished tuple of locations $\overline{o}$ in the heap that serve as the ``entry points'' into the data structure, such as the root of a tree or the ends of a linked list segment. 

\begin{example}[Sorted Linked List]
\label{ex:ds-sorted-list}
Let $C$ be the class defined in Example~\ref{ex:c-heap}. The data structure %
of sorted linked lists is the set of all $(O, I, o_1)$ such that $O$ contains objects $o_1, o_2 \ldots o_n$ %
with the interpretation $\nxt(o_i) = o_{i+1}$ and $\key(o_i) \leq \key(o_{i+1})$ for every $1 \leq i < n$, and $\mathit{next}(o_n) = \nil$. For example, let $(O, I)$ be the $C$-heap described in Example~\ref{ex:c-heap}. The triple $(O, I, o_1)$ is an example of a sorted linked list. %
Here $o_1$ represents the head of the sorted linked list.\qed
\end{example}

\subsection{Intrinsic Definitions of Data Structures}
\label{sec:ids}

In this work, we propose a characterization of data structures using \emph{intrinsic definitions}. Intrinsic definitions consist of
a set of \emph{monadic maps} that associate (ghost) values to each object and a set of \emph{local} conditions that constrain the monadic maps on each location and its neighbors. A $C$-heap is considered to be a valid data structure  if \emph{there exists} a set of monadic maps for the heap that satisfy the local conditions.

Annotations using intrinsic definitions enable local and decidable reasoning for correctness of programs manipulating data structures using the Fix-What-You-Break (FWYB) methodology, which is described later in Section~\ref{sec:proglogic}. We develop the core idea of intrinsic definitions below.

\mypara{Ghost Monadic Maps} We denote by $C_{\Gg} = (\Ss, \Ff \cup \Gg)$ an extension of $C$ with a finite set of monadic (i.e., unary) function symbols $\Gg$. We can think of these as \emph{ghost} fields of objects.

The key idea behind intrinsic definitions is to extend a $C$-heap with a set of ghost monadic maps and formulate local conditions using the maps that characterize the heaps belonging to the data structure. The existence of such ghost maps satisfying the local conditions is then the intrinsic definition. Definitions are parameterized by a multi-sorted first-order logic $\Ll$ in which local conditions are stated. The logic has the sorts $\Ss$ and contains the function symbols in $\Ff \cup \Gg$, as well as interpreted functions over background sorts 
(such as $+$ and $<$ on integers, and $\subseteq$ on sets). %

\begin{definition}[Intrinsic Definition]
\label{defn:ids}
Let $C=(\Ss, \Ff)$ be a class. 
An intrinsic definition $\mathit{IDS}(\overline{y})$ over the class $C$ is a tuple $(\Gg, \Ll, \LC, \varphi(\overline{y}))$ where:
\begin{enumerate}
    \item $\Gg$ is a finite set of \emph{monadic map names and function signatures} disjoint from $\Ff$,
    \item $\Ll$ is a first-order logic over the sorts $\Ss$ containing the interpreted functions of the background sorts as well as the function symbols in $\Ff \cup \Gg$,
    \item A \emph{local condition} formula $\LC$ of the form $\forall x: Loc.\, \rho(x)$ such that $\rho$ is a quantifier-free $\Ll$-formula, and
    \item A \emph{correlation formula} $\varphi(\overline{y})$ that is a quantifier-free $\Ll$-formula over free variables $\overline{y} \in Loc$.\qed
\end{enumerate}

We denote an intrinsic definition by $(\Gg, \LC, \varphi(\overline{y}))$ when the logic $\Ll$ is clear from context. In this work $\Ll$ is typically a decidable combination of quantifier-free theories~\cite{nelson80,nelson-oppen1979,tinellizarba}, containing theories of integers, sets, arrays~\cite{pointwisearrays}, etc., supported effectively in practice by SMT solvers~\cite{Z3,cvc4}. %
\end{definition}

\begin{definition}[Data Structures defined by Intrinsic Definitions]
\label{defn:ds-of-ids}
Let $C=(\Ss, \Ff)$ be a class and $\mathit{IDS}(\overline{y})=(\Gg, \textit{LC}, \varphi(\overline{y}))$ be an intrinsic definition over $C$ consisting of monadic maps $\Gg$, local condition $\LC$ and correlation formula $\varphi$. The data structure defined by $\mathit{IDS}$ is precisely the set of all $(O, I, \overline{o})$ where
\emph{there exists} an interpretation $J$ that extends $I$ with interpretations for the symbols in $\Gg$ such that $O, J \models LC$ and $O, J[\overline{y} \mapsto \overline{o}] \models \varphi(\overline{y})$, where $[\overline{y} \mapsto \overline{o}]$ denotes that the free variables $\overline{y}$ are interpreted as $\overline{o}$.
\end{definition}

Informally, given a data structure $\mathit{DS}$ consisting of triples $(O, I, \overline{o})$, an intrinsic definition demands that there exist monadic maps $\Gg$ such that the $C$-heaps $(O, I)$ in the data structure can be extended with values for maps in $\Gg$ satisfying the local conditions $\LC$, and the entrypoints $\overline{o}$ are characterized in the extension by the quantifier-free formula $\varphi$. 

\begin{example}[Sorted Linked List]
\label{ex:ids-sorted-list}
Recall the data structure of sorted linked lists defined in Example~\ref{ex:ds-sorted-list}. We capture sorted linked lists by an intrinsic definition $\mathit{SortedLL}(y)$ using monadic maps $\sortedll: C \rightarrow \mathit{Bool}$ and $\rank: C \rightarrow \mathbb{Q}^+$ 
such that:
\begin{align*}
\LC \equiv &\;\forall x.\; \Big((\sortedll(x) \,\land\, \nxt(x) \neq \nil) \Rightarrow\\[-7pt] &\hspace{6em}(\sortedll(\nxt(x)) \,\land\, \rank(\nxt(x)) < \rank(x) \,\land\, \key(x) \leq \key(\nxt(x))) \Big)\\[-7pt]
\varphi(y) \equiv &\;\; \sortedll(y)
\end{align*}

In the above definition the $\rank$ field decreases wherever $\sortedll$\, holds as we take the $\textit{next}$ pointer, and hence assures that there are no cycles. Observe that without the constraint on $\rank$, the triple  $(\{o_1,o_2\}, I, o_1)$ where $I = \{\nxt(o_1) = o_2, \nxt(o_2) = o_1, \key(o_1) = \key(o_2) = 0\}$ denoting a two-element circular list would satisfy the definition, which is undesirable. 

Note that the above allows for a heap to contain both sorted lists as well as unsorted lists. We are guaranteed by the local condition that the set of all objects where $\textit{sortedll}$ is true will be the heads of sorted lists.

We can also replace the domain of ranks in the above definition using any strict partial order, say integers or reals (with the usual $<$ order on them), and the definition will continue to define sorted lists. Well-foundedness of the order is not important as heaps are \emph{finite} in our work (see definition of $C$-heaps in Section~\ref{sec:ds})\qed 
\end{example}

\iflong
\section{Fix What You Break (FWYB) Verification Methodology}
\label{sec:proglogic}

In this section we present the second main contribution of this paper: the Fix-What-You-Break (FWYB) methodology for verifying programs with respect to data structure properties expressed using intrinsic definitions. We begin by describing a while programming language and defining the verification problem we study.

\subsection{Programs, Contracts, and Correctness}
\label{sec:triples}

We fix a class $C = (\Ss, \Ff)$ throughout this section.

\mypara{Programs} Figure~\ref{fig:prog-lang} shows the programming language used in this work. It contains assignments, field lookups, and field mutations, as is usual with programming languages for heaps. We can also use variables and expressions over other sorts. Functions can return multiple outputs. We assume that method signatures contain designated output variables, and therefore the return statement does not mention values. We also note that there is no command for deallocating objects. Instead, we assume that the language has garbage collection. Both these assumptions are true in \Dafny where we implement the FWYB methodology (Section~\ref{sec:implementation}).

\madhu{Why not add methods to grammar in Fig1? Something like Decl:= f(xbar) returns ybar $\{$ P $\}$}

\begin{figure}
\begin{align*}
P \coloneqq &\;\; x\, :=\, \nil\;\; |\;\; x\, :=\, y\;\; |\;\; v\, :=\, be \tag{Assignment} \\[-3pt]
&\;\;|\;\; y\, :=\, x.f\;\; |\;\; v\, :=\, x.d \tag{Lookup}\\[-3pt]
&\;\; |\;\; x.f\, :=\, y\;\;|\;\; x.d\, :=\, v \tag{Mutation}\\[-3pt] 
&\;\;|\; \;x\, :=\, \mathsf{new}\; C() \tag{Allocation}\\[-3pt]
&\;\;|\; \;\overline{r}\, :=\, \mathit{Function}(\overline{t}) \tag{Function Call}\\[-3pt]
&\;\; |\;\; \mathsf{skip}\;\;|\;\; \mathsf{assume}\;\mathit{cond}\;\;|\;\; \mathsf{return}\;\;|\;\; P\, ; \, P\;\;|\;\; \mathsf{if}\; \mathit{cond}\; \mathsf{then}\; P\; \mathsf{else}\; P\;\; |\;\; \mathsf{while}\; \mathsf{cond}\; \mathsf{do}\; P\;\\[-1pt]
\mathit{cond} \coloneqq &\;\; x = y\;\;|\;\; x \neq y \;\;|\;\; \mathit{be} \tag{Condition Expressions} 
\end{align*}
\caption{Grammar of while programs with recursion. $x,y$ are variables denoting objects of class $C?$ (i.e., $C$ objects or $\nil$), $v, w$ are a background sort(s) variables, $r, t$ denote variables of any sort, $f$ is a pointer field, $d$ is a data field, and $be$ is a expression of the background sort(s).}
\label{fig:prog-lang}
\end{figure}

\mypara{Operational Semantics} A program configuration consists of a store (a mapping of variables to values of the appropriate type) and a $C$-heap. More formally, it is a triple $\theta = (s, O, I)$ where $s$ is a store, $O$ is a finite set of objects of the class $C$, and $I$ is an interpretation of the function symbols in $\Ff$. This is similar to the notion of configurations for other programming languages for heaps in prior literature~\cite{seplogicprimer,reynolds02,framelogictoplas2023,kassios06}. The operational semantics is the usual one for heap manipulating programs with function calls. In particular, all dereferences must be memory safe and unsafe dereferences lead to the error state $\bot$. %
For the remainder of this section, we simply denote the transition between configurations on a program according to the operational semantics by $\theta \xrightarrow{P} \theta'$, and the satisfaction of a pre/post condition on a configuration by $\theta \models \alpha$.

\mypara{Intrinsic Hoare Triples} The verification problem we study in this paper is \emph{maintenance} of data structure properties. Fix an intrinsic definition $(\Gg, \LC, \varphi(\overline{y}))$ where $\Gg = \{g_1, g_2\ldots, g_k\}$. Let $\overline{z}$ be the input/output variables for a program that we want to verify. We consider pre and post conditions of the form 
\begin{center}
$\soexists\, g_1, g_2\ldots,g_k.\, (\LC \land \varphi(\overline{w}) \land \psi(\overline{z}))$
\end{center}

\noindent
where each $g_i$ is a ghost monadic map (unary functions over locations), $\psi$ is a quantifier-free formula over $\overline{z}$ that can use the ghost monadic maps $g_i$, and $\overline{w}$ is a tuple of variables from $\overline{z}$ whose arity is equal to $\overline{y}$. Note that the above has a second-order existential quantification ($\soexists$) over function symbols $g_1, \ldots, g_k$, and $\LC$ has first-order universal quantification over a single location variable.

Read in plain English, specifications are of the form ``$\,\overline{w}$ points to a data structure $\mathit{IDS}$ such that the (quantifier-free) property $\psi(\overline{z})$ holds''. For example, given inputs $x$ of type $C$ and $k$ of type $\mathit{Int}$ to a program we can specify a precondition that $x$ points to a sorted linked list such that the key stored at $x$ is smaller than $k$. The corresponding formula is $\soexists\,\sortedll,\rank.\, (\LC \land \sortedll(x)\land \key(x) < k)$, where $\LC$ is the local condition for a sorted linked list defined in Example~\ref{ex:ids-sorted-list}.

\medskip
\noindent
We study the validity of the following Hoare Triples:

\begin{center}
$\langle\,\alpha(\overline{x})\,\rangle\;\; \mathrm{P}(\overline{x}, \,\mathit{ret}\!:\, \overline{r})\;\; \langle\,\beta(\overline{x}, \overline{r})\,\rangle$
\end{center}

\noindent
where $\alpha$ and $\beta$ are pre and post conditions of the above form, $\mathrm{P}$ is a program, and $\overline{x}, \overline{r}$ are respectively input and output variables for $\mathrm{P}$.

\smallskip
We use a simple contract for insertion into a sorted linked list as a running example:

\begin{example}[Running Example: Insertion into a Sorted List]
\label{ex:slist-insert-triple}
Let $\mathit{SortedLL}(y) = (\Gg, \LC, \mathit{sorted}(y))$ be the intrinsic definition of a sorted linked list given in Example~\ref{ex:ids-sorted-list} where $\Gg = \{\sortedll,\rank\}$. The following Hoare triple says that insertion into a sorted list returns a sorted list:
\begin{center}
$\langle\,\soexists\, \sortedll,\rank.\, \LC \land \sortedll(x) \,\rangle\; \mathit{sorted\!\!-\!\! insert}(x,k,\, \mathit{ret}\!: x)\; \langle\,\soexists\, \sortedll,\rank.\, \LC \land \sortedll(x)\,\rangle$
\end{center}

\noindent
where $x,r$ are variables of type $C$, $k$ is of type $\mathit{Int}$ and $\mathit{sorted\!\!-\!\! insert}$ is the usual recursive method that inserts a key into a sorted linked list by recursively traversing the list starting from $x$ until it finds the appropriate place to insert $k$. We provide the relevant snippets to explain our methodology in later sections. 

Note that the input $x$ is of type $C$ (rather than $C?$) and therefore it cannot be $\nil$. We do this for simplicity of exposition; one can write a more complex contract allowing for $x$ possibly being $\nil$.
\end{example}

\mypara{Validity of Intrinsic Hoare Triples} We define the validity of Hoare Triples using the notation developed above:

\begin{definition}[Validity of Intrinsic Hoare Triples]
\label{defn:ids-triple-validity}
An intrinsic triple $\langle\, \alpha\, \rangle\, P\, \langle\, \beta\, \rangle$ is \emph{valid} if for every configuration $\theta$ such that $\theta \models \alpha$, transitioning according to $P$ starting from $\theta$ does not encounter the error state $\bot$, %
and furthermore, if $\theta \xrightarrow{P} \theta'$ for some $\theta'$, then $\theta' \models \beta$.
\end{definition}

\smallskip
\subsection*{An Overview of FWYB}

We develop the Fix-What-You-Break (FWYB) methodology in three stages, in the following subsections. We give here an overview of the methodology and the stages.

Recall that intrinsic triples are of the form $\langle\,\soexists\, g_1, g_2\ldots,g_k.\, (\LC \land \varphi \land \alpha)\,\rangle\; P\; \langle\,\soexists\, g_1, g_2\ldots,g_k.\, (\LC \land \varphi \land \beta)\,\rangle$. In Stage 1 (Section~\ref{sec:ghost-code}) we remove the second-order quantification from the monadic maps to obtain specifications of the form $\LC \land \varphi \land \psi$. We do this by re-framing the problem such that the $g_i$ maps are explicitly available in the pre state as ``ghost'' fields of objects (such that they satisfy $\LC \land \varphi \land \alpha$), and requiring the verification engineer to update the fields using \emph{ghost code} and construct the $g_i$ maps for the post state such that they satisfy $\LC \land \varphi \land \beta$. Intuitively, instead of reasoning with specifications that say there exist maps (obtained magically!) in the pre and post state that satisfy certain conditions, the verification engineer argues that for \emph{any} maps given in the pre state satisfying the precondition, we can \emph{compute} a set of corresponding maps for the post state satisfying the postcondition. This transforms the problem of study into triples of the form $\langle\,\LC \land \varphi \land \alpha\,\rangle\; P_\Gg\; \langle\,\LC \land \varphi \land \beta\,\rangle$ where $P_\Gg$ is an augmentation of $P$ with ghost code that updates the values of the fields in $\Gg$ wherever the execution of $P$ breaks the local conditions $\LC$.

In Stage 2 (Section~\ref{sec:broken-sets}) we partially remove the first-order quantification in specifications by restricting the set of objects where local conditions are required to hold. Specifically, we introduce a special variable called the \emph{broken set} $\Br$ and use it to track objects where the program's changes to the heap destroy the local conditions. Then we only require that at any point in the program, local conditions hold for objects not in the broken set. Formally, we transform specifications of the form $(\forall x.\, \rho) \land \varphi \land \psi$ where $\LC \equiv \forall x.\, \rho$ into $(\forall x \notin \Br. \rho) \land \varphi \land \psi$\footnote{This is of course more general, as setting $\Br = \emptyset$ degenerates to the original specifications. We need the more general form to state contracts of called methods and loop invariants at a general point in the program where the set of broken objects may not be empty.}. We require the verification engineer to maintain the $\Br$ set accurately by adding objects to it when the program breaks their local conditions and removing them when the engineer's updates to their ghost fields fixes the local conditions on them. This leaves us with triples of the form $\langle\,(\forall x \notin \Br) \land \varphi \land \alpha\,\rangle\; P_{\Gg,\Br}\; \langle\,(\forall x \notin \Br) \land \varphi \land \beta\,\rangle$ where $P_{\Gg,\Br}$ is an augmentation of $P$ with ghost code that updates both $\Gg$ and $\Br$ correctly. 

In Stage 3 (Section~\ref{sec:well-behaved}) we show how to remove the quantification entirely for programs with \emph{well-behaved} manipulations of the broken set, reducing the problem to triples of the form $\langle\,\varphi \land \alpha\,\rangle\; P_{\Gg,\Br}\; \langle\,\varphi \land \beta\,\rangle$. Intuitively, the well-behaved programming paradigm ensures careful handling of the broken set by forcing the verification engineer to add objects to the broken set if their local conditions are potentially broken by statements that change the heap (mutations, allocation, etc.) and only allowing them to remove the objects from the broken set if the engineer can show that local conditions hold on them. In Section~\ref{sec:cs-slist-insert} we introduce a syntactically definable fragment of well-behaved programs that is capable of expressing many data structure manipulation methods (as shown by our case studies and evaluation in Sections~\ref{sec:case-studies} and~\ref{sec:evaluation}).

One of the salient features of monadic maps and local conditions is that they offer simple \emph{frame reasoning}. When the fields of a location $x$ are mutated, the local condition may cease to hold for $x$ as well as neighbors of $x$ (locations accessible from $x$ using pointer-sequences or locations that can access $x$ using pointer-sequences, where these pointer-sequences are mentioned in the local condition). However, we know that local conditions are not affected on other locations further away, and we exploit this implicit frame reasoning to maintain broken sets. 

In Section~\ref{sec:soundness} we prove the soundness of the entire FWYB methodology, namely that if $\langle\,\varphi \land \alpha\,\rangle\; P_{\Gg,\Br}\; \langle\,\varphi \land \beta\,\rangle$ is valid and $P_{\Gg, \Br}$ is well-behaved, then the triple $\langle\,\soexists\, g_1, g_2\ldots,g_k.\, (\LC \land \varphi \land \alpha)\,\rangle\; P\; \langle\,\soexists\, g_1, g_2\ldots,g_k.\, (\LC \land \varphi \land \beta)\,\rangle$ intended by the user is valid. The validity of the former triple can be discharged effectively by decision procedures for combinations of quantifier-free theories, which makes the FWYB methodology entirely automatic given the ghost annotations.

\subsection{Stage 1: Removing Existential Quantification over Monadic Maps using Ghost Code}
\label{sec:ghost-code}

Consider an intrinsic Hoare Triple $\langle\,\soexists\, g_1, g_2\ldots,g_k.\, (\LC \land \varphi \land \alpha)\,\rangle\; P\; \langle\,\soexists\, g_1, g_2\ldots,g_k.\, (\LC \land \varphi \land \beta)\,\rangle$. Read simply, the precondition says that \emph{there exist} maps $\{g_i\}$ satisfying some properties, and the postcondition says that we must \emph{show the existence} of maps $\{g_i\}$ satisfying the post state properties. 

We remove existential quantification from the problem by re-formulating it as follows: we assume that we are \emph{given} the maps $\{g_i\}$ as part of the pre state such that they satisfy $\LC \land \varphi \land \alpha$, and we require the verification engineer to \emph{compute} the $\{g_i\}$ maps in the post state satisfying $\LC \land \varphi \land \beta$. The engineer computes the post state maps by taking the given pre state maps and `repairing' them on an object whenever the program breaks local conditions on that object. %
The repairs are done using \emph{ghost code}, which is a common technique in verification literature~\cite{Jonesghostcode,lucasghostcode,spiritofghostcode,ReynoldsCraftOfProgramming}. Let us consider an example:

\begin{example}[Proof using Ghost Code]
\label{ex:slist-insert-ghostcode}
We use the running example (Example~\ref{ex:slist-insert-triple}) of insertion into a sorted list. Figure~\ref{fig:slist-normal-ghost} shows on the left a snippet where the key $k$ that is to be inserted lies between the keys of $x$ and $\nxt(x)$ (which we assume is not $\nil$). We ignore the conditionals that determine $\nxt(x) \neq \nil$ and $\key(x) \leq k \leq \key(\nxt(x))$ for brevity.

\begin{figure}[htbp]
\begin{subfigure}[t]{0.4\textwidth}
\raggedright
\begin{alltt}
\annotation{pre:} \ensuremath{\soexists\,\sortedll,\rank.\;\LC\land\sortedll(x)}
\annotation{post:} \ensuremath{\soexists\,\sortedll,\rank.\;\LC\land\sortedll(x)}
 y := x.next;
 z := \keyword{new} C();
 z.key := k;
 z.next := y;
 x.next := z;
 \hfill
 \hfill
 \hfill
\end{alltt}
\end{subfigure}
\hspace{2em}
\begin{subfigure}[t]{0.4\textwidth}
\raggedright
\begin{alltt}
\annotation{pre:} \ensuremath{\LC\land\sortedll(x)}
\annotation{post:} \ensuremath{\LC\land\sortedll(x)}
 y := x.next;
 z := \keyword{new} C();
 z.key := k;
 z.next := y;
 \ghost{z.sortedll := True;}
 x.next := z;
 \ghost{z.rank := (x.rank + y.rank)/2;}
 \codecomment{// LC holds for x,y, and z}
\end{alltt}
\end{subfigure}
\caption{Left: Code and specifications written by the user; Right: Augmented program with ghost updates. Specifications do not quantify over monadic maps.}
\label{fig:slist-normal-ghost}
\end{figure}

Let us assume that objects now have fields $\sortedll$ and $\rank$ satisfying the local conditions in the pre state. The local conditions (see Example~\ref{ex:ids-sorted-list}) say that if $\sortedll$ holds on $x$ (which is true in our case), then $\sortedll$ must also hold on $\nxt(x)$, the $\rank$ of $\nxt(x)$ must be smaller than the $\rank$ of $x$, and its key must be larger than the key of $x$. The program violates some of these conditions.%

Observe that the required order on keys is already satisfied since the insertion is correct. The conditions that are violated are: (a) $\sortedll(z)$ must hold since $\sortedll(x)$ holds, (b) the rank of $z$ must larger than that of $y$, and (c) the rank of $x$ must be larger than that of $z$. We fix these violations using ghost updates, shown in Figure~\ref{fig:slist-normal-ghost} on the right (marked in \textcolor{blue}{blue}).

Here the specifications do not contain the second-order existential quantification over the monadic maps. Instead, they are treated as (ghost) fields and manipulated using the usual field update statements.

A formal proof of the sufficiency of the repairs shown above is complex. In general, the mutation \texttt{x.next := z} may break both \texttt{x} and \texttt{z}, similarly each field update has its own set of \emph{impacted objects}, and one must reason about a set of such updates and fixes. Our contributions in Stage 2 and Stage 3 presented in later sections enable the engineer to systematically track broken objects and reason about the sufficiency of repairs effectively.\qed
\end{example}

We note a point of subtlety here: the second triple in the above example eliminates existential quantification over $\Gg$ by claiming something stronger than the original specification, namely that for \emph{any} maps $\sortedll$ and $\rank$ such that $\LC \land \sortedll(x)$ is satisfied in the pre state, there is a \emph{computation} that yields corresponding maps in the post state such that the property is preserved. The onus of coming up with such a computation is placed on the verification engineer.

\medskip
Formally, fix an intrinsically defined data structure $(\Gg, \LC, \varphi)$. We extend the class signature $C = (\Ss, \Ff)$ (and consequently the programming language) to $C_\Gg = (\Ss, \Ff \cup \Gg)$ and treat the symbols in $\Gg$ as \emph{ghost fields} of objects of class $C$ in the program semantics. We do not formalize ghost code here as it is standard; please see prior literature for a formal exposition~\cite{Jonesghostcode,lucasghostcode,spiritofghostcode,ReynoldsCraftOfProgramming}. Intuitively, ghost variables/fields cannot influence the computation of non-ghost variables/fields. In particular, this means that when conditional statements or loops use ghost variables in the condition, the body of the statement must also consist entirely of ghost code. We must also ensure that such `ghost loops' are always terminating since a nonterminating ghost loop changes the meaning of the original program. We do this by requiring the verification engineer to write a ranking function which decreases with each iteration of the ghost loop. We henceforth mean the validity of a Hoare Triple over a program containing such ghost loops with ranking functions to include the condition that the ranking function is shown to be decreasing. Given a program augmented with ghost code as described above, we use below the idea of `projecting out' the ghost code, which refers to the original program with all ghost code simply eliminated. Note that the definition of ghost code ensures that the projection operation is sensible. \amcomment{Madhu, please take a look at this} %

\smallskip
Performing the transformation described in this section %
reduces the verification problem to proving triples of the form $\langle\, \LC \land \varphi \land \alpha \,\rangle\; P_{\Gg}\; \langle\, \LC \land \varphi \land \beta \,\rangle$, where there is no existential quantification over $\Gg$ and $P_\Gg$ is an augmentation of $P$ with ghost code that updates the $\Gg$ maps. The following proposition captures the correctness of this reduction:

\begin{proposition}
\label{prop:ghost-code}
Let $\psi_{\mathit{pre}}$ and $\psi_{\mathit{post}}$ be quantifier-free formulae over $\Ff \cup \Gg$. If $\langle\, \psi_{\mathit{pre}} \,\rangle\; P_{\Gg}\; \langle\, \psi_{\mathit{post}} \,\rangle$ is valid then $\langle\, \exists g_1, g_2\ldots , g_k.\, \psi_{\mathit{pre}} \,\rangle\; P\; \langle\, \exists g_1, g_2\ldots , g_k.\,\psi_{\mathit{post}} \,\rangle$ is valid~\footnote{Here the notion validity for both triples is given by Definition~\ref{defn:ids-triple-validity}, where configurations are interpreted appropriately with or without the ghost fields.}, where $P$ is the projection of $P_\Gg$ obtained by eliminating ghost code. 
\end{proposition}

\madhu{There's no formal definition of ghost code. I think we need to say at least precisely what it is... maybe just code that can read original and ghost variables and write to ghost variables... and terminating...} 

\subsection{Stage 2: Relaxing Universal Quantification using Broken Sets}
\label{sec:broken-sets}

We turn to verifying programs whose pre and post conditions are of the form $\LC \land \gamma$, where $\LC \equiv \forall z.\, \rho(z)$ is the local condition. %
Consider a program $P$ that maintains the data structure. The local conditions are satisfied everywhere in both the pre and post state of $P$. However, they need not hold everywhere in the intermediate states. In particular, $P$ may call a method $N$ which may neither receive nor return a proper data structure. %
To reason about $P$ modularly we must be able to express contracts for methods like $N$. %
To do this we must be able to talk about program states where only some objects may satisfy the local conditions.

\smallskip
\mypara{Broken Sets} We introduce in programs a ghost set variable $\Br$ that represents the set of (potentially) broken objects. Intuitively, at any point in the program the local conditions must always be satisfied on every object that is \emph{not} in the broken set. Formally, for a program $P$ we extend the signature of $P$ with $\Br$ as an additional input and an additional output. We also write pre and post conditions of the form $(\forall z \notin \Br.\, \rho(z)) \land \gamma$ to denote that local conditions are satisfied everywhere outside the broken set, where $\gamma$ can now use $\Br$. In particular, given the Hoare triple 
\begin{center}
$\langle\, (\forall z.\, \rho) \land \alpha\, \rangle\; P_\Gg(\overline{x},\,\mathit{ret}\!:\,\overline{y})\; \langle\, (\forall z.\, \rho) \land \beta\, \rangle$    
\end{center}

\noindent
from Stage 1, we instead prove the following Hoare triple (whose validity implies the validity of the triple above):

\begin{center}
$\langle\, (\forall z \notin \Br.\, \rho) \land \alpha \land \Br = \emptyset\, \rangle\; P_{\Gg,\Br}(\overline{x}, \Br,\,\mathit{ret}\!:\,\overline{y}, \Br)\; \langle\, (\forall z \notin \Br.\, \rho) \land \beta \land \Br = \emptyset\, \rangle$
\end{center}

\noindent
where $\Br$ is a ghost input variable of the type of set of objects and $P_{\Gg,\Br}$ is an augmentation of $P$ with ghost code that computes the $\Gg$ maps as well as the $\Br$ set satisfying the postcondition. %

$P$ may also call other methods $N$ with bodies $Q$. We similarly extend the input and output signatures of the called methods and use the broken set to write appropriate contracts for the methods, introducing triples of the form $\langle\, (\forall z \notin\Br.\, \rho) \land \alpha_N\,\rangle\; Q_\Br(\overline{s},\Br,\, \mathit{ret}\!:\, \overline{r},\Br)\; \langle\,(\forall z \notin\Br.\, \rho) \land \beta_N\,\rangle$. Again, $Q_{\Gg,\Br}$ is an augmentation of $Q$ with ghost code that updates $\Gg$ and $\Br$. %

Note the conjunct $\Br = \emptyset$ in the pre and post conditions for $P$, which make them equivalent to the original specifications. Indeed, for the main method that preserves the data structure property, the broken set is empty at the beginning and end of the program. %
However, called methods or loop invariants can talk about states with nonempty broken sets. %
We require the verification engineer to write ghost code that maintains the broken set accurately. The correctness of this transformation follows from an argument similar to the correctness of Proposition~\ref{prop:ghost-code}.

The above transformation turns the problem of verifying a complex program with multiple auxiliary methods into a modular verification problem where each method has a contract of the same form. We describe how to exploit this uniformity in specifications to achieve decidable reasoning in the next section.

\subsection{Stage 3: Eliminating the Universal Quantifier for Well-Behaved Programs}
\label{sec:well-behaved}

We consider triples of the form 
\begin{center}
$\langle\, (\forall z \notin \Br.\, \rho) \land \alpha\, \rangle\; P_{\Gg,\Br}(\overline{x}, \Br,\,\mathit{ret}\!:\,\overline{y}, \Br)\; \langle\, (\forall z \notin \Br.\, \rho) \land \beta\, \rangle$
\end{center} 

\noindent
where $P_{\Gg,\Br}$ is a program augmented with ghost updates to the $\Gg$-fields as well as the $\Br$ set, and $\alpha,\beta$ are quantifier-free formulae that can also mention the fields in $\Gg$ and the $\Br$ set. In this stage we would like to eliminate the quantified conjunct entirely and instead ask the engineer to prove the validity of the triple 
\begin{center}
$\{\alpha\}\; P_{\Gg,\Br}(\overline{x}, \Br,\,\mathit{ret}\!:\,\overline{y}, \Br)\; \{\beta\}$
\end{center}

However, the above two triples are not, in general, equivalent (as broken sets can be manipulated wildly). In this section we define a syntactic class of \emph{well-behaved} programs that force the verification engineer to maintain broken sets correctly, and for such programs the above triple are indeed equivalent. For example, for a field mutation, well-behaved programs require the engineer to determine the set of \emph{impacted objects} where local conditions may be broken by the mutation. The well-behavedness paradigm then mandates that the engineer add the set of impacted objects to the broken set immediately following the mutation statement. Similarly, well-behaved programs do not allow the engineer to remove an object from the broken set unless they show that the local conditions hold on that object. The imposition of this discipline ensures that programmers carefully preserve the meaning of the broken set (i.e., objects outside the broken set must satisfy local conditions). This allows for the quantified conjunct in the triple given by Stage 2 to be dropped since it always holds for a well-behaved program. Let us look at such a program in the context of our running example:

\begin{example}[Well-Behaved Sorted List Insertion]
\label{ex:slist-insert-wellbehaved}
We first relax the universal quantification as described in Stage 2 (Section~\ref{sec:broken-sets}) and rewrite the pre and post conditions to $\forall\,z\notin\Br.\,\LC(z))\land\sortedll(x)\land\Br=\emptyset$.

We then write the following well-behaved augmentation of the original program. We show the value of the broken set through the program in comments on the right:

\vspace{0.5em}
\hspace{2em}
\begin{minipage}[t]{0.3\textwidth}
\begin{alltt} 
\annotation{pre:} \ensuremath{\sortedll(x)\land\Br=\emptyset}
\annotation{post:} \ensuremath{\sortedll(x)\land\Br=\emptyset}
 \ghost{assert x \ensuremath{\notin} Br;}
 \autobr{assume LC(x);}
 y := x.next;    \codecomment{// \{\}}
 z := \keyword{new} C();
 \autobr{Br := Br \ensuremath{\cup} \{z\};} \codecomment{// \{z\}}
 z.key := k;
 \autobr{Br := Br \ensuremath{\cup} \{z\};} \codecomment{// \{z\}}
 z.next := y;
 \autobr{Br := Br \ensuremath{\cup} \{z\};} \codecomment{// \{z\}}
 \hfill
\end{alltt}
\end{minipage}
\hspace{5em}
\begin{minipage}[t]{0.6\textwidth}
\begin{alltt}
 \ghost{z.sortedll := True;}
 \autobr{Br := Br \ensuremath{\cup} \{z\};} \codecomment{// \{z\}}
 x.next := z;
 \autobr{Br := Br \ensuremath{\cup} \{x\};} \codecomment{// \{x,z\}}
 \ghost{z.rank := (x.rank + y.rank)/2;}
 \autobr{Br := Br \ensuremath{\cup} \{z\};} \codecomment{// \{x,z\}}
 \codecomment{// x and z satisfy LC}
 \autobr{assert LC(z);}
 \ghost{Br := Br \ensuremath{\setminus} \{z\};} \codecomment{// \{x\}}
 \autobr{assert LC(x);}
 \ghost{Br := Br \ensuremath{\setminus} \{x\};} \codecomment{// \{\}}
\end{alltt}
\end{minipage}

\smallskip
We depict the statements that are enforced by the well-behavedness paradigm in \textcolor{nicepink}{pink} and the ghost updates written by the verification engineer in \textcolor{blue}{blue}. As we can see, the paradigm adds the set of impacted objects to the broken set after each mutation and allocation. Determining the impact set of a mutation is nontrivial; we show how to construct them in Section~\ref{sec:cs-slist-insert}. We also see that if the engineer wants to remove $x$ from the broken set then they are required to show that $\LC(x)$ holds (assert followed by removal from $\Br$). Finally, as seen at the beginning of the program, if we show $x$ does not belong to $\Br$ then well-behavedness allows us to infer that $\LC(x)$ holds. This is possible because well-behaved programs always maintain the broken set correctly. We leave the formal correctness of the above program with the quantifier-free specifications to the reader. 

\smallskip
\noindent
\textbf{\textit{Putting it All Together.}} Let us call the above program $P_{\Gg,\Br}$. Since it is well-behaved and satisfies the contract $\langle\,\sortedll(x)\land\Br=\emptyset\,\rangle\; P_{\Gg,\Br}\; \langle\,\sortedll(x)\land\Br=\emptyset\,\rangle$, we can conclude that it satisfies the contract $\langle\,(\forall z \notin \Br.\,\rho) \land\sortedll(x)\land\Br=\emptyset\,\rangle\; P_{\Gg,\Br}\; \langle\,(\forall z \notin \Br.\,\rho) \land\sortedll(x)\land\Br=\emptyset\,\rangle$. By Proposition~\ref{prop:ghost-code}, this in turn means that the triple with the user's original program given in Example~\ref{ex:slist-insert-ghostcode} with existential quantification over monadic maps and universal quantification over objects in the specifications is valid! Therefore, via FWYB we can reason about the correctness of programs with respect to intrinsic specifications by checking the correctness of (augmented) programs with respect to the quantifier-free specifications, which can be discharged efficiently in practice using SMT solvers which support decision procedures for combinations of theories~\cite{Z3,cvc4}.\qed

\end{example}

\smallskip
\noindent
We now develop the general theory of well-behaved programs.

\subsection*{Rules for Constructing Well-Behaved Programs}

We define the class of well-behaved programs using a set of rules. We first introduce some notation. 

We distinguish the triples consisting of augmented programs and quantifier-free annotations by $\{\psi_{\mathit{pre}}\}\; P\; \{\psi_{\mathit{post}}\}$, with $\{\}$ brackets rather than $\langle\,\rangle$. %
We denote that such a triple is provable by $\vdash\{\psi_{\mathit{pre}}\}\; P\; \{\psi_{\mathit{post}}\}$. Our theory is agnostic to the mechanism for proving such Hoare triples correct (and in evaluations, we use the off-the-shelf verification tool {\sc Dafny}). %
However, we assume that the underlying proof mechanism is sound with respect to the operational semantics, checking in particular that dereferences are memory safe. We denote that a code snippet $P$ is well-behaved by $\wb P$. We also denote that local conditions hold on an object $x$ by $\LC(x)$ for clarity (rather than the technically correct $\rho(x)$ where $\rho$ is the matrix of $\LC$).

\smallskip
The rules for constructing well-behaved programs are given in Figure~\ref{fig:well-behaved}. %

Skip, assignment, lookup, and return statements are vacuously well-behaved. They do not change the heap and therefore preserve the local condition on any object where they already held.

\textsc{Mutation}.\hspace{0.25em}  Mutation statements potentially break local conditions and must therefore grow the broken set. Let $A$ be a finite set of location terms over $x$ such that for any $z \notin A$, if $\LC(z)$ held before the mutation, then it continues to hold after the mutation. We refer to such a set $A$ as an \emph{impact set} for the mutation, and we update the broken set after a mutation with its impact set (see mutation statements in Example~\ref{ex:slist-insert-ghostcode} for instances of this rule). %
Of course, the impact set may not always be expressible as a finite set of terms, but this is indeed the case for all intrinsic definitions of data structures that we define in this paper. We show how to construct provably correct impact sets for field mutations in practice in Section~\ref{sec:cs-slist-insert}.

\textsc{Allocation}.\hspace{0.25em} An allocation statement does not modify the heap on any existing object. Therefore, we simply update the broken set by adding the newly created object $x$ (this was also the case in Example~\ref{ex:slist-insert-wellbehaved}).

\textsc{Function Call}.\hspace{0.25em} We state axiomatically that function call statements are well-behaved, leaving it to the verification engineer to capture the contract of the called function accurately, including its effect on $\Br$.

\textsc{Assert LC and Remove}.\hspace{0.25em} The rules described above grow the broken set. We also want to shrink the broken set when the verification engineer fixes the local conditions on some broken locations. The \textsc{Assert LC and Remove} rule enables this. This rule is used in Example~\ref{ex:slist-insert-wellbehaved} as the sequence of statements \texttt{\autobr{assert\,LC(x);\,}\ghost{Br\,:=\,Br\ensuremath{\setminus}\{x\}}}. Informally, the verification engineer is required to show that $\LC(x)$ holds before removing $x$ from $\Br$.

\textsc{Infer LC Outside Br}.\hspace{0.25em} We also add a rule for instantiating $\LC$ on objects outside the broken set. Since our methodology makes the conjunct $\forall x \notin \Br.\, \rho(x)$ implicit, we lose some proving power. For example, if a program modifies a location $x$ in a red black tree, we may want to update the values of the ghost monadic maps on $x$ depending on whether its children are red or black, and using the information given by the local constraints on its children (since the children are not broken and therefore satisfy the local constraints). The \textsc{Infer LC Outside Br} rule enables this reasoning mechanism. We employ this rule in Example~\ref{ex:slist-insert-wellbehaved} using the sequence of statements \texttt{\ghost{assert\,x\ensuremath{\notin}Br;\,}\autobr{assume\,LC(x)}}.

Finally, compositions, conditionals, and loops involving well-formed subprograms are well-formed. Note that $\Br$ is a ghost variable and therefore cannot appear in the conditions for conditional statements or loops.

In the above presentation we use only one broken set for simplicity of exposition. However, our general framework allows for finer-grained broken sets that can track breaks over a partition on the local conditions. For example, Section~\ref{sec:cs-overlaid} illustrates verifying deletion in an overlaid data structure (consisting of a linked list and a binary search tree) using two broken sets: one each corresponding to the local conditions of the two component data structures.

\begin{figure}
\begin{mathpar}
\inferrule[Skip/Assignment/Lookup/Return]{ \\ }{\wb s \textrm{   where $s$ is of the form} \\\\ \textrm{$\mathsf{skip}$,~\texttt{x:=y},~  \texttt{x:=y.f},~ or~ $\mathsf{return}$}} 
\and
\inferrule[Mutation]{ \\\\ \vdash \{\, z \notin A \land LC(z) \land x \neq \nil \,\}\; x.f := v\; \{\, LC(z) \,\} }{ \wb\; x.f := v\,;\, \Br := \Br \cup A \\\\ \textrm{where $A$ is a finite set of location terms over $x$} }
\and
\inferrule[Allocation]{ \\ }{\wb\; x := \mathsf{new}\, C()\,;\, \Br := \Br \cup \{x\}}
\and
\inferrule[Function Call]{ \\ }{\wb\; \overline{y}, \Br := \mathit{Function}(\overline{x},\Br)}
\and
\inferrule[Infer LC Outside Br]{ \\ }{ \wb\; \mathsf{if}\; (x \neq \nil \land x \notin \Br)\; \mathsf{then}\; \mathsf{assume}\,LC(x)%
}
\and
\inferrule[Assert LC and Remove]{ \\ }{ \wb\; \mathsf{if}\; LC(x)\; \mathsf{then}\; \Br := \Br \setminus \{x\}%
}
\and
\inferrule[Composition]{ \\\\ \wb\; P \\ \wb Q }{ \wb\; P\,;\, Q }
\and
\inferrule[If-Then-Else]{ \\\\ \wb\; P \\ \wb Q }{ \wb\; \mathsf{if}\; \mathit{cond}\; P\; \mathsf{else}\; Q \\\\ \textrm{where $\mathit{cond}$ does not mention $\Br$} }
\and
\inferrule[While]{ \\\\ \wb\; P }{ \wb\; \mathsf{while}\; \mathit{cond}\; \mathsf{do}\; P \\\\ \textrm{where $\mathit{cond}$ does not mention $\Br$} }
\end{mathpar}
\caption{Rules for constructing well-behaved programs. Local condition formula instantiated at $x$ is denoted by $\LC(x)$. The statement $(\mathsf{if}\, \mathit{cond}\, \mathsf{then}\, S)$ is sugar for $(\mathsf{if}\, \mathit{cond}\, \mathsf{then}\, S\, \mathsf{else}\, \mathsf{skip})$.}
\label{fig:well-behaved}
\end{figure}

\subsection{Soundness of FWYB}
\label{sec:soundness}

In this section we state the soundness of the FWYB methodology. We first define some terminology.

Fix a main method $M$ with input variables $\overline{x}, \Br$ and output variables $\overline{y},\Br$. Let the body of $M$ be $P$. Let $N_i, 1 \leq i \leq k$ be a set of auxiliary methods that $P$ calls with bodies $Q_i$, where the methods $N_i$ also have similar signatures with $\Br$ as the last input parameter and last output parameter. Note that the bodies $P$ and $Q_i$ contain updates of ghost fields. Let us denote by a program the collection of methods $[(M: P); (N_1: Q_1)\ldots (N_k: Q_k)]$. We define the projection of $M$ to a user-level program:

\begin{definition}[Projection of Augmented Code to User Code]
\label{defn:code-projection}
The projection of the augmented program $[(M:P); (N_1:Q_1)\ldots (N_k:Q_k)]$ is the user-level program $[(M':P'); (N_1':Q_1')\ldots (N_k':Q_k')]$ such that:
\begin{enumerate}
    \item If $M$ has signature $(\overline{x},\Br,\, \mathit{ret}\!:\overline{y},\Br)$, then $M'$ has signature $(\overline{x},\,\mathit{ret}\!:\overline{y})$. Similarly for each $N_i, 1 \leq i \leq k$, the corresponding method $N_i'$ removes $\Br$ from the signature.
    \item $P'$ is derived from $P$ by: (a) eliminating all ghost code, and (b) replacing each function call statement of the form $\overline{r},\Br := N_j(\overline{z},\Br)$ with the statement $\overline{r} := N_j(\overline{z})$. Similarly each $Q_i'$ is derived from the corresponding $Q_i$ by a similar transformation.
\end{enumerate}
\end{definition}

We now state the soundness theorem.

\begin{theorem}[FWYB Soundness]
\label{thm:soundness}
Let $(\Gg, \LC, \varphi)$ be an intrinsic definition with $\Gg = \{g_1,g_2\ldots, g_l\}$. Let $[(M:P); (N_1:Q_1)\ldots, (N_k:Q_k)]$ be an augmented program constructed using the FWYB methodology such that $\,\wb P$ and $\,\wb Q_i,\, 1\leq i \leq k$, i.e., the programs $P$ and $Q_i$ are well-behaved (according to the rules in Figure~\ref{fig:well-behaved}). Let $\psi_{\mathit{pre}}$ and $\psi_{\mathit{post}}$ be quantifier-free formulae that do not mention $\Br$ (but can mention the maps in $\Gg$). Finally, let $[(M':P'); (N_1':Q_1')\ldots, (N_k':Q_k')]$ be the projected user-level program according to Definition~\ref{defn:code-projection}. Then:

\smallskip
\noindent
If the triple

\begin{center}
$\{\varphi \land \psi_{\mathit{pre}} \land \Br = \emptyset \}\; P\; \{\varphi \land \psi_{\mathit{post}} \land \Br = \emptyset\}$
\end{center}

\noindent
is valid, then the triple

\begin{center}
$\langle\,\soexists\, g_1,g_2\ldots, g_l.\, (\LC \land \varphi \land \psi_{\mathit{pre}}) \,\rangle\; P'\; \langle\,\soexists\, g_1,g_2\ldots, g_l.\, (\LC \land \varphi \land \psi_{\mathit{post}}) \,\rangle$
\end{center}

\noindent
is valid (according to Definition~\ref{defn:ids-triple-validity}).\qed
\end{theorem}

Informally, the soundness theorem says that given a user-written program, if we augment it with updates to ghost fields and the broken set only using the discipline for well-behaved programs and show that if the broken set is empty at the beginning of the program it will be empty at the end, then the original user-written program, with contracts on preservation of the data structure, is correct. We provide a proof of the theorem in Appendix~\ref{app:soundness}.

\subsection{Generating Quantifier-Free Verification Conditions}
\label{app:qfree-vcs}

The FWYB methodology described in previous sections shows that we can soundly reduce the problem of verifying programs with intrinsic specifications to the problem of verifying programs (with ghost code) with quantifier-free contracts. We then argue that we can reason with the latter using combinations of various quantifier-free theories including sets and maps with pointwise updates. In this section we detail some subtleties involved in the argument.

It is well-known that we can reason with scalar programs with quantifier-free contracts by generating quantifier-free verification conditions (which in turn can be handled by SMT solvers). However, this is not immediately clear for programs that dynamically manipulate heaps. In particular, commands such as allocation and function calls pose challenges in formulating quantifier-free verification conditions.

At a high level, our solution transforms the given heap program into a scalar program that explicitly encodes changes to the heap. Specifically, we show an encoding for the field mutation, allocation, and function call statements. 

\mypara{Modeling Field Mutation} As described earlier, we model the monadic maps and fields as updatable maps~\cite{pointwisearrays}. Formally, we introduce a map $M_f$ (also called an \emph{array} in SMT solvers like Z3~\cite{Z3}) for every field/monadic map $f$. We then encode the commands for field lookup and mutation as map operations. For example, the mutation $x.f \coloneqq y$ is encoded as $M_f[x] := y$.

\mypara{Modeling Allocation} 
We model programs in a safe garbage-collected programming language. %
We introduce a ghost global variable $\Alloc$ to model the allocated set of objects in the program. We then add several assumptions (i.e., $\mathsf{assume}$ statements) throughout the program. Specifically, we assume for every program parameter of type Object, the parameter itself as well as the values of the monadic maps of type Object/Set-of-Objects on the parameter are all contained in $\Alloc$. For example, in the case of our running example (Example~\ref{ex:slist-insert-wellbehaved}), we add the assumptions $x \in \Alloc$ and $\nxt(x) \neq \nil \Rightarrow \nxt(x) \in \Alloc$. If we had a monadic map $\hslist$ corresponding to the heaplet of the sorted list, we would also add the assumption $\hslist(x) \subseteq \Alloc$. Similarly, whenever an object is dereferenced on a field of type Object/Set-of-Objects in the program, we add an assumption that the resulting value is contained in $\Alloc$. Note that these are quantifier-free assumptions. They can be added soundly since they are valid under the semantics of the underlying language.

We then model allocation by introducing a new object to $\Alloc$ and ensure that the default values of the various fields on the newly allocated object belong to $\Alloc$. These constraints can be expressed using a quantifier-free formula over maps.

\mypara{Modeling Heap Change Across Function Calls} The main challenge in modeling function calls is to ensure the ability to do frame reasoning. To do this, we extend the programming language with a \emph{modified set} annotation for methods. We require the modified set to be a term of type Set-of-Objects that is constructed using object variables in the current scope and monadic maps over them. In the case of our running example (Example~\ref{ex:slist-insert-triple}), we would add a monadic map $\hslist$ of type Set-of-Objects corresponding to the heaplet of the sorted list and annotate the program with $\hslist(x)$ as the modified set. Figure~\ref{fig:slist-insert-code} shows the full version of sorted list insertion with the modified set annotation.

Given a modified set $\Mod$, we model changes to the heap across a function call by introducing new maps corresponding to the various fields (including monadic maps) after the call. We then add assumptions that the values of the new maps are equal to the values of the maps before the call on all locations that do not belong to the modified set $\Mod$. Although this constrains the maps on unboundedly many objects, it can be written without quantifiers by using pointwise operators on maps~\cite{pointwisearrays}. %
Formally, for a field $f$ modeled as a map $M_f$, we introduce a new map $M_f'$ and update $M_f$ as:
\begin{center}
$M_f[x] := \ite(x \in \Mod,\, M_f'[x],\, M_f[x])$
\end{center}
The above update can be expressed using pointwise operators as $M_f \coloneqq \ite(\Mod,M_f',M_f)$, where the $\ite$ operator is applied pointwise over the maps $\Mod$, $M_f$, and $M_f'$. The value of the field $f$ on an object $x$ after the call will then be equal to $x.f$ before the call if $x$ was not modified, and a \emph{havoc}-ed value given by $M_f'$ otherwise. Pointwise operators are supported by the generalized array theory~\cite{pointwisearrays} whose quantifier-free fragment is decidable.

\smallskip
Program verifiers like Boogie~\cite{boogie3} offer VC generation frameworks that are amenable to the modeling described in this section. Indeed, our implementation of the IDS/FWYB methodology described in Section~\ref{sec:implementation} uses Boogie.

\else
\section{Fix What You Break (FWYB) Verification Methodology}
\label{sec:proglogic}

In this section we present the second main contribution of this paper: the Fix-What-You-Break (FWYB) methodology. %
We begin by describing a while programming language and defining the verification problem we study. We fix a class $C = (\Ss, \Ff)$ throughout this section.

\subsection{Programs, Contracts, and Correctness}
\label{sec:triples}

\mypara{Programs} Figure~\ref{fig:prog-lang} shows the programming language used in this work. Note that we can use variables and expressions over non-object sorts. Functions can return multiple outputs. We assume that method signatures contain designated output variables and therefore the return statement does not mention values. 

Our language is safe (i.e., allocated locations cannot point to un-allocated locations) and garbage-collected. Formally we consider configurations $\theta$ consisting of a store (map from variables to values) and a heap along with an error state $\bot$ to model error on a null dereference. We denote that a formula $\alpha$ is satisfied on a configuration $\theta$ by writing $\theta \models \alpha$.

\begin{figure}
\begin{align*}
P \coloneqq &\;\; x\, :=\, \nil\;\; |\;\; x\, :=\, y\;\; |\;\; v\, :=\, be \;\;|\;\; y\, :=\, x.f\;\; |\;\; v\, :=\, x.d \\[-3pt]
&\;\; |\;\; x.f\, :=\, y\;\;|\;\; x.d\, :=\, v \;\;|\; \;x\, :=\, \mathsf{new}\; C() \;\;|\; \;\overline{r}\, :=\, \mathit{Function}(\overline{t}) \\[-3pt]
&\;\; |\;\; \mathsf{skip}\;\;|\;\; \mathsf{assume}\;\mathit{cond}\;\;|\;\; \mathsf{return}\;\;|\;\; P\, ; \, P \;\;|\;\; \mathsf{if}\; \mathit{cond}\; \mathsf{then}\; P\; \mathsf{else}\; P\;\; |\;\; \mathsf{while}\; \mathsf{cond}\; \mathsf{do}\; P\; \\[-1pt]
\mathit{cond} \coloneqq &\;\; x = y\;\;|\;\; x \neq y \;\;|\;\; \mathit{be} \;\;\;\;\textrm{(Condition Expressions)}
\end{align*}
\vspace{-2em}
\caption{Grammar of while programs with recursion. $x,y$ are variables denoting objects of class $C?$ (i.e., $C$ objects or $\nil$), $v, w$ are a background sort(s) variables, $r, t$ denote variables of any sort, $f$ is a pointer field, $d$ is a data field, and $be$ is a expression of the background sort(s).}
\label{fig:prog-lang}
\end{figure}

\mypara{Intrinsic Hoare Triples} The verification problem we study in this paper is \emph{maintenance} of data structure properties. Fix an intrinsic definition $(\Gg, \LC, \varphi(\overline{y}))$ where $\Gg = \{g_1, g_2\ldots, g_k\}$. Let $\overline{z}$ be the input/output variables for a program that we want to verify. We consider pre and post conditions of the form 
\begin{center}
$\soexists\, g_1, g_2\ldots,g_k.\, (\LC \land \varphi(\overline{w}) \land \psi(\overline{z}))$
\end{center}

\noindent
where each $g_i$ is a ghost monadic map (unary function over locations), $\psi$ is a quantifier-free formula over $\overline{z}$ that can use the ghost monadic maps $g_i$, and $\overline{w}$ is a tuple of variables from $\overline{z}$ whose arity is equal to $\overline{y}$. Note that the above has a second-order existential quantification ($\soexists$) over function symbols $g_1, \ldots, g_k$, and $\LC$ has first-order universal quantification over a single location variable. Read in plain English, ``$\,\overline{w}$ points to a data structure $\mathit{IDS}$ such that the (quantifier-free) property $\psi(\overline{z})$ holds''. %

\noindent
We study the validity of the following Hoare Triples:

\begin{center}
$\langle\,\alpha(\overline{x})\,\rangle\;\; \mathrm{P}(\overline{x}, \,\mathit{ret}\!:\, \overline{r})\;\; \langle\,\beta(\overline{x}, \overline{r})\,\rangle$
\end{center}

\noindent
where $\alpha$ and $\beta$ are pre and post conditions of the above form, $\mathrm{P}$ is a program, and $\overline{x}, \overline{r}$ are input and output variables for $\mathrm{P}$ respectively.

\begin{example}[Running Example: Insertion into a Sorted List]
\label{ex:slist-insert-triple}
Let $\mathit{SortedLL}(y) = (\Gg, \LC, \mathit{sorted}(y))$ as in Example~\ref{ex:ids-sorted-list} where $\Gg = \{\sortedll,\rank\}$. The following Hoare triple says that insertion into a sorted list returns a sorted list:
\begin{center}
$\langle\,\soexists\, \sortedll,\rank.\, \LC \land \sortedll(x) \,\rangle\; \mathit{sorted\!\!-\!\! insert}(x,k,\, \mathit{ret}\!: x)\; \langle\,\soexists\, \sortedll,\rank.\, \LC \land \sortedll(x)\,\rangle$
\end{center}

\noindent
where $x,r$ are variables of type $C$, $k$ is of type $\mathit{Int}$ and $\mathit{sorted\!\!-\!\! insert}$ is the usual recursive method. %

\end{example}

\mypara{Validity of Intrinsic Hoare Triples} We now define the validity of Hoare Triples.

\begin{definition}[Validity of Intrinsic Hoare Triples]
\label{defn:ids-triple-validity}
An intrinsic triple $\langle\, \alpha\, \rangle\, P\, \langle\, \beta\, \rangle$ is \emph{valid} if for every configuration $\theta$ such that $\theta \models \alpha$, transitioning according to $P$ starting from $\theta$ does not encounter the error state $\bot$, %
and furthermore, if $\theta$ transitions to $\theta'$ under $P$, then $\theta' \models \beta$.
\end{definition}

\subsection{Ghost Code}
\label{sec:ghost-code-defn}

In this work we consider the augmentation of procedures with \emph{ghost} or non-executed code. Ghost code involves the manipulation of a set of distinct \emph{ghost variables} and \emph{ghost fields}, distinguished from regular or `user' variables and fields. In program verification, ghost code provides a programmatic way of constructing values/functions that witness a particular property.

We defer a formal definition of ghost code to the Appendix of our supplementary material\footnote{Our supplementary material is available in either our technical report \cite{idstechreport} or the permanent DOI record at \url{https://doi.org/10.1145/3656450}.} %
and only provide intuition here. Intuitively, ghost variables/fields cannot influence the computation of non-ghost variables/fields. Therefore, ghost variables and maps can be assigned values from user variables and maps, but the reverse is not allowed. Similarly, when conditional statements or loops use ghost variables in the condition, the body of the statement must also consist entirely of ghost code. Simply, ghost code cannot control the flow of the user program. These conditions can be checked statically. Finally, we also require that ghost loops and functions always terminate since nonterminating ghost code can change the meaning of the original program. Our definition is agnostic to the technique used to establish termination, however, we use ranking functions to establish termination in our implementation in \Dafny. 

We formalize the above into a grammar that extends the original programming language in Figure~\ref{fig:prog-lang} into a ghost code-augmented language in Figure~\ref{fig:lang-with-ghost} in Appendix~\ref{app:ghost-code-defn} of our supplementary material. The language of ghost programs is similar to $P$ in Figure~\ref{fig:prog-lang}, except that we do not have allocation or assume statements, and loops/functions must always terminate. See prior literature for a more detailed formal treatment of ghost code~\cite{Jonesghostcode,lucasghostcode,spiritofghostcode,ReynoldsCraftOfProgramming}. 

\mypara{Projection that Eliminates Ghost Code} We can define the notion of `projecting out' ghost code, which takes a program that contains ghost code and yields a pure user program with all ghost code simply eliminated. Intuitively, the fact that ghost code does not affect the execution of the underlying user program makes the projection operation sensible.

Fix a main method $M$ with body $P$. Let $N_i, 1 \leq i \leq k$ be a set of auxiliary methods with bodies $Q_i$ that $P$ can call. Note that the bodies $P$ and $Q_i$ contain ghost code. Let us denote a program containing these methods by $[(M: P); (N_1: Q_1)\ldots (N_k: Q_k)]$. We then define projection as follows:

\begin{definition}[Projection of Ghost-Augmented Code to User Code]
\label{defn:code-projection}
The projection of the ghost-augmented program $[(M:P); (N_1:Q_1)\ldots (N_k:Q_k)]$ is the user program $[(\hat{M}:\hat{P}); (\hat{N_1}:\hat{Q_1})\ldots (\hat{N_k}:\hat{Q_k})]$ such that:
\begin{enumerate}
    \item The input (resp. output) signature of $\hat{M}$ is that of $M$ with the ghost input (resp. output) parameters removed. 
    \item $\hat{P}$ is derived from $P$ by: (a) eliminating all ghost code, and (b) replacing each non-ghost function call statement of the form $\overline{r} := N_j(\overline{t})$ with the statement $\overline{s} := \hat{N_j}(\overline{u})$, where $\overline{u}$ is the non-contiguous subsequence of $\overline{t}$ with the elements corresponding to ghost input parameters removed and $\overline{s}$ is obtained from $\overline{r}$ similarly. Each $\hat{Q_i}$ is derived from the corresponding $Q_i$ by a similar transformation.
\end{enumerate}
\end{definition}

We provide an expanded version of this definition in our supplementary material in Appendix~\ref{app:ghost-code-defn}.

\smallskip
\subsection*{An Overview of FWYB}

We develop the Fix-What-You-Break (FWYB) methodology in three stages, in the following subsections. We give here an overview of the methodology and the stages.

Recall that intrinsic triples are of the form $\langle\,\soexists\, g_1, g_2\ldots,g_k.\, (\LC \land \varphi \land \alpha)\,\rangle\; P\; \langle\,\soexists\, g_1, g_2\ldots,g_k.\, (\LC \land \varphi \land \beta)\,\rangle$. In Stage 1 (Section~\ref{sec:ghost-code}) we remove the second-order quantification. We do this by requiring the verification engineer to explicitly construct the $g_i$ maps in the post state from the maps in the pre state using \emph{ghost code}. We then obtain triples of the form $\langle\,\LC \land \varphi \land \alpha\,\rangle\; P_\Gg\; \langle\,\LC \land \varphi \land \beta\,\rangle$ where $P_\Gg$ is an augmentation of $P$ with ghost code that updates the $\Gg$ maps.

Note that the $\LC$ in the contract universally quantifies over objects. In Stages 2 (Section~\ref{sec:broken-sets}) and 3 (Section~\ref{sec:well-behaved}) we remove the quantification by explicitly tracking the objects where the local conditions do not hold and treating them as implicitly true on all other objects. We call this set $\Br$ the \emph{broken set}. Intuitively, the broken set grows when the program mutates pointers or makes other changes to the heap, and shrinks when the verification engineer repairs the $\Gg$ maps using ghost code to satisfy the $\LC$ on the broken objects. The specifications assume an empty broken set at the beginning of the program and the engineer must ensure that it is empty again at the end of the program. However, they do not have to track the objects manually. We develop in Stage 3 (Section~\ref{sec:well-behaved}) a discipline for writing only \emph{well-behaved} manipulations of the broken set. This reduces the problem to triples of the form $\langle\,\varphi \land \alpha\,\rangle\; P_{\Gg,\Br}\; \langle\,\varphi \land \beta\,\rangle$, where $P_{\Gg,\Br}$ contains ghost code for updating both $\Gg$ and $\Br$. Note that these specifications are quantifier-free, and checking them can be effectively automated using SMT solvers~\cite{Z3,cvc4}.

\subsection{Stage 1: Removing Existential Quantification over Monadic Maps using Ghost Code}
\label{sec:ghost-code}

Consider an intrinsic Hoare Triple $\langle\,\soexists\, g_1, g_2\ldots,g_k.\, (\LC \land \varphi \land \alpha)\,\rangle\; P\; \langle\,\soexists\, g_1, g_2\ldots,g_k.\, (\LC \land \varphi \land \beta)\,\rangle$. Read simply, the precondition says that \emph{there exist} maps $\{g_i\}$ satisfying some properties, and the postcondition says that we must \emph{show the existence} of maps $\{g_i\}$ satisfying the post state properties. 

We remove existential quantification from the problem by re-formulating it as follows: we assume that we are \emph{given} the maps $\{g_i\}$ as part of the pre state such that they satisfy $\LC \land \varphi \land \alpha$, and we require the verification engineer to \emph{compute} the $\{g_i\}$ maps in the post state satisfying $\LC \land \varphi \land \beta$. The engineer computes the post state maps by taking the given pre state maps and `repairing' them on an object whenever the program breaks local conditions on that object. %
The repairs are done using ghost code, which is a common technique in verification literature~\cite{Jonesghostcode,lucasghostcode,spiritofghostcode,ReynoldsCraftOfProgramming}. 

\medskip
Formally, fix an intrinsically defined data structure $(\Gg, \LC, \varphi)$. We extend the class signature $C = (\Ss, \Ff)$ (and consequently the programming language) to $C_\Gg = (\Ss, \Ff \cup \Gg)$ and treat the symbols in $\Gg$ as \emph{ghost fields} of objects of class $C$ in the program semantics. Performing the transformation described above %
reduces the verification problem to proving triples of the form $\langle\, \LC \land \varphi \land \alpha \,\rangle\; P_{\Gg}\; \langle\, \LC \land \varphi \land \beta \,\rangle$, where there is no existential quantification over $\Gg$ and $P_\Gg$ is an augmentation of $P$ with ghost code that updates the $\Gg$ maps. The following proposition captures the correctness of this reduction:

\begin{proposition}
\label{prop:ghost-code}
Let $\psi_{\mathit{pre}}$ and $\psi_{\mathit{post}}$ be quantifier-free formulae over $\Ff \cup \Gg$. If $\langle\, \LC \land \psi_{\mathit{pre}} \,\rangle\; P_{\Gg}\; \langle\, \LC \land \psi_{\mathit{post}} \,\rangle$ is valid then $\langle\, \exists g_1, g_2\ldots , g_k.\, \LC \land \psi_{\mathit{pre}} \,\rangle\; P\; \langle\, \exists g_1, g_2\ldots , g_k.\,\LC \land \psi_{\mathit{post}} \,\rangle$ is valid~\footnote{Here the notion of validity for both triples is given by Definition~\ref{defn:ids-triple-validity}, where configurations are interpreted appropriately with or without the ghost fields.}, where $P$ is the projection of $P_\Gg$ obtained by eliminating ghost code. 
\end{proposition}

\begin{proof}[Proof Gist]
The full proof is in our supplementary material. The first Hoare triple shows that if we are given \emph{any} maps $g_i$ (implicitly encoded as values of ghost fields) that satisfy $LC$ in the pre-state, then the program with ghost code computes a modified version of these maps such that the $LC$ is holds in the post-state. Surely then, if there was a set of maps $g_i$ that satisfied $LC$ in the pre-state, there will exists a set of maps $g_i'$ that satisfy $LC$ in the post-state.
\end{proof}

\noindent
We note a point of subtlety about the reduction in this stage here: the simplified triple eliminates existential quantification over $\Gg$ by claiming something stronger than the original specification, namely that for \emph{any} maps $\{g_i\}$ such that $\psi_{\mathit{pre}}$ is satisfied in the pre state, there is a \emph{computation} that yields corresponding maps in the post state such that $\psi_{\mathit{post}}$ holds. The onus of coming up with such a computation is placed on the verification engineer.

\subsection{Stage 2: Relaxing Universal Quantification using Broken Sets}
\label{sec:broken-sets}

We turn to verifying programs whose pre and post conditions are of the form $\LC \land \gamma$, where $\LC \equiv \forall z.\, \rho(z)$ is the local condition. %
Consider a program $P$ that maintains the data structure. The local conditions are satisfied everywhere in both the pre and post state of $P$. However, they need not hold everywhere in the intermediate states. In particular, $P$ may call a method $N$ which may neither receive nor return a proper data structure. %
To reason about $P$ modularly we must be able to express contracts for methods like $N$. %
To do this we must be able to talk about program states where only some objects may satisfy the local conditions.

\smallskip
\mypara{Broken Sets} We introduce in programs a ghost set variable $\Br$ that represents the set of (potentially) broken objects. Intuitively, at any point in the program the local conditions must always be satisfied on every object that is \emph{not} in the broken set. Formally, for a program $P$ we extend the signature of $P$ with $\Br$ as an additional input and an additional output. We also write pre and post conditions of the form $(\forall z \notin \Br.\, \rho(z)) \land \gamma$ to denote that local conditions are satisfied everywhere outside the broken set, where $\gamma$ can now use $\Br$. In particular, given the Hoare triple 
\begin{center}
$\langle\, (\forall z.\, \rho(z)) \land \alpha\, \rangle\; P_\Gg(\overline{x},\,\mathit{ret}\!:\,\overline{y})\; \langle\, (\forall z.\, \rho(z)) \land \beta\, \rangle$
\end{center}

\noindent
from Stage 1, we instead prove the following Hoare triple (whose validity implies the validity of the triple above):

\begin{center}
$\langle\, (\forall z \notin \Br.\, \rho(z)) \land \alpha \land \Br = \emptyset\, \rangle\; P_{\Gg,\Br}(\overline{x}, \Br,\,\mathit{ret}\!:\,\overline{y}, \Br)\; \langle\, (\forall z \notin \Br.\, \rho(z)) \land \beta \land \Br = \emptyset\, \rangle$
\end{center}

\noindent
where $\Br$ is a ghost input variable of the type of set of objects and $P_{\Gg,\Br}$ is an augmentation of $P$ with ghost code that computes the $\Gg$ maps as well as the $\Br$ set satisfying the postcondition. %

$P$ may also call other methods $N$ with bodies $Q$. We similarly extend the input and output signatures of the called methods and use the broken set to write appropriate contracts for the methods, introducing triples of the form $\langle\, (\forall z \notin\Br.\, \rho(z)) \land \alpha_N\,\rangle\; Q_\Br(\overline{s},\Br,\, \mathit{ret}\!:\, \overline{r},\Br)\; \langle\,(\forall z \notin\Br.\, \rho(z)) \land \beta_N\,\rangle$. Again, $Q_{\Gg,\Br}$ is an augmentation of $Q$ with ghost code that updates $\Gg$ and $\Br$. %

For the main method that preserves the data structure property, the broken set is empty at the beginning and end of the program. %
However, called methods or loop invariants can talk about states with nonempty broken sets. %
We require the verification engineer to write ghost code that maintains the broken set accurately. The soundness of this reduction is captured by the following Proposition:

\begin{proposition}
\label{prop:ghost-code-Br}
Let $\alpha$ and $\beta$ be quantifier-free formulae over $\Ff \cup \Gg$ (they cannot mention $\Br$). If $\langle\, (\forall z \notin \Br.\, \rho(z)) \land \alpha \land \Br = \emptyset\, \rangle\; P_{\Gg,\Br}(\overline{x}, \Br,\,\mathit{ret}\!:\,\overline{y}, \Br)\; \langle\, (\forall z \notin \Br.\, \rho(z)) \land \beta \land \Br = \emptyset\, \rangle$ is valid then $\langle\, (\forall z.\, \rho(z)) \land \alpha\, \rangle\; P_\Gg(\overline{x},\,\mathit{ret}\!:\,\overline{y})\; \langle\, (\forall z.\, \rho(z)) \land \beta\, \rangle$ is valid, where $P_\Gg$ is the projection of $P_{\Gg,\Br}$ obtained by eliminating the statements that manipulate $\Br$. 
\end{proposition}

The proof of this proposition is similar to the proof of Proposition~\ref{prop:ghost-code}, except that projections only eliminate $\Br$. We provide a detailed argument in our supplementary material in Appendix~\ref{app:stages-soundness}.

\subsection{Stage 3: Eliminating the Universal Quantifier for Well-Behaved Programs}
\label{sec:well-behaved}

We consider triples of the form 
\begin{center}
$\langle\, (\forall z \notin \Br.\, \rho(z)) \land \alpha\, \rangle\; P_{\Gg,\Br}(\overline{x}, \Br,\,\mathit{ret}\!:\,\overline{y}, \Br)\; \langle\, (\forall z \notin \Br.\, \rho(z)) \land \beta\, \rangle$
\end{center} 

\noindent
where $P_{\Gg,\Br}$ is a program augmented with ghost updates to the $\Gg$-fields as well as the $\Br$ set, and $\alpha,\beta$ are quantifier-free formulae that can also mention the fields in $\Gg$ and the $\Br$ set. In this stage we would like to eliminate the quantified conjunct entirely and instead ask the engineer to prove the validity of the triple 
\begin{center}
$\{\alpha\}\; P_{\Gg,\Br}(\overline{x}, \Br,\,\mathit{ret}\!:\,\overline{y}, \Br)\; \{\beta\}$
\end{center}

However, the above two triples are not, in general, equivalent (as broken sets can be manipulated wildly). In this section we define a syntactic class of \emph{well-behaved} programs that force the verification engineer to maintain broken sets correctly, and for such programs the above triple are indeed equivalent. For example, for a field mutation, well-behaved programs require the engineer to determine the set of \emph{impacted objects} where local conditions may be broken by the mutation. The well-behavedness paradigm then mandates that the engineer add the set of impacted objects to the broken set immediately following the mutation statement. Similarly, well-behaved programs do not allow the engineer to remove an object from the broken set unless they show that the local conditions hold on that object. The imposition of this discipline ensures that programmers carefully preserve the meaning of the broken set (i.e., objects outside the broken set must satisfy local conditions). This allows for the quantified conjunct in the triple obtained from Stage 2 to be dropped since it always holds for a well-behaved program. Let us look at such a program:

\begin{example}[Well-Behaved Sorted List Insertion]
\label{ex:slist-insert-wellbehaved}
We use the running example (Example~\ref{ex:slist-insert-triple}) of insertion into a sorted list. We consider a snippet where the key $k$ to be inserted lies between the keys of $x$ and $\nxt(x)$ (which we assume is not $\nil$). We ignore the conditionals that determine $\nxt(x) \neq \nil$ and $\key(x) \leq k \leq \key(\nxt(x))$ for brevity.

We first relax the universal quantification as described in Stage 2 (Section~\ref{sec:broken-sets}) and rewrite the pre and post conditions to ($\forall\,z\notin\Br.\,\LC(z))\land\sortedll(x)\land\Br=\emptyset$. Making the first conjunct implicit, we write the following program that manipulates the broken set in a well-behaved manner. We show the value of the broken set through the program in comments on the right:

\vspace{0.5em}
\hspace{2em}
\begin{minipage}[t]{0.3\textwidth}\footnotesize
\begin{alltt}
\annotation{pre:} \ensuremath{\sortedll(x)\land\Br=\emptyset}
\annotation{post:} \ensuremath{\sortedll(x)\land\Br=\emptyset}
 \ghost{assert x \ensuremath{\notin} Br;}
 \autobr{assume LC(x);}
 y := x.next;    \codecomment{// \{\}}
 z := \keyword{new} C();
 \autobr{Br := Br \ensuremath{\cup} \{z\};} \codecomment{// \{z\}}
 z.key := k;
 \autobr{Br := Br \ensuremath{\cup} \{z\};} \codecomment{// \{z\}}
 z.next := y;
 \autobr{Br := Br \ensuremath{\cup} \{z\};} \codecomment{// \{z\}}
 \hfill
\end{alltt}
\end{minipage}
\hspace{5em}
\begin{minipage}[t]{0.6\textwidth}
\footnotesize
\begin{alltt}
 \ghost{z.sortedll := True;}
 \autobr{Br := Br \ensuremath{\cup} \{z\};} \codecomment{// \{z\}}
 x.next := z;
 \autobr{Br := Br \ensuremath{\cup} \{x\};} \codecomment{// \{x,z\}}
 \ghost{z.rank := (x.rank + y.rank)/2;}
 \autobr{Br := Br \ensuremath{\cup} \{z\};} \codecomment{// \{x,z\}}
 \codecomment{// x and z satisfy LC}
 \autobr{assert LC(z);}
 \ghost{Br := Br \ensuremath{\setminus} \{z\};} \codecomment{// \{x\}}
 \autobr{assert LC(x);}
 \ghost{Br := Br \ensuremath{\setminus} \{x\};} \codecomment{// \{\}}
\end{alltt}
\end{minipage}

\smallskip
We depict the statements enforced by the well-behavedness paradigm in \textcolor{nicepink}{pink} and the ghost updates written by the verification engineer in \textcolor{blue}{blue}. Observe that the paradigm adds the impacted objects to the broken set after each mutation and allocation. Determining the impact set of a mutation is nontrivial; we show how to construct them in Section~\ref{sec:cs-slist-insert}. Note also that to remove $x$ from the broken set we must show $\LC(x)$ holds (assert followed by removal from $\Br$). Finally, we see at the beginning of the snippet that if we show $x \notin \Br$ then we can infer that $\LC(x)$ holds. This follows from the meaning of the broken set. %

\smallskip
\noindent
\textbf{\textit{Putting it All Together.}} The above program corresponds to the program $P_{\Gg,\Br}$ obtained from the Stage 3 reduction, consisting of ghost updates to the $\Gg$ maps and $\Br$. Since it is well-behaved and satisfies the contract $\langle\,\sortedll(x)\land\Br=\emptyset\,\rangle\; P_{\Gg,\Br}\; \langle\,\sortedll(x)\land\Br=\emptyset\,\rangle$ we can conclude that it satisfies the contract $\langle\,(\forall z \notin \Br.\,\rho(z)) \land\sortedll(x)\land\Br=\emptyset\,\rangle\; P_{\Gg,\Br}\; \langle\,(\forall z \notin \Br.\,\rho(z)) \land\sortedll(x)\land\Br=\emptyset\,\rangle$. Using Propositions~\ref{prop:ghost-code} and~\ref{prop:ghost-code-Br} we can project out all augmented code and conclude that the triple given in Example~\ref{ex:slist-insert-triple} with the user's original program and intrinsic specifications is valid! In this way, using FWYB we can verify programs with respect to intrinsic specifications by verifying augmented programs with respect to quantifier-free specifications. The latter can be discharged efficiently in practice using SMT solvers~\cite{Z3,cvc4} (see Section~\ref{sec:qfree-vcs}).\qed

\end{example}

\medskip
\noindent
We dedicate the rest of this section to developing the general theory of well-behaved programs.

\subsection*{Rules for Constructing Well-Behaved Programs}

We define the class of well-behaved programs using a set of rules. We first introduce some notation. 

We distinguish the triples over the augmented programs and quantifier-free annotations by $\{\psi_{\mathit{pre}}\}\, P\, \{\psi_{\mathit{post}}\}$, with $\{\}$ brackets rather than $\langle\,\rangle$. %
$\vdash\{\psi_{\mathit{pre}}\}\, P\, \{\psi_{\mathit{post}}\}$ denotes that a triple is provable. %
Our theory is agnostic to the underlying mechanism for proving triples correct (we use the off-the-shelf verification tool \Boogie in our evaluation). %
However, we assume that the mechanism is sound with respect to the operational semantics. %
We denote that a snippet $P$ is well-behaved by $\wb P$. We also denote that local conditions hold on an object $x$ by $\LC(x)$. 

\smallskip
Figure~\ref{fig:well-behaved} shows the rules for writing well-behaved programs. We only explain the interesting cases here. %

\textsc{Mutation}.\hspace{0.25em} Since mutations can break local conditions, we must grow the broken set. Let $A$ be a finite set of object-type terms over $x$ such that for any $z \notin A$, if $\LC(z)$ held before the mutation, then it continues to hold after the mutation. We refer to such a set $A$ as an \emph{impact set} for the mutation, and we update $\Br$ after a mutation with its impact set. %
The impact set may not always be expressible as a finite set of terms, but this is indeed the case for all the intrinsically defined data structures we use in this paper. We show how to construct impact sets in Section~\ref{sec:cs-slist-insert}.

\textsc{Allocation}.\hspace{0.25em} Allocation does not modify the heap on any existing object. Therefore, we simply update the broken set by adding the newly created object $x$ (this was also the case in Example~\ref{ex:slist-insert-wellbehaved}).

\textsc{Assert LC and Remove}.\hspace{0.25em} This rule allows us to shrink the broken set once the verification engineer fixes the local conditions on a broken location. The snippet \texttt{\autobr{assert\,LC(x);\,}\ghost{Br\,:=\,Br\ensuremath{\setminus}\{x\}}} in Example~\ref{ex:slist-insert-wellbehaved} uses this rule. Informally, the verification engineer is required to show that $\LC(x)$ holds before removing $x$ from $\Br$.

\textsc{Infer LC Outside Br}.\hspace{0.25em} Recall that for well-behaved programs we know implicitly that $\forall x \notin \Br.\, \rho(x)$ holds. This rule allows us to instantiate this implicit fact on objects that we can show lie outside the broken set. The snippet \texttt{\ghost{assert\,x\ensuremath{\notin}Br;\,}\autobr{assume\,LC(x)}} in Example~\ref{ex:slist-insert-wellbehaved} uses this rule.

\smallskip
\noindent
We show that the above rules are sound for the elimination of the universal quantifier in Stage 3:

\begin{proposition}
\label{prop:well-behaved-sound}
Let $[(M:P); (N_1:Q_1)\ldots, (N_k:Q_k)]$ be a program (which can use $\Gg$ and $\Br$) such that $\wb P$ and $\wb Q_i, 1 \leq i \leq k$. Let $\alpha$ and $\beta$ be quantifier-free formulae over $\Ff \cup \Gg$ which can use $\Br$. If $\{\alpha\}\; P(\overline{x}, \Br,\,\mathit{ret}\!:\,\overline{y}, \Br)\; \{\beta\}$ is valid, then $\langle\, (\forall z \notin \Br.\, \rho(z)) \land \alpha\, \rangle\; P(\overline{x}, \Br,\,\mathit{ret}\!:\,\overline{y}, \Br)\; \langle\, (\forall z \notin \Br.\, \rho(z)) \land \beta\, \rangle$ is valid.
\end{proposition}

We prove the above proposition by structural induction on the rules in Figure~\ref{fig:well-behaved}. We provide the proof in Appendix~\ref{app:stages-soundness} of our supplementary material.

\smallskip
In the above presentation we use only one broken set for simplicity of exposition. Our general framework allows for finer-grained broken sets that can track breaks over a partition on the local conditions. For example, in Section~\ref{sec:cs-overlaid} we verify deletion in an overlaid data structure consisting of a linked list and a binary search tree using two broken sets: one each for the local conditions of the two component data structures.

\begin{figure}
\begin{mathpar}\footnotesize
\inferrule[Skip/Assignment/Lookup/Return]{ \\ }{\wb s \textrm{   where $s$ is of the form} \\\\ \textrm{$\mathsf{skip}$,~\texttt{x:=y},~  \texttt{x:=y.f},~ or~ $\mathsf{return}$}} 
\and
\inferrule[Mutation]{ \\\\ \vdash \{\, z \notin A \land LC(z) \land x \neq \nil \,\}\; x.f := v\; \{\, LC(z) \,\} }{ \wb\; x.f := v\,;\, \Br := \Br \cup A \\\\ \textrm{where $A$ is a finite set of location terms over $x$} }
\and
\inferrule[Allocation]{ \\ }{\wb\; x := \mathsf{new}\, C()\,;\, \Br := \Br \cup \{x\}}
\and
\inferrule[Function Call]{ \\ }{\wb\; \overline{y}, \Br := \mathit{Function}(\overline{x},\Br)}
\and
\inferrule[Infer LC Outside Br]{ \\ }{ \wb\; \mathsf{if}\; (x \neq \nil \land x \notin \Br)\; \mathsf{then}\; \mathsf{assume}\,LC(x)%
}
\and
\inferrule[Assert LC and Remove]{ \\ }{ \wb\; \mathsf{if}\; LC(x)\; \mathsf{then}\; \Br := \Br \setminus \{x\}%
}
\and
\inferrule[Composition]{ \\\\ \wb\; P \\ \wb Q }{ \wb\; P\,;\, Q }
\and
\inferrule[If-Then-Else]{ \\\\ \wb\; P \\ \wb Q }{ \wb\; \mathsf{if}\; \mathit{cond}\; P\; \mathsf{else}\; Q \\\\ \textrm{where $\mathit{cond}$ does not mention $\Br$} }
\and
\inferrule[While]{ \\\\ \wb\; P }{ \wb\; \mathsf{while}\; \mathit{cond}\; \mathsf{do}\; P \\\\ \textrm{where $\mathit{cond}$ does not mention $\Br$} }
\end{mathpar}
\caption{Rules for constructing well-behaved programs. Local condition formula instantiated at $x$ is denoted by $\LC(x)$. The statement $(\mathsf{if}\, \mathit{cond}\, \mathsf{then}\, S)$ is sugar for $(\mathsf{if}\, \mathit{cond}\, \mathsf{then}\, S\, \mathsf{else}\, \mathsf{skip})$.}
\label{fig:well-behaved}
\end{figure}

\subsection{Soundness of FWYB}
\label{sec:soundness}

In this section we state the soundness of the FWYB methodology. 

\begin{theorem}[FWYB Soundness]
\label{thm:soundness}
Let $(\Gg, \LC, \varphi)$ be an intrinsic definition with $\Gg = \{g_1,g_2\ldots, g_l\}$. Let $[(M:P); (N_1:Q_1)\ldots, (N_k:Q_k)]$ be an augmented program constructed using the FWYB methodology such that $\,\wb P$ and $\,\wb Q_i,\, 1\leq i \leq k$, i.e., the programs $P$ and $Q_i$ are well-behaved (according to the rules in Figure~\ref{fig:well-behaved}). Let $\varphi$, $\psi_{\mathit{pre}}$, and $\psi_{\mathit{post}}$ be quantifier-free formulae that do not mention $\Br$ (but can mention the maps in $\Gg$). Finally, let $[(\hat{M}:\hat{P}); (\hat{N_1}:\hat{Q_1})\ldots, (\hat{N_k}:\hat{Q_k})]$ be the projected user-level program according to Definition~\ref{defn:code-projection}. Then, if the triple:
\begin{center}
$\{\varphi \land \psi_{\mathit{pre}} \land \Br = \emptyset \}\; P\; \{\varphi \land \psi_{\mathit{post}} \land \Br = \emptyset\}$
\end{center}

\noindent
is valid, then the triple

\begin{center}
$\langle\,\soexists\, g_1,g_2\ldots, g_l.\, (\LC \land \varphi \land \psi_{\mathit{pre}}) \,\rangle\; \hat{P}\; \langle\,\soexists\, g_1,g_2\ldots, g_l.\, (\LC \land \varphi \land \psi_{\mathit{post}}) \,\rangle$
\end{center}

\noindent
is valid (according to Definition~\ref{defn:ids-triple-validity}).
\end{theorem}

Informally, the soundness theorem says that given a user-written program, if we (a) augment it with updates to ghost fields and the broken set only using the discipline for well-behaved programs, and (b) show that if the broken set is empty at the beginning of the program it will be empty at the end, then the original user-written program satisfies the intrinsic specifications on preservation of the data structure. 

The proof of the theorem trivially follows from the soundness of the three stages. Let us write $P$ as $P_{\Gg,\Br}$ to emphasize that the program contains ghost code that manipulates both the $\Gg$ maps and $\Br$. We begin with the fact that $\{\varphi \land \psi_{\mathit{pre}} \land \Br = \emptyset \}\; P_{\Gg,\Br}\; \{\varphi \land \psi_{\mathit{post}} \land \Br = \emptyset\}$ is valid. Since $P$ and its auxiliary functions are well-behaved we have from Proposition~\ref{prop:well-behaved-sound} that $\langle\, (\forall z \notin \Br.\, \rho(z)) \land \varphi \land  \psi_{\mathit{pre}}\, \rangle\; P_{\Gg,\Br}\; \langle\, (\forall z \notin \Br.\, \rho(z)) \land \varphi \land  \psi_{\mathit{post}}\, \rangle$ is valid.

Next, we use Proposition~\ref{prop:ghost-code-Br} to conclude that $\langle\, (\forall z.\, \rho(z)) \land \varphi \land  \psi_{\mathit{pre}}\, \rangle\; P_\Gg\; \langle\, (\forall z.\, \rho(z)) \land \varphi \land  \psi_{\mathit{post}}\, \rangle$ is valid, where $P_\Gg$ is the projection of $P_{\Gg,\Br}$ obtained by eliminating the statements that manipulate $\Br$. Finally, we use Proposition~\ref{prop:ghost-code}, along with the fact that $\forall z.\, \rho(z)$ is $\LC$ and $\hat{P_\Gg}$ is the same as $\hat{P}$ to conclude that $\langle\,\soexists\, g_1,g_2\ldots, g_l.\, (\LC \land \varphi \land \psi_{\mathit{pre}}) \,\rangle\; \hat{P}\; \langle\,\soexists\, g_1,g_2\ldots, g_l.\, (\LC \land \varphi \land \psi_{\mathit{post}}) \,\rangle$ is valid\footnotemark.\qed

\footnotetext{The presentation of FWYB augments the original program $P$ with manipulations to $\Gg$ and $\Br$ in separate stages. This is done for clarity of exposition. This may not be possible in general since we may write ghost code with expressions that use both the $\Gg$ maps and $\Br$. However, we can combine the proofs of Propositions~\ref{prop:ghost-code} and ~\ref{prop:ghost-code-Br} to show the soundness of projecting out all ghost code in a single stage, and Theorem~\ref{thm:soundness} continues to hold in the general case.}

\subsection{Generating Quantifier-Free Verification Conditions}
\label{sec:qfree-vcs}
We state at several points in this paper that verifying augmented programs with quantifier-free specifications reduces to validity over combinations of quantifier-free theories. However, this is not obvious. Unlike scalar programs, quantifier-free contracts do not guarantee quantifier-free verification conditions (VCs) for heap programs. In particular, commands such as allocation and function calls pose challenges. However, we show that in our case it is indeed possible to obtain quantifier-free VCs. We do this by transforming a given heap program into a scalar program that explicitly models changes to the heap. We model allocation using a ghost set $\Alloc$ corresponding to the allocated objects and update it when a new object is allocated. We reason about arbitrary changes to the heap across a function call by requiring a `modifies' annotation from the user and adding assumptions that the fields of objects outside the modified set of a function call remain the same across the call. %
We express these assumptions using parameterized map updates which are supported by the generalized array theory {\cite{pointwisearrays}}. We detail this reduction in our supplementary material in Appendix~{\ref{app:qfree-vcs}}.

\fi

\section{Illustrative Data Structures and Verification}
\label{sec:case-studies}

Intrinsic definitions and the fix-what-you-break verification methodology are new concepts that require thinking afresh about data structures and annotating methods that operate over them. In this section, we present several classical data structures and methods over them, and illustrate how the verification engineer can write intrinsic definitions (which maps to choose, and what the local conditions ensure) and how they can fix broken sets to prove programs correct. 

\subsection{Insertion into a Sorted List}
\label{sec:cs-slist-insert}

In this section we present the verification of insertion into a sorted list implemented in the FWYB methodology in its entirety. Our running example in Section~\ref{sec:proglogic} illustrates the key technical ideas involved in verifying the program. In this section we present an end-to-end picture that mirrors the verification experience in practice.

\mypara{Data Structure Definition} We first revise the definition of a sorted list (Example~\ref{ex:ids-sorted-list}) with a different set of monadic maps. %
We have the following monadic maps $\Gg$--- $\prev: C \rightarrow C?$, $\len: C \rightarrow \mathbb{N}$, $\keys: C \rightarrow \mathit{Set(Int)}$, $\hslist: C \rightarrow \mathit{Set(C)}$ that model the \emph{previous} pointer (inverse of next), length of the sorted list, the set of keys stored in it, and its heaplet (set of locations that form the sorted list) respectively. We use the length, keys, and heaplet maps to state full functional specifications of methods. The local conditions are:
\begin{equation}
\label{eq:slist-lc-dafny}
\begin{aligned}
\forall x.\, \nxt(x) \neq \nil \Rightarrow\; & (~ \key(x) \leq \key(\nxt(x)) ~\land~ \prev(\nxt(x)) = x\\
&\land~ \len(x) = 1+ \len(\nxt(x)) \,\land\, \keys(x) = \{ \key(x) \} \cup \keys(\nxt(x))\\
&\land~ \hslist(x) = \{ x \} \uplus \hslist(\nxt(x))~) \textrm{\hspace{4em}($\uplus$: disjoint union)}\\
~~~~~\land~ \prev(x) \neq \nil \Rightarrow\; & \nxt(\prev(x)) = x\\ %
~~~~~\land~ \nxt(x) = \nil \Rightarrow\; & (~\len(x) = 1 \,\land\, \keys(x) = \{ \key(x) \} \,\land\, \hslist(x) = \{ x \}~)
\end{aligned}
\end{equation}

The above definition is slightly different from the one given in Example~\ref{ex:ids-sorted-list}. %
The $\len$ map replaces the $\rank$ map, requiring additionally that lengths of adjacent nodes differ by $1$. 

The $\prev$ map is a gadget we find useful in many intrinsic definitions. The constraints on $\prev$ ensure that the $C$-heaps satisfying the definition only contain non-merging lists. To see why this is the case, consider for the sake of contradiction distinct objects $o_1, o_2, o_3$ such that $\nxt(o_1) = \nxt(o_2) = o_3$. Then, we can see from the local conditions that we must simultaneously have $\prev(o_3) = o_1$ and $\prev(o_3) = o_2$, which is impossible. %
Finally, the $\hslist$ and $\keys$ maps represent the heaplet and the set of keys stored in the sorted list (respectively).

The heads of all sorted lists in the $C$-heap is then defined by the following correlation formula:
\begin{center}
$\varphi(y) \equiv \prev(y) = \nil.$
\end{center}

\mypara{Constructing Provably Correct Impact Sets for Mutations} %
We now instantiate the rules developed in Section~\ref{sec:well-behaved} for sorted lists. Recall that well-behaved programs must update the broken set with the impact set of a mutation. %
Table~\ref{tbl:slist-mutations} captures the impact set for each field mutation. Note that the terms denoting the impacted objects belong to $A_f$ only if they do not evaluate to $\nil$. %

\begin{figure*}%
\centering
\begin{minipage}{0.45\textwidth}
\begin{tikzpicture}[node distance={25mm}, thick, main/.style = {draw, circle}]
\node[main,draw=brown,fill=lightgray] (x) {$x$};
\node[main,draw=brown,fill=lightgray] (y) [right of=x] {$y$}; 
\node (ytext) [above right=1mm and -8mm of y,text width = 7mm] {$y =$ $\old(\nxt(x))$};
\node[main] (z) [below=6mm of y] {$z$};
\node[main] (w) [left of=x] {$w$};
\draw[->] (w) -- (x) node[below,midway] {$\nxt$};
\draw[<-] (w) to [bend left] node[above,midway] {$\prev$} (x);
\draw[dashed,->] (x) -- (y); %
\draw[<-] (x) to [bend left] node[above,midway] {$\prev$} (y) ;
\draw[->] (x) -- (z) node[below left,midway] {$\nxt$}; 
\end{tikzpicture}
\caption{\footnotesize Reasoning about the set of objects broken by \texttt{x.next := z}. The dashed arrow represents the old $\nxt$ pointer before the mutation. The grey nodes denote objects where local conditions can be broken by the mutation. We see that only $x$ and $y$ may violate $\nxt$ and $\prev$ being inverses.} 
\label{fig:next-broken}
\end{minipage}
\hspace{15pt}
\begin{minipage}{0.45\textwidth}
\begin{table}[H]
\caption{\footnotesize Table of impact sets corresponding to field mutations for sorted lists (See~\ref{eq:slist-lc-dafny} in Section~\ref{sec:cs-slist-insert}). $\old(t)$ refers to the value of the term $t$ before the mutation. Terms only belong to the sets if not equal to $\nil$.}
{\small
\begin{tabular}{|c|c|}
 \hline
 Mutated Field $f$ & Impacted Objects $A_f$ \\ 
 \hline
 $x.\nxt$ & $\{x, \old(\nxt(x))\}$ \\ 
 $x.\key$ & $\{x, \prev(x)\}$\\
 $x.\prev$ & $\{x, \old(\prev(x))\}$\\
 $x.\hslist$ & $\{x, \prev(x)\}$ \\
 $x.\len$ & $\{x, \prev(x)\}$ \\
 $x.\keys$ & $\{x, \prev(x)\}$ \\
 \hline
\end{tabular}
}
\vspace{1em}
\label{tbl:slist-mutations}
\end{table}
\end{minipage}
\end{figure*}

Let us consider the correctness of  Table~\ref{tbl:slist-mutations}, focusing on the mutation of $\nxt$ as an example. Figure~\ref{fig:next-broken} illustrates the heap after the mutation \texttt{x.next := z}. %
We make the following key observation: the local constraints $\LC(v)$ for an object $v$ refer only to the properties of objects $v$, $\nxt(v)$, and $\prev(v)$ (see~\ref{eq:slist-lc-dafny}), i.e., objects that are at most ``one step'' away on the heap. %
Therefore, the only objects that can be broken by the mutation \texttt{x.next := z} are those that are one step away from $x$ either via an incoming or an outgoing edge via pointers $\nxt$ and $\prev$. This is a general property of intrinsic definitions: \emph{mutations cannot immediately affect objects that are far away on the heap}. %
\footnote{Note that a mutation can necessitate changes to monadic maps for an unbounded number of nodes \emph{eventually}; however, these are not necessary immediately. As we fix monadic maps on a broken object, its neighbors could get broken and need to be fixed, leading to their neighbors breaking, etc. This can lead to a ripple effect that would eventually require an unbounded number of locations to be fixed.}

In our case, we claim that the impact set contains at most $x$ and $\old(\nxt(x))$. 
Here's a proof (see Fig~\ref{fig:next-broken}): %
Consider $z$ such that $z \neq \old(\nxt(x))$ (as there is no real mutation otherwise). If $z$ was not broken before the mutation, then it cannot be the case that $\prev(z)=x$. %
Looking at the local conditions, it is clear that such a $z$ will remain unbroken after the mutation. Now consider a $w$ not broken before the mutation such that $\nxt(w)=x$. Then it follows from the local conditions that %
there can only be one such (unbroken) $w$, and further $w \neq x$. $w$'s fields are not mutated, and by examining $\LC$, it is easy to see that $w$ will not get broken (as $\LC(v)$ does not refer to $\nxt(\nxt(v))$). The argument is the same for $w$ such that $\prev(x)=w$. Finally, consider a $y$ not broken before the mutation such that $\prev(y)=x$. We can then see from the local conditions that $y=\old(\nxt(x))$, which is already in the impact set.

The above argument is subtle, but we can automatically check whether impact sets declared by a verification engineer are correct. The \textsc{Mutation} rule in Figure~\ref{fig:well-behaved} characterizes the impact set $A_{next}$ for mutation of the field $next$ as follows:
\begin{center}
$\vdash\; \{u \neq x \land u \neq \nxt(x) \land \LC(u) \land x \neq \nil\}\; x.\nxt := z\; \{\LC(u)\}$
\end{center}

\noindent
The above says that any location $u$ that is not in the impact set which satisfied the local conditions before the mutation must continue to satisfy them after the mutation. We present the formulation for the general case in our supplementary material in Appendix~\ref{app:wellbehaved}. Finally, note that the validity of the above triple is decidable. In our realization of the FWYB methodology we prove our impact sets correct by encoding the triple in \Boogie (see Section~\ref{sec:evaluation}).

\mypara{Macros that Ensure Well-Behaved Programs} In Section~\ref{sec:well-behaved} we characterized well-behaved programs as a set of syntactic rules (Figure~\ref{fig:well-behaved}). 
We can realize these restrictions using macros:

\begin{enumerate}
    \item \texttt{Mut(x,f,v,Br)} for each $f \in \Ff \cup \Gg$, which represents the sequence of statements~\texttt{x.f := v;\, Br := Br U A$_{\texttt{f}}$(x)}. Here \texttt{A$_{\texttt{f}}$(x)} is the impact set corresponding to the mutation on $f$ on $x$ as given by the table above. 
    This macro is used instead of \texttt{x.f := v}
    and automatically ensures that the impact set is added to the broken set.
    \item \texttt{NewObj(x,Br)}, which represents the statements~\texttt{x := new C();\, Br := Br U \{ x \}}. This macro is used instead of \texttt{x := new C()} and ensures that any newly allocated object is automatically added to the broken set.
    \item \texttt{AssertLCAndRemove(x,Br)}, which represents the statements~\texttt{assert LC(x);\, Br := Br\! \textbackslash\!\, \{ x \}}. 
    This macro is allowed anytime the engineer wants to assert that $x$ satisfies the local condition, and then remove it from the broken set.\footnote{We extend our basic programming language defined in Figure~\ref{fig:prog-lang} with an $\mathsf{assert}$ statement and give it the usual semantics (program reaches an error state if the assertion is not satisfied, but is equivalent to $\mathsf{skip}$ otherwise).}
    \item \texttt{InferLCOutsideBr(x, Br)}, which represents the statements~\texttt{assert (x\! $\neq$\! nil $\land$ x\! $\notin$\! Br);\, assume LC(x)}. This allows the engineer at any time to assert that $x$ is not in the broken set and assume it satisfies the local condition.  
\end{enumerate}

The above macros correspond to the rules \textsc{Mutation}, \textsc{Allocation}, \textsc{Assert LC and Remove}, and \textsc{Infer LC Outside Br} respectively. 
Restricting to the syntactic fragment that  contains the above macros and disallows mutation and allocation otherwise enforces the \emph{programming discipline} that ensures well-behaved programs.

\iflong %

\begin{figure}
\scriptsize
\raggedright
\begin{subfigure}{0.45\textwidth}
\centering
\begin{alltt}
 \annotation{pre:} \ensuremath{\Br=\emptyset}
 \annotation{post:} \ensuremath{\LC(r)\,\land\,\prev(r)=\nil}
\hspace{20pt}\ensuremath{\land \Br=\,\ite(\old(\prev(x))=\nil,\,\emptyset,\,\{\old(\prev(x))\})}
\hspace{20pt}\ensuremath{\land \len(r)=\,\old(\len(x))+1}
\hspace{20pt}\ensuremath{\land \keys(r)=\,\old(\keys(x))\cup\{k\}}
\hspace{20pt}\ensuremath{\land \old(\hslist(x))\subset\hslist(r)}
 \annotation{modifies:} \ensuremath{\hslist(x)}
 \funcname{sorted_list_insert}(x: \type{C}, k: \type{Int}, Br: \type{Set(C)}) 
 \keyword{returns} r: \type{C}, Br: \type{Set(C)}
 \{
   \ghost{InferLCOutsideBr(x, Br);}
   \keyword{if} (x.key \ensuremath{\geq} k) \keyword{then} \{ \codecomment{// k inserted before x}
     NewObj(z, Br);         \codecomment{// \{z\}}
\hspace{18pt}Mut(z, key, k, Br);    \codecomment{// \{z\} since z.prev = nil}
\hspace{18pt}Mut(z, next, x, Br);   \codecomment{// \{z\} since z.next = nil}
\hspace{18pt}\ghost{Mut(z, hslist, \{z\} \ensuremath{\cup} x.hslist, Br);} \codecomment{// \{z\}}
\hspace{18pt}\ghost{Mut(z, length, 1 + x.length, Br);}   \,\codecomment{// \{z\}}
\hspace{18pt}\ghost{Mut(z, keys,  \{k\} \ensuremath{\cup} x.keys, Br);}    \codecomment{// \{z\}}
\hspace{18pt}\ghost{Mut(x, prev, z, Br);}     \codecomment{// \{z, x, old(prev(x))\}}
\hspace{18pt}\ghost{AssertLCAndRemove(z, Br);}  \codecomment{// \{x, old(prev(x))\}}
\hspace{18pt}\ghost{AssertLCAndRemove(x, Br);} \codecomment{// \{old(prev(x))\}}
\hspace{18pt}r := z;
   \}
   \keyword{else} \{ 
   \keyword{if} (x.next = nil) \keyword{then} \{ \codecomment{// one-element list}
\hspace{18pt}NewObj(z, Br);
\hspace{18pt}Mut(z, key, k, Br);
\hspace{18pt}Mut(z, next, nil, Br);
\hspace{18pt}\ghost{Mut(z, hslist, \{z\}, Br);}
\hspace{18pt}\ghost{Mut(z, length, 1, Br);}
\hspace{18pt}\ghost{Mut(z, keys, \{k\}, Br);}
\hspace{18pt}Mut(x, next, z, Br);
\end{alltt} 
\end{subfigure}
\hspace{10pt}
\begin{subfigure}{0.45\textwidth}
\raggedright
\begin{alltt}
\hspace{18pt}\ghost{Mut(z, prev, x, Br);}
\hspace{18pt}\ghost{AssertLCAndRemove(z, Br);}
\hspace{18pt}\ghost{Mut(x, prev, nil, Br);}
\hspace{18pt}\ghost{Mut(x, hslist, \{x\} \ensuremath{\cup} \{z\}, Br);}
\hspace{18pt}\ghost{Mut(x, length, 2, Br);}
\hspace{18pt}\ghost{Mut(x, keys, \{x.key\} \ensuremath{\cup} \{k\}, Br);}
\hspace{18pt}\ghost{AssertLCAndRemove(x, Br);}
\hspace{18pt}r := x;
   \}
   \keyword{else} \{ \codecomment{// recursive case}
\hspace{18pt}y := x.next;
\hspace{18pt}\ghost{InferLCOutsideBr(y, Br);}
\hspace{18pt}tmp, Br := \funcname{sorted_list_insert}(y, k, Br);   \codecomment{// \{x\}}
\hspace{18pt}\ghost{InferLCOutsideBr(y, Br);}
\hspace{18pt}\ghost{if (y.prev = x) then \{
\hspace{18pt}  Mut(y, prev, nil, Br);    \codecomment{// \{y, x\}}
\hspace{18pt}\}}
\hspace{18pt}Mut(x, next, tmp, Br);      \codecomment{// \{y, x\}}
\hspace{18pt}\ghost{AssertLCAndRemove(y, Br);}   \codecomment{// \{x\}}
\hspace{18pt}\ghost{Mut(tmp, prev, x, Br);}      \codecomment{// \{tmp, x\}}
\hspace{18pt}\ghost{AssertLCAndRemove(tmp, Br);} \codecomment{// \{x\}}
\hspace{18pt}\ghost{Mut(x, hslist, \{x\} \ensuremath{\cup} tmp.hslist, Br);}  \codecomment{// \{x, prev(x)\}}
\hspace{18pt}\ghost{Mut(x, length, 1 + tmp.length, Br);}     \,\codecomment{// \{x, prev(x)\}}
\hspace{18pt}\ghost{Mut(x, keys, \{x.key\} \ensuremath{\cup} tmp.keys, Br);}  \codecomment{// \{x, prev(x)\}}
\hspace{18pt}\ghost{Mut(x, prev, nil, Br);}           \codecomment{// \{x, old(prev(x))\}}
\hspace{18pt}\ghost{AssertLCAndRemove(x, Br);}   \codecomment{// \{old(prev(x))\}}
\hspace{18pt}r := x;
   \}\}
 \}
\end{alltt}
\end{subfigure}
\caption{Code for insertion into a sorted list written in the syntactic fragment for well-behaved programs(Section~\ref{sec:cs-slist-insert}). Black lines denote code written by the user, and blue lines denote lines written by the verification engineer. The comments on the right show the state of the broken set $\Br$ after the statement on the corresponding line.}
\label{fig:slist-insert-code}
\end{figure}

\iflong
\mypara{Verifying Sorted List Insertion} 
\fi
We provide the specifications and the code augmented with ghost annotations in Figure~\ref{fig:slist-insert-code}. 

\textit{Specifications.} The precondition states that the broken set is empty at the beginning of the program. The postcondition states that the returned object $r$ satisfies the local conditions and satisfies the correlation formula for a sorted list (i.e., $\prev(r)= \nil$). However, the broken set is only empty if the input object $x$ was the head of a sorted list, and it is $\{\prev(x)\}$ otherwise. %
The other conjuncts express functional specifications for insertion in terms of the length, heaplet, and set of keys. We also add a `modifies' clause which enables program verifiers for heap manipulating programs to utilize frame reasoning across function calls.

\textit{Summary.} The proof works at a high-level as follows: we recurse down the list, reaching the appropriate object $x$ before which the new key must be inserted. This is the first branch in Figure~\ref{fig:slist-insert-code}, and we show the broken set at each point in the comments to the right. We create the new object $z$ with the appropriate key and point $z.\nxt$ to $x$. We then fix the local conditions on $x$ and $z$. However, these fixes break the $\LC$ on $\old(\prev(x))$. We maintain this property up the recursion, at each point fixing $\LC$ on $x$ and breaking it on $\old(\prev(x))$ in the process. This is shown in the last branch in the code. We eventually reach the head of the sorted list, whose $\prev$ in the pre state is $\nil$, and at that point the fixes do not break anything else, i.e., the broken set is empty (as desired). 

The verification engineer adds ghost code to perform these fixes as shown in \textcolor{blue}{blue} in Figure~\ref{fig:slist-insert-code}. We can also see that there are essentially as many lines of ghost code as there are lines of user code; we compare these values across our benchmark suite (see Table~\ref{tbl:results}) and find that this is typical for many methods. However, the verification conditions for the (augmented) program are \emph{decidable} because they can be stated using quantifier-free formulas over decidable combinations of theories including maps, map updates, and sets.

\else

\smallskip
\noindent
We present the full well-behaved code written using the above macros and discuss it in our supplementary material in Appendix~{\ref{app:cs-slist-insert}}.

\fi

\iflong %
\subsection{BST Right-Rotation}
\label{sec:cs-bst-rotate}
We now turn to another data structure and method that illustrates intrinsic definitions for trees, namely verifying a right rotate on a binary search tree. Such an operation is a common tree operation, and rotations are used widely in maintaining balanced search trees, such as AVL and Red-Black Trees, on which several of our benchmarks operate. 

We augment the definition of binary trees discussed in Section \ref{sec:intro} to include the $min: BST \rightarrow Real$ and $max: BST \rightarrow Real$ maps, which capture the minimum and maximum keys stored in the tree rooted at a node, to help enforce binary search tree properties locally. The local condition and the impact sets are as below:

\begin{figure}[H]
    \centering
    \hspace{-0.5em}
    \begin{subfigure}[b]{0.6\linewidth}
        {\scriptsize
        \begin{equation*}
        \begin{aligned}
        \mathit{LC} \equiv \forall x.\, min(x) & \leq key(x) \leq max(x) \\
        \land~ (p(x) \neq \nil \Rightarrow\; & l(p(x)) = x \lor r(p(x)) = x) \\
        \land~ (l(x) = nil \Rightarrow\; & min(x) = key(x)) \\
        \land~ (l(x) \neq \nil \Rightarrow\; & p(l(x)) = x \land rank(l(x)) < rank(x) \\
        &\!\!\!\!\!\land~  max(l(x)) < key(x)
        ~\land~ min(x) = min(l(x)))\\
        \land~ (r(x) = nil \Rightarrow\; & max(x) = key(x)) \\
        \land~ (r(x) \neq \nil \Rightarrow\; & p(r(x)) = x \land rank(r(x)) < rank(x) \\
        &\!\!\!\!\!\land~  min(r(x)) > key(x)
        ~\land~  max(x) = max(r(x)))\\
        \end{aligned}
        \end{equation*}
        }
    \end{subfigure}%
    \hspace{0.1em}
    \begin{subfigure}[b]{0.35\linewidth}
        \begin{center}
        {\scriptsize
        \begin{tabular}{ |c|c| } 
         \hline
         Mutated Field $f$ & Impacted Objects $A_f$ \\ 
         \hline
         $l$ & $\{x, \old(l(x))\}$ \\ 
         $r$ & $\{x, \old(r(x))\}$\\
         $p$ & $\{x, \old(p(x))\}$\\
         $\textit{key}$ & $\{x\}$ \\
         $\textit{min}$ & $\{x, p(x)\}$ \\
         $\textit{max}$ & $\{x, p(x)\}$ \\
         $\textit{rank}$ & $\{x, p(x)\}$ \\
         \hline
        \end{tabular}
        }
        \end{center}
    \end{subfigure}%
\end{figure}

\iflong
We present the gist of how the data structure is repaired here and leave the fully annotated program to the Appendix~\ref{app:cs-bst-rotate}. 
\else
We first describe the gist of how the data structure is repaired and provide the fully annotated program below. 
\fi
Recall that in a BST right rotation, that there are two nodes $x$ and $y$ such that $y$ is $x$'s left child. After the rotation is performed, $y$ becomes the new root of the subtree, while $x$ becomes $y$'s right child. Several routine updates of the  monadic map $p$ (parent) will have to be made. The most interesting update is that of the $rank: BST \rightarrow Real$ map. Since $y$ is now the root of the affected subtree, its rank must be greater than all its children. One way of doing this is to increase $y$'s rank to something greater than $x$'s rank. This works if $y$ has no parent, but not in general. To solve this issue, we use the density of the Reals to set the rank of $y$ to $(rank(x) + rank(p(y)))/2$. 
Note that there are a fixed number of ghost map updates, as the various monadic maps for distant ancestors and descendents of $x,y$ do not change (the min/max of subtrees of such nodes do not change).

\fi

\subsection{Reversing a Sorted List}
\label{sec:cs-slist-reverse}
We return to lists for another case study: reversing a sorted list. The purpose of this example is to demonstrate how the fix-what-you-break philosophy works with iteration/loops. We augment the definition of sorted linked lists from Case Study~\ref{sec:cs-slist-insert} to make sortedness optional and determined by predicates that capture sortedness in non-descending order, with $\sorted: C \rightarrow Bool$, and sortedness
with non-ascending order, with $\revsorted: C \rightarrow Bool$. The relevant additions to the local condition and the impact sets for these monadic maps can be seen below:
\vspace*{-0.3cm}
\begin{figure}[H]
    \centering\footnotesize
    \hspace{-0.15\linewidth}
    \begin{subfigure}{0.5\linewidth}
        {\footnotesize
        \begin{equation*}
        \begin{aligned}
        (&\nxt(x) \neq \nil \Rightarrow\\
        &\sorted(x) \Rightarrow (\key(x) \leq \key(\nxt(x))
        \land \sorted(x) = \sorted(\nxt(x)))\\
        &\land~ \revsorted(x) \Rightarrow 
        (\key(x) \geq \key(\nxt(x))\\
        & \qquad\qquad\qquad\qquad\land~ \revsorted(x) = \revsorted(\nxt(x))))\\
        \end{aligned}
        \end{equation*}
        }
    \end{subfigure}%
    \hspace{0.12\linewidth}
    \begin{subfigure}{0.2\linewidth}
        \begin{center}
        {\scriptsize
        \begin{tabular}{ |c|c| } 
         \hline
         Mutated Field $f$ & Impacted Objects $A_f$ \\ 
         \hline
         $\sorted$& $\{x, \prev(x)\}$ \\
         $\revsorted$& $\{x, \prev(x)\}$ \\
         \hline
        \end{tabular}
        }
        \end{center}
    \end{subfigure}%
\end{figure}
\vspace*{-0.2cm}

We present the full local condition and code in our supplementary material in Appendix \ref{app:cs-slist-reverse}. However, the gist of the method is that we are popping $C$ nodes off of the front of a temporary list $cur$, and pushing them to the front of a new reversed list $ret$ repeatedly using a loop. A technique we use to verify loops using FWYB is to maintain that the broken set contains no nodes or only a finite number of nodes for which we specify how they are broken. In the case of this method, $Br$ remains empty, as the loop maintains $cur$ and $ret$ as two valid lists, not modifying any other nodes. When popping $x$ from $cur$ and adding it to $ret$, in addition to repairing the new $cur$ by setting its parent pointer to $\nil$, we also need to update fields such as $\len$ and $\keys$ on $x$, so it satisfies the relevant local conditions as the new head of the $ret$ list.

\iflong %
\subsection{Merging Sorted Lists}
\label{sec:cs-slist-merge}
We demonstrate the ability of intrinsic definitions to handle multiple data structures at once, using
the example of in-place  merging of two sorted lists. %
The method merges the two lists by reusing the two lists' elements, which is a natural pattern for imperative code. Once again, we extend the definition of sorted lists from Case Study \ref{sec:cs-slist-insert}. We add the predicates $list1: C \rightarrow Bool$, $list2: C \rightarrow Bool$, and $list3: C \rightarrow Bool$, to indicate disjoint classes of lists. The relevant local condition and impact sets are:
\begin{figure}[H]
    \centering
    \hspace{-0.15\linewidth}
    \begin{subfigure}{0.5\linewidth}
        {\scriptsize
        \begin{equation*}
        \begin{aligned}
        &(list1(x) \lor list2(x) \lor list3(x)) \\
        \land~ &\neg(list1(x) \land list2(x)) \land \neg(list2(x) \land list3(x)) \\
        \land~ &\neg(list1(x) \land list3(x)) \\
        \land~ &(list1(x) \Rightarrow (\nxt(x)\neq\nil \Rightarrow list1(\nxt(x)))) \\
        \land~ &(list2(x) \Rightarrow (\nxt(x)\neq\nil \Rightarrow list2(\nxt(x)))) \\
        \land~ &(list3(x) \Rightarrow (\nxt(x)\neq\nil \Rightarrow list3(\nxt(x)))) \\
        \end{aligned}
        \end{equation*}
        }
    \end{subfigure}%
    \hspace{0.1\linewidth}
    \begin{subfigure}{0.2\linewidth}
        \begin{center}
        {\scriptsize
        \begin{tabular}{ |c|c| } 
         \hline
         Mutated Field $f$ & Impacted Objects $A_f$ \\ 
         \hline
         $list1$& $\{x, \prev(x)\}$ \\
         $list2$& $\{x, \prev(x)\}$ \\
         $list3$& $\{x, \prev(x)\}$ \\
         \hline
        \end{tabular}
        }
        \end{center}
    \end{subfigure}%
\end{figure}

Disjointness of the three lists is ensured by insisting that every object has at most one of the three list predicates hold.

We give a gist of the proof of the merge method. 
The recursive program compares the keys at the heads of the first and second sorted lists, and adds the appropriate node to the front of the third list. It turns out that we can easily update the ghost maps for this node (making it belong to the third list, and updating its parent pointer and key set) as well as updating the parent pointer of the head of the list where the node is removed from. When one of the lists is empty, we append the third list to the non-empty list using a single pointer mutation and then, using a ghost loop, we update the nodes in the appended list to make $list3$ true (this needs a  loop invariant involving the broken set).

\fi

\subsection{Circular Lists}
\label{sec:cs-cl-insertback}
Our next example is circular lists. This example illustrates a neat trick in FWYB that where we assert that we can reach a special node known as a \emph{scaffolding} node, and that in addition to asserting properties on the node that is given to the method, one can also assert properties on this scaffolding node.
In order to make verification of properties on this scaffolding node easier, the scaffolding node remains unchanged in the data structure, and is never deleted. We start with a data structure containing a pointer $next: C \rightarrow C$ and a monadic map $prev: C \rightarrow C$. We build on this data structure to define circular lists by adding a monadic map $last: C \rightarrow C$ where $last(u)$ for any location $u$  points to the last item in the list, which is the scaffolding node in this case. The scaffolding node $x$ must in turn point to another node whose $last$ map points to $x$ itself: this ensures cyclicity. We also define monadic maps $length: C \rightarrow Nat$ and $rev\_length: C \rightarrow Nat$ to denote the distance to the $last$ node by following $prev$ or $next$ pointers. The partial local conditions for $x$ are as below:
        {\small
        \begin{equation*}
        \begin{aligned}
        &(x=last(x) \Rightarrow last(\nxt(x))=x \land length(x)=0 \land rev\_length(x)=0) \\
        \land~ &(x\neq last(x) \Rightarrow last(\nxt(x))=last(x) \land length(x)=length(\nxt(x))+1 \\
        &\qquad\qquad\quad \land rev\_length(x)=rev\_length(\prev(x))+1) \\
        \end{aligned}
        \end{equation*}
        }

Here is the gist of inserting a node at the back of a circular list. We are given a node $x$ such that $\nxt(x) = last(x)$ (at the end of a cycle). We insert a newly allocated node after $x$, making local repairs there. Then, in a ghost loop similar to the one in Case Study~\ref{sec:cs-slist-reverse}, we make appropriate updates to the $\len$ and $\keys$ maps, which are not fully described here, following the $prev$ map until we reach $last(x)$. 
Like in the previous case study, we present the full local condition and code in our supplementary material in Appendix {\ref{app:cs-cl-insertback}}.

\subsection{Overlaid Data Structure of List and BST}
\label{sec:cs-overlaid}

One of the settings where intrinsic definitions shine is in defining and manipulating an \emph{overlaid data structure} that overlays a linked list and a binary search tree. The list and tree share the same locations, and the $\textit{next}$ pointer threads them into a linked list while the $\textit{left, right}$ pointers on them defines a BST. Such structures are often used in systems code (such as the Linux kernel) to save space~\cite{overlaidpaper}. 
For example, I/O schedulers use an overlaid structure as above, where the list/queue stores requests in FIFO order while the bst enables faster searching according requests with respect to a key. 
While there has been work in verification of memory safety of such structures~\cite{overlaidpaper}, we aim here to check preservation of such data structures.

Intrinsic definition over such an overlaid data structure is pleasantly \emph{compositional}. We simply take intrinsic definitions for lists and trees, and take the union of the monadic maps and the conjunction of their local conditions.
The only thing that is left is then to ensure that they contain the same set of locations. 
We introduce a monadic map $\textit{bst}\_\textit{root}$ that maps every node to its root in the bst, and introduce a monadic map $\textit{list}\_\textit{head}$ that maps every node to the head of the list it belong to (using appropriate local conditions). 
We then demand that all locations in a list have the same $\textit{bst}\_\textit{root}$ and all locations in a tree have the same $\textit{list}\_\textit{head}$, using local conditions. We also define monadic maps that define the bst-heaplet for tree nodes and list-heaplet for list nodes (the locations that belong to the tree under the node or the list from that node, respectively) using local conditions.
We define a correlation predicate $\textit{Valid}$ that relates the head $h$ of the list and root $r$ of the tree
by demanding that the bst-root of $h$ is $r$ and the list-head of $r$ is $h$, and furthermore, the list-heaplet of $h$ and tree-heaplet of $r$ are equal. This predicate can be seen here:
\begin{center}
$\textit{Valid} \equiv\; \textit{bst}\_\textit{root}(h) = r \wedge
    \textit{list}\_\textit{root}(r) = h 
   \land~ \textit{list}\_\textit{heaplet}(h) = \textit{bst}\_\textit{heaplet}(r)$
\end{center}

We prove certain methods manipulating this overlaid structure correct (such as deleting the first element of the list and removing it both from the list as well as the BST). These ghost annotations are mostly compositional--- with exceptions for fields whose mutation impacts the local condition of multiple components, we fix monadic maps for the BST component in the same way we fix them for stand-alone BSTs and fix monadic maps for the list component in the same way we fix them for stand-alone lists. In fact, we maintain two broken sets, one for BST and one for list, as updating a pointer for BST often does not break the local property for lists, and vice versa.

\medskip
\mypara{Limitations} In modeling the data structures above, we crucially used the fact that for any location, there is at most one location (or a bounded number of locations) that has a field pointing to this location. We used this fact to define an inverse pointer (\emph{prev} or \emph{parent}/\emph{p}), which allows us to capture the impact set when a location's fields are mutated. Consequently, we do not know how to model structures where locations can have unbounded indegree. We could model these inverse pointers using a sequence/array of pointers, but verification may get more challenging. Data structures with unbounded outdegree can however be modeled using just a linked-list of pointers and hence seen as a structure with bounded outdegree.

\section{Implementation and Evaluation}

\subsection{Implementation Strategy of IDS and FWYB in \bfBoogie}
\label{sec:implementation}

We implement the technique of intrinsically defined data structures and FWYB verification in the program verifier \Boogie~\cite{boogie}.
{\sc Boogie} is a low-level imperative programming language which supports systematic generation of verification conditions that are checked using SMT solvers.

We choose {\sc Boogie} as it is a low-level verification condition generator. We expect that scalar programs with quantifier-free specifications, annotations, and invariants, and given our careful modeling of the heap and its modification across function calls (Section~{\ref{sec:qfree-vcs}}), reduces to quantifier-free verification conditions that fall into decidable logics. We further cross-check that our encodings indeed generate decidable queries by checking the generated SMT files. Furthermore, a plethora of higher-level languages compile to {\sc Boogie} (e.g., VCC and Havoc for C~\cite{vcc,havoc}, {\sc Dafny}~\cite{dafny} with compilation to .NET, Civl for concurrent programs~\cite{civl}, Move for smart contracts~\cite{move}, etc.). Implementing a technique in {\sc Boogie} hence shows a pathway for implementing IDS and FWYB for higher-level languages as well.

\mypara{Modeling Fix-What-You-Break Verification in Boogie}
We model heaps in {\sc Boogie} by having a sort $\mathit{Loc}$ of locations and modeling pointers as maps from $\mathit{Loc}$ to sorts. We implement monadic maps also as maps from locations to field values. We implement our benchmarks using the \emph{macros} for well-behaved programming defined in Section~\ref{sec:cs-slist-insert}. %
We implement allocation with an $\Alloc$ set and heap change across function calls using parameterized map updates as described in
\iflong
Section~\ref{sec:qfree-vcs}. 
\else
Section~\ref{sec:qfree-vcs} and our supplementary material in Appendix~\ref{app:qfree-vcs}. 
\fi

We ensure that the VCs generated by {\sc Boogie} fall into decidable fragments, and there are several components that ensure this.
First, note that all specifications (contracts and invariants) are quantifier-free. Second, pure functions (used to implement local conditions) are typically encoded using quantification, %
but we ensure {\sc Boogie} inlines them to avoid quantification. Third, heap updates that are the effect of procedures and set operations for set-valued monadic maps %
are modeled using parameterized map updates~{\cite{pointwisearrays}}, which {\sc Boogie} supports natively. Finally, we cross-check that the generated SMT query is quantifier-free and decidable by checking the absence of statements that introduce quantified reasoning, including {\texttt{exists}}, {\texttt{forall}}, and {\texttt{lambda}}.

\subsection{Benchmarks}
We evaluate our technique on a variety of data structures and methods that manipulate them. Our benchmark suite consists of 
data structure manipulation methods for a variety of different list and tree data structures, including sorted lists, circular lists, binary search trees, and balanced binary search trees such as Red-Black trees and AVL trees. Methods include 
core functionality such as search, insertion and deletion. The suite includes an \emph{overlaid} data structure
that overlays a binary search tree and a linked list, implementing methods needed by a simplified version of the Linux deadline IO scheduler~\cite{overlaidpaper}. %
The contracts for these functions are complete functional specifications that not only ask for maintenance of the data structure, but correctness properties involving the returned values, the keys stored in the container, and the heaplet of the data structure.

\subsection{Evaluation}
\begin{table}[t]
\caption{Implementation and verification of {\sc Boogie} programs on the benchmarks. The columns give data structure, size of local conditions for capturing the datatructure as number of conjuncts, method, lines of executable code in the method, lines of specification (pre/post), lines of ghost code annotations (invariants/monadic map updates), and verification time in seconds.~\label{tbl:results}}
\vspace{-0.5em}
\begin{center}
{\footnotesize
\begin{tabular}{ rr | rrr | rrr }
 \multirow{2}{*}{Data Structure} & $LC$ & \multirow{2}{*}{Method} & LOC+Spec & Verif. & \multirow{2}{*}{Method} & LOC+Spec & Verif.\\
 & Size &  & +Ann & Time(s) & & +Ann & Time(s)\\
\hline
\multirow{4}{*}{Singly-Linked List} & \multirow{4}{*}{8} & Append & 4+11+10 & 2.0 & Insert-Back & 6+13+12 & 2.0 \\
& & Copy-All & 7+8+9 & 2.0 & Insert-Front & 3+13+7 & 2.0\\
& & Delete-All & 10+9+16 & 2.0 & Insert & 9+13+23 & 2.0\\
& & Find & 4+4+2 & 1.9 & Reverse & 6+8+18 & 2.1\\
\hline
\multirow{3}{*}{Sorted List} & \multirow{3}{*}{14} & Delete-All & 10+9+16 & 2.1 & Merge & 11+9+20 & 2.1 \\
& & Find & 4+4+2 & 1.9 & Reverse & 5+14+22 & 2.1 \\
& & Insert & 9+16+27 & 2.1 & &  & \\
\hline
\multirow{2}{*}{\shortstack[r]{Sorted List\\(w. $min, max$ maps)}} & \multirow{2}{*}{20} & \multirow{2}{*}{Concatenate} & \multirow{2}{*}{6+10+13} & \multirow{2}{*}{2.2} & \multirow{2}{*}{Find-Last} & \multirow{2}{*}{5+10+9} & \multirow{2}{*}{2.0} \\
& & & & & & & \\
\hline
\multirow{2}{*}{Circular List} & \multirow{2}{*}{27} & Insert-Front & 4+12+41 & 2.3 & Delete-Front & 3+12+39 & 2.4 \\
& & Insert-Back & 5+14+45 & 2.4 & Delete-Back & 3+13+55 & 2.4\\
\hline
\multirow{2}{*}{Binary Search Tree} & \multirow{2}{*}{35} & Find & 4+3+5 & 2.0 & Delete & 10+13+30 & 2.8 \\
& & Insert & 9+12+37 & 2.7 & Remove-Root & 17+15+47 & 3.8\\
\hline
\multirow{2}{*}{Treap} & \multirow{2}{*}{37} & Find & 4+3+5 & 2.0 & Delete & 10+13+30 & 3.1 \\
& & Insert & 19+12+74 & 10.2 & Remove-Root & 24+15+74 & 5.4\\
\hline
\multirow{2}{*}{AVL Tree} & \multirow{2}{*}{45} & Insert & 12+12+36 & 5.1 & Find-Min & 5+5+8 & 2.1 \\
& & Delete & 43+13+62 & 5.3 & Balance & 40+17+95 & 5.0\\
\hline
\multirow{3}{*}{Red-Black Tree} & \multirow{3}{*}{48} & Insert & 76+12+203 & 74.1 & Del-L-Fixup & 33+20+93 & 8.9 \\
& & Delete & 56+13+76 & 5.8 & Del-R-Fixup & 33+20+93 & 7.4 \\
& & Find-Min & 5+5+8 & 2.1 & &  & \\
\hline
\multirow{1}{*}{BST+Scaffolding} & \multirow{1}{*}{59} & Delete-Inside & 1+24+51 & 4.8 & Remove-Root & 44+31+61 & 10.2 \\
\hline
\multirow{2}{*}{\shortstack[r]{Scheduler Queue\\(overlaid SLL+BST)}} & \multirow{2}{*}{72} & Move-Request & 4+10+8 & 2.9 & BST-Delete-Inside & 1+29+55 & 4.9 \\
& & List-Remove-First & 5+13+10 & 2.7 & BST-Remove-Root & 44+36+65 & 15.0\\
\hline
\end{tabular}

\iflong
\begin{tabular}{ rrrrrrrr } 
\multirow{2}{*}{Data Structure} & Size of & \multirow{2}{*}{Method} & Lines of & Lines of & Lines of & Verification\\ 
& $LC$s & & Code & Specification & Annotation &  Time (s) \\ 
\hline
\multirow{8}{*}{Singly-Linked List} & \multirow{8}{*}{9} & Append & 4 & 11 & 7 &  1.4\\
 & & Copy-All & 5 & 8 & 7 & 1.3\\
 & & Delete-All & 8 & 9 & 13 & 1.4\\
 & & Find & 3 & 4 & 1 & 1.2\\
 & & Insert-Back & 5 & 13 & 10 & 1.3\\
 & & Insert-Front & 3 & 13 & 5 & 1.3\\
 & & Insert & 8 & 13 & 19 & 1.4 \\
 & & Reverse & 5 & 8 & 16 & 1.4 \\
\hline
\multirow{5}{*}{Sorted List} & \multirow{5}{*}{14} &  Delete-All & 8 & 9 & 14 & 1.5\\
 & & Find & 3 & 4 & 1 & 1.3\\
 & & Insert & 8 & 16 & 23 & 1.5\\
 & & Merge & 10 & 9 & 16 & 1.6\\
 & & Reverse & 5 & 14 & 8 & 1.6\\
\hline
\multirow{2}{*}{\shortstack[r]{Sorted List\\(w. $min, max$ maps)}} & \multirow{2}{*}{20} &  Concatenate & 6 & 10 & 10 & 1.7\\
 & & Find-Last & 3 & 10 & 8 & 1.4\\
\hline
\multirow{4}{*}{Circular List} & \multirow{4}{*}{27} & Insert-Front & 3 & 12 & 28 & 1.8 \\
 & & Insert-Back & 4 & 14 & 32 & 1.8 \\
 & & Delete-Front & 2 & 12 & 29 & 1.9 \\
 & & Delete-Back & 2 & 13 & 42 & 1.9 \\
\hline
\multirow{4}{*}{Binary Search Tree} & \multirow{4}{*}{35} & Find & 3 & 3 & 4 & 1.4 \\
 & & Insert & 8 & 12 & 31 & 2.0 \\
 & & Delete & 9 & 13 & 24 & 2.2\\
 & & Remove-Root & 15 & 15 & 39 & 2.9 \\
\hline
\multirow{4}{*}{Treap} & \multirow{4}{*}{37} & Find & 3 & 3 & 4 & 1.4 \\
 & & Insert & 14 & 12 & 61 & 6.4\\
 & & Delete & 9 & 13 & 24 & 2.3\\
 & & Remove-Root & 20 & 15 & 64 & 3.7 \\
\hline
\multirow{4}{*}{AVL Tree} & \multirow{4}{*}{45} & Insert & 11 & 12 & 29 & 3.0 \\
 & & Delete & 39 & 13 & 54 & 3.6 \\
 & & Find-Min & 2 & 5 & 7 & 1.5 \\
 & & Balance & 34 & 17 & 84 & 3.5 \\
\hline
\multirow{5}{*}{Red-Black Tree} & \multirow{5}{*}{48} & Insert & 46 & 12 & 155 & 27.6 \\
 & & Delete & 52 & 13 & 63 & 4.0 \\
 & & Find-Min & 2 & 5 & 7 & 1.6 \\
 & & Del-L-Fixup & 26 & 20 & 75 & 4.9 \\
 & & Del-R-Fixup & 26 & 20 & 75 & 4.8 \\
\hline
\multirow{2}{*}{BST+Scaffolding} & \multirow{2}{*}{59} & Delete-Inside & 1 & 24 & 35 & 3.3 \\
 & & Remove-Root & 32 & 31 & 53 & 6.6 \\
\hline
\multirow{4}{*}{\shortstack{Scheduler Queue\\(overlaid SLL+BST)}} & \multirow{4}{*}{72} & Move-Request & 4 & 10 & 5 & 2.7 \\
 & & List-Remove-First & 5 & 13 & 7 & 2.4 \\
 & & BST-Delete-Inside & 1 & 29 & 39 & 4.3 \\
 & & BST-Remove-Root & 32 & 36 & 55 & 9.9 \\
\hline
\end{tabular}
\fi
}
\end{center}
\vspace{-1em}
\end{table}

We first evaluate the following two research questions:

\noindent
{\bf RQ1: Can the data structures be expressed using IDS, and can the FWYB methodology for methods on these structures be expressed in \bfBoogie?}

\noindent
{\bf RQ2: Is \bfBoogie with decidable verification condition generation dispatched to SMT solvers effective in verifying these methods?}

As we have articulated earlier, intrinsic definitions and monadic map updates require a new way of thinking about programs and repairs. We implement the specifications using monadic maps and local conditions, and the benchmarks using the well-behavedness macros and ghost updates. We were able to express all data structures and FWYB annotations for the methods on these structures for our benchmarks in \Boogie (RQ1). Importantly, we were able to write quantifier-free modular contracts for the auxiliary methods and loop invariants using the monadic maps and strengthening the contracts using quantifier-free assertions on broken sets (which may not be empty for auxiliary methods). We do not prove termination for these methods except for ghost loops and ghost recursive procedures (termination for latter is required for soundness). We provide the benchmarks with annotations in a public repository \cite{idsartifact}. %

Our annotation measures and verification results are detailed in the table in Table~\ref{tbl:results}, for $42$ methods across $10$ data structure definitions. These measurements were taken from a machine with an Intel\texttrademark \ Core i5-4460 processor at 3.20 GHz. We found the verification performance excellent overall (RQ2): all the methods verify in under 2 minutes, and all but four verify in under 10 seconds. We used the option that sets the maximum number of VC splits to $8$ in {\sc Boogie}. The times reported for each method are the sum of times taken for the following steps: verifying that the impact sets are correct (<3s for all data structures), generating verification conditions with {\sc Boogie}, injecting parametric update implementations, and solving the SMT queries.

Notice that the lines of ghost code written is nontrivial, but these are typically simple, involving programmatically repairing monadic maps and manipulating broken sets. 
In fact, a large fraction ($\sim 60\%$) of ghost updates in our benchmarks were \emph{definitional updates} that simply update a field according to its definition in the local condition. An example is updating $x.\textit{length}$ to $x.\textit{next}.\textit{length}+1$ for lists. We believe that the annotation burden can be significantly lowered in future work by automating such updates. 
More importantly, note that none of the programs required further annotations like instantiations, triggers, inductive lemmas, etc. in order to prove them correct.

\medskip
\noindent
{\bf RQ3: What is the performance impact of generating decidable verification conditions?}

In order to study this, we implemented the entire benchmark suite described in Table {\ref{tbl:results}} in {\sc Dafny}, a higher-level programming language that uses {\sc Boogie} to perform its verification. We implemented the data structures and the FWYB methodology identically in {\sc Dafny} to the {\sc Boogie} version.

\begin{wrapfigure}{r}{0.45\textwidth}
\begin{center}
\vspace{-1.5em}
\includegraphics[width=0.45\textwidth]{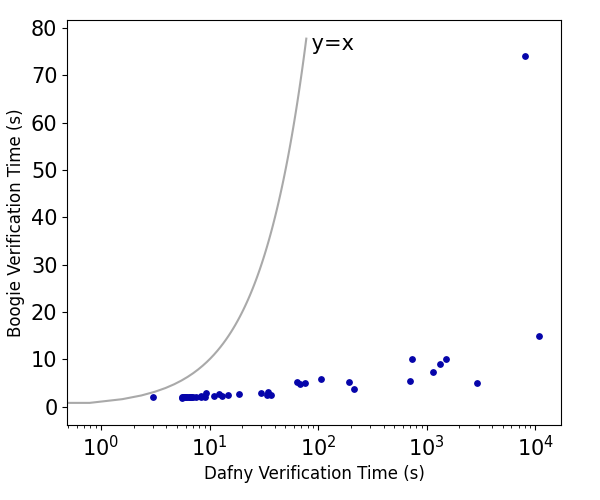}
\vspace{-2em}
\end{center}
\end{wrapfigure}

Even though our annotations are all quantifier-free, {\sc Dafny} generates {\sc Boogie} code where several aspects of the language, in particular allocation and heap change across function calls, are modeled using \emph{quantifiers}, resulting in quantified queries to SMT solvers. The scatter plot on the right shows the performance of {\sc Boogie} and {\sc Dafny} on the benchmarks. The plot clearly strongly suggests that even though {\sc Dafny} is able to prove the FWYB-annotated programs correct, using decidable verification conditions results in much better performance. We hence believe that implementing program verifiers (such as {\sc Dafny}) that exploit the fact that FWYB annotations can be compiled to annotations in {\sc Boogie} that result in decidable VCs is a promising future direction to achieve faster high-level IDS+FWYB frameworks.

\label{sec:evaluation}

\section{Related Work}
\label{sec:relwork}
There have been mainly two paradigms to automated verification of programs annotated with rich contracts written in logic. The first is to restrict the specification logic so that verification conditions fall into a decidable logic. The second allows validity of verification conditions to fall into an undecidable or even an incomplete logic (where validity is not even recursively enumerable), but support effective strategies nevertheless, using heuristics, lemma synthesis, and further annotations from the programmer~\cite{seplogicprimer,reynolds02,vcdryad,nguyen08,verifast,PiskacWiesZufferey2014,ChinDavidNguyen2007,ChuAutomaticInductionProofs,fossil,sls,BerdineCalcagnoOHearn2005,RegionLogic,RegionLogic1,RegionLogic2,boogie-meets-regions,Smallfoot,Distefano2006,CalcagnoBiAbductive2011}. In this paper, we have proposed a new paradigm of predictable verification that calls for programmers to write a reasonable amount of extra annotations under which validity of verification conditions becomes decidable. To the best of our knowledge, we do not know of any other work of this style (where validity of verification conditions is undecidable but an upfront set of annotations renders it decidable). 

\mypara{Decidable Verification}
There is a rich body of research on decidable logics for heap verification: 
first-order logics with reachability~\cite{levami}, the logic {\sc Lisbq} in the {\sc Havoc} tool~\cite{backtothefuture}, several decidable fragments of separation logic known~\cite{decseplogicberdine,PiskacWiesZufferey2013} as well as fragments that admit a decidable entailment problem~\cite{echenim}. Decidable logics based on interpreting bounded treewidth data structures on trees have also been studied, for separation logics as well as other logics~\cite{strand1,strand2,bddtwseplogic}. In general, these logics are heavily restricted--- the magic wand in separation logic quickly leads to undecidability~\cite{almightywand}, the general entailment problem for separation logic with inductive predicates is undecidable~\cite{undecseplogic}, and validity of first-order logic with recursive definitions is undecidable and not even recursively enumerable and does not admit complete proof procedures.

\mypara{Validity Checking of Undecidable and Incomplete Logics}
 Heap verification using undecidable and incomplete logics has been extensively in the literature. The work on natural proofs~\cite{vcdryad,loding18} for imperative programs and work on Liquid Types \cite{liquidtypes08} for functional programs propose such approaches that utilize SMT solvers, but require extra user help in the form of inductive lemmas to verify programs. Users need to understand the underlying heuristic SMT encoding mechanisms and their shortcomings, as well as theoretical shortcomings (the difference between fixed point and least fixed point semantics of recursive definitions) in order to provide these lemmas~(see \cite{loding18,fossil,fluid}). In contrast, the user help we seek in this work is upfront ghost code that updates monadic maps to satisfy local conditions independently of the heuristics the solvers use. Furthermore, for programs with such annotations, we assure decidable validation of the associated verification conditions.

\mypara{Monadic Maps} 
Monadic maps have been exploited in earlier work in other forms for simplifying verification of properties of global structures. 
In shape analysis~\cite{shapeanalysis}, 
monadic predicates are often used to express inductively defined properties of single locations on the heap. %
In separation logic, the \emph{iterated separating conjunction operator}, introduced already by Reynolds in 2002~\cite{ReynoldsSepLogic02},  expresses local properties of each location, and is akin to monadic maps. Iterated separation conjunction has been used in verification, for both arrays as well as for  data structures, in various forms~\cite{Muller16ISC,jstar}.
The work on verification using flows~\cite{gowiththeflow,localreasoningforglobalproperties,flows-pldi20,flows-oopsla21,flows-oopsla22,flows-tacas23} introduces predicates based on flows, and utilizes such predicates in iterated separation formulas to express global properties of data structures and to verify algorithms such as the concurrent Harris list. In these works, local properties of locations and proof systems based on them are explored, but we do not know of any work exploiting monadic maps for decidable reasoning, which is crucial for predictable verification.
The work by Gyori et al~\cite{Gyori17} exploits a class of monadic maps called \emph{linear measures} that satisfy certain algebraic properties in order to incrementally maintain and check properties of linked lists at runtime.

\mypara{Ghost Code} The methodology of writing ghost code is a common paradigm in deductive program verification~\cite{Jonesghostcode,lucasghostcode,spiritofghostcode,ReynoldsCraftOfProgramming} and supported by verification tools such as {\sc Boogie} and {\sc Dafny}~\cite{boogie,dafny}. 
Ghost code involves code that manipulates auxiliary variables to perform a parallel computation with the original code without affecting it.
Our use of ghost code establishes the required monadic maps that satisfy local conditions by allowing the programmer to construct the maps and verify the local conditions using a disciplined programming methodology. Furthermore, we assure that the original code with the ghost code results in decidable verification problems, which is a salient feature not found typically in other contexts where ghost code is used.

\section{Conclusions}
\label{sec:conclusions}
We introduced intrinsic definitions
that eschew recursion/induction and instead define data structures using monadic maps and local conditions. Proving that a program maintains a valid data structure hence requires only maintaining monadic maps and verifying the local conditions %
on locations that get broken. 
Furthermore, verifying that engineer-provided ghost code annotations are indeed correct falls into decidable theories, leading to a predictable verification framework.

\mypara{Future Work} 
First, it would be useful to develop verification engines for higher-level languages (like Java~\cite{verifast}, Rust~\cite{verrus}, and Dafny~\cite{dafny}) that 
that have native support for intrinsic definitions and produce verification conditions in decidable theories that SMT solvers can handle (see RQ3 in Section~\ref{sec:evaluation}). Second, it would be interesting to see how intrinsic definitions with fix-what-you-break proof methodology can coexist and exchange information with traditional recursive definitions with induction-based proof methodology. %
Third, as mentioned in Section~\ref{sec:evaluation}, many updates of monadic maps are straightforward using definitions, and tools that automate this can reduce annotation burden significantly. 
Fourth, we are particularly intrigued with the ease with which intrinsic data structures capture more complex data structures such as overlaid data structures. Exploring intrinsic definitions for verifying concurrent and distributed programs that maintain data structures is particularly interesting. Fifth, intrinsic definitions opens up an entirely new approach to defining properties of structures that simplify reasoning. We believe that
exploiting intrinsic definitions in other verification contexts, like mathematical structures used in specifications (e.g., message queues in distributed programs), parameterized concurrent programs (configurations modeled as unbounded sequences of states), and programs that manipulate big data concurrently (like Apache Spark) are exciting future directions. Finally, it would be interesting to adapt IDS for functional programs. Since functional data structures are not mutable, ghost fields will always meet local conditions. However, we may need to \emph{(re-)establish} rather than \emph{repair} local conditions, which may require ghost code, e.g., establishing that the ghost map $\mathit{sorted}$ on a functional list $x$ is true.

\section*{Artifact Availability Statement}
We have prepared a publicly available artifact \cite{idsartifact} containing our benchmark suite and a Docker image for reproducing our evaluation in Section \ref{sec:evaluation}.

\section*{Acknowledgements}
We thank the anonymous reviewers for their valuable feedback. In particular, we thank one reviewer for pointing out that {\sc Boogie} supports parameterized map updates natively; using this feature simplifies the present version of our paper and artifact. This work is supported in part by a research grant from Amazon.

\bibliographystyle{ACM-Reference-Format}
\bibliography{refs}
\newpage
\appendix
\section{Details for Section~\ref{sec:proglogic}}
\label{app:proglogic}

\subsection{Operational Semantics}
\label{sec:opsem}

\begin{figure}[ht!]
\begin{align*}
\bot & \xrightarrow{*} \bot\\
(s, O, I) & \xrightarrow{\mathsf{skip}} (s, O, I)\\
(s, O, I) & \xrightarrow{x\, :=\, \nil} (s[x \mapsto \nil], O, I)\\
(s, O, I) & \xrightarrow{x\, :=\, y} (s[x \mapsto s(y)], O, I)\\
(s, O, I) & \xrightarrow{v\, :=\, be} (s[v \mapsto e], O, I) \hspace{3em} \text{where $be$ interprets to $e$ according to $s$ and $I$}\\
(s, O, I) & \xrightarrow{y\, :=\, x.f} (s[y \mapsto I(f,s(x))], O, I) \hspace{5pt}\text{ if } (f, s(x)) \in \mathit{dom}(I) \hspace{5em} \text{(similarly for $v\, :=\, x.d$)}\\
(s, O, I) & \xrightarrow{y\, :=\, x.f} \bot \hspace{5pt}\text{ if } (f, s(x)) \notin \mathit{dom}(I) \hspace{14em} \text{(similarly for $v\, :=\, x.d$)}\\
(s, O, I) & \xrightarrow{x.f\, :=\, y} (s, O, I[(f, s(x)) \mapsto s(y)]) \hspace{5pt}\text{ if } (f, s(x)) \in \mathit{dom}(I) \hspace{4em} \text{(similarly for $x.d\, :=\, v$)}\\
(s, O, I) & \xrightarrow{x.f\, :=\, y} \bot \hspace{5pt}\text{ if } (f, s(x)) \notin \mathit{dom}(I) \hspace{14em} \text{(similarly for $x.d\, :=\, v$)}\\
(s, O, I) & \xrightarrow{x\, :=\, \mathsf{new}\; C()} (s[x \mapsto o], O \cup \{o\}, I[(f, o) \mapsto \mathit{default}_f])\\
&\hspace{6em} \text{for some $o \in \mathbb{N}$ such that $o\notin O$}\\
(s, O, I) & \xrightarrow{\overline{r}\, :=\, \mathit{Function}(\overline{t})} (s[\overline{r} \mapsto s'(\overline{n})], O', I') \hspace{5pt}\text{ if } (\emptyset[\overline{m} \mapsto s(\overline{t})], O, I) \xrightarrow{Q(\overline{m}, \,\mathit{ret}:\, \overline{n})} (s', O', I')\\
& \hspace{6em} \text{ where $Q(\overline{m},\, \mathit{ret}:\, \overline{n})$ is the code of the method $\mathit{Function}$,}\\
&\hspace{6.3em} \text{with $\overline{m}$ and $\overline{n}$ being the formal input and output parameters for $Q$}\\
(s, O, I) & \xrightarrow{\mathsf{assume}\, \mathit{cond}} (s, O, I) \hspace{5pt}\text{ if }\text{$\mathit{cond}$ interprets to $\mathit{True}$ according to $s$ and $I$}\\
(s, O, I) & \xrightarrow{P_1\,;\, P_2} (s'', O'', I'') \hspace{5pt}\text{ if } (s, O, I) \xrightarrow{P_1} (s',O', I')\\
&\hspace{8.2em}\text{ and } (s', O', I') \xrightarrow{P_2} (s'',O'', I'') \text{ for some } (s', O', I')\\
(s, O, I) & \xrightarrow{\mathsf{if}\, \mathit{cond}\, \mathsf{then}\, P_1\, \mathsf{else}\, P_2} (s', O', I') \hspace{5pt}\text{ if } (s, O, I) \xrightarrow{\mathsf{assume}\, \mathit{cond}\,;\;P_1} (s', O', I')\\
(s, O, I) & \xrightarrow{\mathsf{if}\, \mathit{cond}\, \mathsf{then}\, P_1\, \mathsf{else}\, P_2} (s', O', I') \hspace{5pt}\text{ if } (s, O, I) \xrightarrow{\mathsf{assume}\, \neg\mathit{cond}\,;\;P_2} (s', O', I')\\
(s, O, I) & \xrightarrow{\mathsf{while}\, \mathsf{cond}\, \mathsf{do}\, P} (s', O', I') \hspace{5pt}\text{ if } (s, O, I) \xrightarrow{\mathsf{assume}\,\mathit{cond}\,;\,P\,;\,\mathsf{while}\, \mathsf{cond}\, \mathsf{do}\, P} (s', O', I')\\
(s, O, I) & \xrightarrow{\mathsf{while}\, \mathsf{cond}\, \mathsf{do}\, P} (s, O, I) \hspace{5pt}\text{ if } (s, O, I) \xrightarrow{\mathsf{assume}\,\neg\mathit{cond}} (s, O, I)\\
\end{align*}
\caption{Operational Semantics}
\label{fig:op-sem}
\end{figure}

We give the formal operational semantics for programs in our language (Figure~\ref{fig:prog-lang}) in Figure~\ref{fig:op-sem} below.

Configurations are of the form $(s, O , I)$ where $O \subset_\mathit{finite} \mathbb{N}$ represents the set of allocated objects, $s$ represents the store and interprets program variables, and $I$ represents the heap and interprets mutable fields in $\Ff$---including ghost fields $\Gg$ when they are used--- on $O$ (interpretations are total). Although formally $s$ and $I$ are a family of functions indexed by the sorts of the variables (resp. signatures of the maps), we abuse notation and use $s(x)$ to denote the interpretation of a variable $x$, and similarly $I(f, o)$ to denote the value of the field $f$ on the object $o$ in the configuration. We add a sink state $\bot$ to model error.

Our language is safe, (i.e., allocated locations cannot point to un-allocated locations) and garbage-collected. The operational semantics is the usual one for such programs. Figure~\ref{fig:op-sem} presents a simplified operational semantics without considering $\mathsf{return}$ statements. The full semantics adds a marker to signify completion of a procedure. Procedures can only end after $\mathsf{return}$ statements (we syntactically disallow statements after a $\mathsf{return}$) or at the end of a program.

The rules for assignments, skip, sequencing, conditionals, and loops are trivial. De-referencing a variable that does not point to an object (i.e., is $\nil$) leads to the error state $\bot$. Allocation ensures memory safety by assigning the value of a field $f$ on a newly allocated object to a constant $\mathit{default}_f$. For pointer fields this value is $\nil$. Finally, we define the operational semantics for function calls using summaries.

\subsection{Ghost Code}
\label{app:ghost-code-defn}

\begin{figure}[h!]
\begin{align*}
P \coloneqq &\;\; x\, :=\, \mathit{Expr}[\mathit{Var}_U, \Ff] \;\;|\;\; y\, :=\, x.f \;\;|\;\; x.f\, :=\, y \;\; |\; \;z\, :=\, \mathsf{new}\; C() \\[-3pt]
&\;\;|\; \;\overline{r}\, :=\, \mathit{Func}(\overline{t}) \hspace{3em} \overline{r},\overline{t} \text{ are variables in } \mathit{Var}_U \cup \mathit{Var}_G \\[-3pt]
&\;\;\; \text{(Functions can have ghost input/output parameters)} \\[-3pt]
&\;\;|\; \; GP \\[-3pt]
&\;\;\; \text{(GP are ``pure'' ghost programs)} \\[-3pt]
&\;\; |\;\; \mathsf{skip}\;\;|\;\; \mathsf{assume}\;\mathit{cond}\;\;|\;\; \mathsf{return} \\[-3pt]
&\;\;|\;\; P\, ; \, P \;\;|\;\; \mathsf{if}\; \mathit{cond}\; \mathsf{then}\; P\; \mathsf{else}\; P\;\; |\;\; \mathsf{while}\; \mathsf{cond}\; \mathsf{do}\; P\; \\[-1pt]
\mathit{cond} \coloneqq &\;\; \mathit{BoolExpr}[\mathit{Var}_U,\Ff] \\[3pt]
GP \coloneqq &\;\; a\, :=\, \mathit{Expr}[\mathit{Var}_U \cup \mathit{Var}_G, \Ff \cup \Gg] \;\;|\;\; b\, :=\, x.g \;\;|\;\; b\, :=\, x.f \\[-3pt]
&\;\;\; \text{(Ghost variables can read from both user and ghost variables/maps)} \\[-3pt]
&\;\;|\;\; x.g\, :=\, b \;\;|\;\; x.g\, :=\, y \\[-3pt]
&\;\;\; \text{(Ghost maps can only be assigned values from ghost variables)} \\[-3pt]
&\;\;|\; \;\overline{s}\, :=\, \mathit{GhostFunc}(\overline{v}) \hspace{2em} \overline{s},\overline{v} \text{ are variables in } \mathit{Var}_G,\, \mathit{GhostFunc} 
\text{ is \textbf{always terminating}}\\[-3pt]
&\;\; |\;\; \mathsf{skip}\;\;|\;\; GP\, ; \, GP \;\;|\;\; \mathsf{if}\; \mathit{Gcond}\; \mathsf{then}\; GP\; \mathsf{else}\; GP \\[-3pt]
&\;\; |\;\; \mathsf{while}\; \mathsf{Gcond}\; \mathsf{do}\; GP\; \hspace{1em} \text{loop is \textbf{always terminating}} \\[-1pt]
\mathit{Gcond} \coloneqq &\;\; \mathit{BoolExpr}[\mathit{Var}_U \cup \mathit{Var}_G,\Ff \cup \Gg]
\end{align*}
\caption{Grammar of programs with ghost code. $x,y, z$ are user variables $\mathit{Var}_U$, $a,b$ are ghost variables $\mathit{Var}_G$, $f \in \Ff$ is a user field, and $g \in \Gg$ is a ghost map. Notation $\mathit{Expr}[\mathit{Vars},\mathit{Maps}]$ denotes expressions over the vocabulary given by variables $\mathit{Vars}$ and maps $\mathit{Maps}$, similarly $\mathit{BoolExpr}[\mathit{Vars},\mathit{Maps}]$ denotes boolean expressions. Termination for ghost loops and functions can be established in any way.}
\label{fig:lang-with-ghost}
\end{figure}

In this section we formally define our programming language augmented with ghost code, as well as the projection of ghost-augmented code to `user' code.

Fix a set of user variables $\mathit{Var}_U$ and ghost variables $\mathit{Var}_G$. We already introduced user fields $\Ff$ and ghost fields/maps $\Gg$ in Section~\ref{sec:ids}. We define a programming language over this vocabulary in Figure~\ref{fig:lang-with-ghost} below. The main aspects to note are: (a) ghost variables can read from user variables/maps, but the reverse is not allowed, (b) ghost conditionals and loops must only contain bodies that are purely ghost, and (c) ghost loops and functions must always terminate. These choices ensure that ghost variables do not affect the execution of the user program. We can formalize this claim using the idea of `projecting out' ghost code and obtaining a pure user program

\mypara{Projection that Eliminates Ghost Code} Fix a main method $M$ with body $Q_0$. Let $N_i, 1 \leq i \leq k$ be a set of auxiliary methods with bodies $Q_i$ that $Q_0$ can call. Note that the bodies $Q_0$ and $Q_i$ contain ghost code. Let us denote a program containing these methods by $[(M: Q_0); (N_1: Q_1)\ldots (N_k: Q_k)]$. We then define projection as follows:

\begin{definition}[Projection of Ghost-Augmented Code to User Code]
\label{defn:code-projection-formal}
The projection of the ghost-augmented program $[(M:Q_0); (N_1:Q_1)\ldots (N_k:Q_k)]$ is the user program $[(\hat{M}:\hat{Q_0}); (\hat{N_1}:\hat{Q_1})\ldots (\hat{N_k}:\hat{Q_k})]$ such that:
\begin{enumerate}
    \item The input (resp. output) signature of $\hat{M}$ is that of $M$ with the ghost input (resp. output) parameters removed. Formally, given a sequence of parameters $\overline{x}$ with some elements in the sequence marked as $\mathsf{ghost}$, we can define the projection as the sequence formed by the non-continguous subsequence of parameters in $\overline{x}$ consisting of non-ghost parameters.
    \item $\hat{Q_0}$ is derived from $Q_0$ by: (a) eliminating all ghost code, i.e., replacing yields of the nonterminal $\mathit{GP}$ in Figure~\ref{fig:lang-with-ghost} with $\mathsf{skip}$, and (b) replacing each non-ghost function call statement of the form $\overline{r} := N_j(\overline{t})$ with the statement $\overline{s} := \hat{N_j}(\overline{u})$, where $\overline{u}, \overline{s}$ are obtained from $\overline{t}, \overline{r}$ by projecting out the elements corresponding to the ghost parameters in the signature of $N_j$. Each $\hat{Q_i}$ is derived from the corresponding $Q_i$ by a similar transformation.

    We define this formally as a recursive transformation on the structure of the grammar of $P$ (ghost-augmented programs) in Figure~\ref{fig:lang-with-ghost}:
    \begin{align*}
        \mathit{Projection}(GP) &=\; \mathsf{skip} \\
        \mathit{Projection}(\overline{r} \coloneqq \mathit{Func}(\overline{t})) &=\; \overline{s} \coloneqq \hat{\mathit{Func}}(\overline{u})\\
        &\hspace{1.5em}\text{$\overline{u}, \overline{s}$ are obtained from $\overline{t}, \overline{r}$ by projecting out}\\
        &\hspace{1.5em}\text{elements corresponding to ghost parameters}\\
        \mathit{Projection}(\mathsf{stmt}) &=\; \mathsf{stmt} \hspace{1em}\text{for all other statements}\\
        \mathit{Projection}(P_1; \, P_2) &=\; \mathit{Projection}(P_1)\,; \, \mathit{Projection}(P_2)\\
        \mathit{Projection}(\mathsf{if}\; \mathit{cond}\; \mathsf{then}\; P_1\; \mathsf{else}\; P_2) &=\; \mathsf{if}\; \mathit{cond}\; \mathsf{then}\; \mathit{Projection}(P_1)\; \mathsf{else}\; \mathit{Projection}(P_2)\\
        \mathit{Projection}(\mathsf{while}\; \mathsf{cond}\; \mathsf{do}\; P) &=\; \mathsf{while}\; \mathsf{cond}\; \mathsf{do}\; \mathit{Projection}(P)
    \end{align*}    
\end{enumerate}
\end{definition}

\iflong\else %

\fi %

\section{Proofs of Soundness for Stages 1, 2, and 3 of FWYB}
\label{app:stages-soundness}

In this section we detail the proofs of soundness for the various stages of the FWYB methodology. We first introduce some notation and show some preliminary lemmas.

\mypara{Projection for Configurations} The stages of FWYB deal with two kinds of triples, one whose validity is stated with respect to configurations that interpret ghost variables and maps, and one over configurations that only interpret user variables and fields. Given a configuration $C$ that interprets ghost variables/maps, we denote by $\hat{C}$ the projection of that configuration to user variables that simply eliminates all ghost interpretations. Conversely, given a configuration $c$ we say that $C$ extends $c$ with an interpretation for ghost variables/maps if $\hat{C} = c$. We define $\hat{bot} = \bot$.

\smallskip
\mypara{Lemmas About Projection that Eliminates Ghost Code} We show the following lemmas about projection that eliminates ghost code (Definition~\ref{defn:code-projection-formal}). We assume that there is only one procedure $M$ in the program for simplicity of presentation. Recall that $M$ can contain ghost code and $\hat{M}$ is the projection of $M$ that eliminates the ghost code (with appropriately modified input/output parameters).

\begin{lemma}
\label{lem:ghost-code-termination-basic}
Let $C_1$ be a configuration that interprets ghost variables/maps. If $M$ is a ``pure ghost'' program, (i.e., a yield of $\mathit{GP}$ in the grammar in Figure~\ref{fig:lang-with-ghost}), then $M$ always terminates starting from $C_1$. 
\end{lemma}

The above lemma says that pure ghost programs always terminate. It follows directly from the definition of ghost code which requires pure ghost loops and functions to be terminating.\qed 

\begin{lemma}
\label{lem:ghost-code-termination-advanced}
Let $c$ be a configuration that does not interpret ghost variables/maps. If $\hat{M}$ (projected code that does not contain ghost code) terminates starting from $c$, then $M$ (which contains additional ghost code) must terminate starting from any configuration $C$ that extends $c$.
\end{lemma}

The above lemma says that the termination of the original user program is preserved by any augmentation with ghost code. In a certain sense, it `lifts' Lemma~\ref{lem:ghost-code-termination-basic} to programs that contain both user and ghost code.

\begin{proof}
The lemma follows from structural induction on the definition of projection, i.e., on the structure of the grammar for the nonterminal $P$ in Figure~\ref{fig:lang-with-ghost}. The argument for basic statements is trivial. For pure ghost programs the result follows from Lemma~\ref{lem:ghost-code-termination-basic}. The argument for all compositions (sequential, conditional, loop) and function calls follows from the induction hypothesis.
\end{proof}

\smallskip

\begin{lemma}
\label{lem:user-config-unaffected-basic}
Let $C_1$ be a configuration that interprets ghost variables/maps. If $M$ is a ``pure ghost'' program and $M$ starting from $C_1$ reaches some $C_2$ and $C_2 \neq \bot$, then $\hat{C_1} = \hat{C_2}$. 
\end{lemma}

The above lemma says that ghost code does not affect the values of user (non-ghost) variables and maps. It follows trivially by structural induction on the $\mathit{GP}$ grammar, using the definition of operational semantics (Figure~\ref{fig:op-sem}). The key observation is that $\mathit{GP}$ syntactically disallows non-ghost variables/maps to read from ghost variables/maps.\qed

We can similarly `lift' the above lemma to programs that contain both user code and ghost code.

\begin{lemma}
\label{lem:user-config-unaffected-advanced}
Let $c$ be a configuration that does not interpret ghost variables/maps. If $\hat{M}$ starting from $c_1$ reaches some $c_2$, then $M$ starting from any configuration $C_1$ that extends $c_1$ must either reach $\bot$ or some $C_2$ that extends $c_2$.
\end{lemma}

The above lemma says that augmentation with ghost code does not affect how the original program executes. 

\begin{proof}
As with Lemma~\ref{lem:ghost-code-termination-advanced}, we proceed by structural induction on the grammar for $P$ in Figure~\ref{fig:lang-with-ghost}. The argument for basic non-ghost statements follows trivially from the definition of operational semantics. They key observation is that non-ghost statements do not affect the values of ghost variables/maps (ensured by the syntactic restrictions). For pure ghost programs the result follows from Lemma~\ref{lem:user-config-unaffected-basic}. The argument for all compositions (sequential, conditional, loop) and function calls follows from the induction hypothesis.
\end{proof}

\section*{Proof of Proposition~\ref{prop:ghost-code}}

We can state the proposition simply as follows: if $\langle\, \LC \land \psi_{\mathit{pre}} \,\rangle\; M\; \langle\, \LC \land \psi_{\mathit{post}} \,\rangle$ is valid, then $\langle\, \exists g_1, g_2\ldots , g_k.\, \LC \land \psi_{\mathit{pre}} \,\rangle\; \hat{M}\; \langle\, \exists g_1, g_2\ldots , g_k.\,\LC \land \psi_{\mathit{post}} \,\rangle$ is valid.

Fix configurations (without ghost state) $c_1, c_2$ such that $c_1$ satisfies $\exists g_1, g_2\ldots , g_k.\, \LC \land \psi_{\mathit{pre}}$ and $\hat{M}$ starting from $c_1$ reaches $c_2$. To show that the given Hoare triple for $\hat{M}$ is valid, we must establish that $c_2$ is not $\bot$, and further that $c_2$ satisfies $\exists g_1, g_2\ldots , g_k.\, \LC \land \psi_{\mathit{post}}$.

Since $c_1 \models \exists g_1, g_2\ldots , g_k.\, \LC \land \psi_{\mathit{pre}}$, by the semantics of second order logic there exists a configuration (taken as a model) extending $c_1$, say $C_1$, such that $C_1 \models \LC \land \psi_{\mathit{pre}}$. First, using Lemma~\ref{lem:ghost-code-termination-advanced} we have that $M$ starting from $C_1$ must terminate. Further, since the triple $\langle\, \LC \land \psi_{\mathit{pre}} \,\rangle\; M\; \langle\, \LC \land \psi_{\mathit{post}} \,\rangle$ is valid, it must be the case that $M$ starting from $C_1$ reaches some $C_2$ such that $C_2 \neq \bot$ and $C_2 \models \LC \land \psi_{\mathit{post}}$. 

We now use Lemma~\ref{lem:user-config-unaffected-advanced} to conclude that $\hat{C_2} = c_2$. Since $C_2 \neq \bot$, we have that $c_2 \neq \bot$. Further, since $C_2 \models \LC \land \psi_{\mathit{post}}$, we have from the semantics of the logic that $C_2 \models \exists g_1, g_2\ldots , g_k.\, \LC \land \psi_{\mathit{post}}$. 

Observe that the formula $\exists g_1, g_2\ldots , g_k.\, \LC \land \psi_{\mathit{post}}$ is stated over the common vocabulary of $C_2$ and $c_2$, where the interpretations of the two configurations agree. Therefore, we can conclude that $c_2 \models \exists g_1, g_2\ldots , g_k.\, \LC \land \psi_{\mathit{post}}$. This concludes the proof.\qed

\subsection*{Proof of Proposition~\ref{prop:ghost-code-Br}}

The proof of Proposition~\ref{prop:ghost-code-Br} is similar to the above proof for Proposition~\ref{prop:ghost-code}, except that we must now consider a definition of ghost code (akin to the development in Section~\ref{app:ghost-code-defn}) that only considers the variable $\Br$ as ghost. 

Repeating the arguments in the proof of Proposition~\ref{prop:ghost-code} appropriately, we obtain that if 

\begin{center}
$\langle\, (\forall z \notin \Br.\, \rho(z)) \land \alpha \land \Br = \emptyset\, \rangle\; P_{\Gg,\Br}(\overline{x}, \Br,\,\mathit{ret}\!:\,\overline{y}, \Br)\; \langle\, (\forall z \notin \Br.\, \rho(z)) \land \beta \land \Br = \emptyset\, \rangle$
\end{center}

\noindent
is valid, then

$\langle\, \exists \Br.\,\big((\forall z \notin \Br.\, \rho(z)) \land \alpha \land \Br = \emptyset\big)\, \rangle\; P_{\Gg}(\overline{x},\,\mathit{ret}\!:\,\overline{y})\; \langle\, \exists \Br.\,\big((\forall z \notin \Br.\, \rho(z)) \land \beta \land \Br = \emptyset\big)\, \rangle$

\noindent
is valid. This triple can be simplified to $\langle\, (\forall z.\, \rho(z)) \land \alpha\, \rangle\; P_\Gg(\overline{x},\,\mathit{ret}\!:\,\overline{y})\; \langle\, (\forall z.\, \rho(z)) \land \beta\, \rangle$, which concludes the proof.

\section*{Proof of Proposition~\ref{prop:well-behaved-sound}}

Given a well-behaved program $P$ such that $\{\alpha\}\; P\; \{\beta\}$ is valid, we must show that $\langle\, (\forall z \notin \Br.\, \rho(z)) \land \alpha\, \rangle\; P\; \langle\, (\forall z \notin \Br.\, \rho(z)) \land \beta\, \rangle$ is valid.

The proof proceeds by an induction on the nesting depth of method calls in a trace of the program $P$. We elide this level of induction here because it is routine. Importantly, given a particular execution of the program $P$, we must show that the claim holds, assuming it holds for all method calls occurring in the execution. We show this by structural induction on the proof of well-behavedness of $P$.

There are several base cases. 

\textsc{Skip/Assignment/Lookup/Return}\hspace{1em} There is nothing to show for skip, assignment, lookup, or return statements. These do not change the heap at all and the rule does not update $\Br$ either, therefore if $\langle\,\alpha\,\rangle\;\mathsf{stmt}\;\langle\,\beta\,\rangle$ is valid then certainly $\langle\,(\forall z\notin\Br.\,\rho) \land \alpha\,\rangle\;\mathsf{stmt}\;\langle\,(\forall z\notin\Br.\,\rho) \land\beta\,\rangle$ is valid.

\textsc{Mutation}\hspace{1em} The claim is true for the mutation rule since by the premise of the rule we update the broken set with the impact set consisting of all potential objects where local conditions may not hold.

\textsc{Function Call}\hspace{1em} Here we simply appeal to the induction hypothesis.

\textsc{Allocation}\hspace{1em} We refer to our operational semantics, which ensures that no object points to a freshly allocated object. Therefore, the allocation of an object could have only broken the local conditions on itself at most.

\textsc{Infer LC Outside Br}\hspace{1em} There is nothing to prove for this rule as it does not alter the $\Br$ set at all.

\textsc{Assert LC and Remove}\hspace{1em} The claim holds for this rule by construction. If $\LC$ holds everywhere outside $\Br$ , and we know that $\LC(x)$ holds, then we can conclude that $\LC$ holds everywhere outside $\Br \setminus \{x\}$.

It only remains to show that the claim holds for larger well-behaved programs obtained by composing smaller well-behaved programs using sequencing, branching, or looping constructs. The proof here is trivial as the argument for sequencing is trivial (we can think of a loop as unboundedly many sequenced compositions of the smaller well-behaved program): we can \emph{always} compose two well-behaved programs to obtain a well-behaved program.\qed

\section{Details For Well-Behaved Programming}
\label{app:wellbehaved}

\subsection*{General Construction for Automatically Checking Correctness of Impact Sets}

Fix a class with maps $\Ff \cup \Gg = \{f_1,f_2,\ldots f_n\}$ (includes both original and ghost fields) and an intrinsic definition $(\Gg, \LC, \varphi)$ over which we prove correctness of programs. Without loss of generality, let $f_1,\ldots f_k$ for some $k \leq n$ alone correspond to pointer fields (where the range is an object); the others we assume are data fields that range over background sorts. In the sequel we assume for simplicity that $\LC(x)$ only relates the fields of $x$ with those of $f_i(x)$ for $1 \leq i \leq k$, i.e., the local conditions only constrain the fields of $x$ with those of its neighboring objects that are ``one pointer hop'' away from $x$.

Consider a mutation $x.f \coloneqq y$ for some $f$ in $f_1$ through $f_n$ and an arbitrary $y$. It is clear that the only set of objects whose local condition can be impacted by this mutation are those that are one pointer hop away via an incoming or outgoing edge in the heap (seen as a directed graph with labeled edges corresponding to pointers), apart from $x$ itself. In general there can be unboundedly many such objects, but in our work we only handle impact sets that can be expressed as a finite set of terms over $x$ (see Section~\ref{sec:well-behaved} under `Rules for Constructing Well-Behaved Programs'). Note here that the impact set can be larger than the set of impacted objects as we only require that objects not belonging to the impact set retain that $\LC$ holds on them under mutation. However, we attempt here to construct of impact sets that are as small as possible.

Following the above discussion, let us assume that the impact set consists of terms from the following set:
\begin{center}
$\mathit{ImpactableObjects} = \{x, f_1(x),\ldots f_k(x)\} \cup \{\old(f(x)) \,|\, f \textrm{ is a pointer field}\}$
\end{center}

The reader may be inclined to suggest here that when $f$ is a pointer field, $y$ (the new value of $f(x)$) may also belong to the minimal impact set. However, this is not possible in general since $y$ is arbitrary, and in particular $y$ can be an object in the heap that is ``far away'' from $x$, i.e., not one pointer hop away (either incoming or outgoing). The same argument applies to terms over $y$. Therefore, if the (minimal) impact set is at all expressible as a set of terms over the vocabulary of the mutation statement it must be a subset of the terms in the set $\mathit{ImpactableObjects}$ defined above.

Let this subset of terms be $A$. We then generate the following triple to check that $A$ is in fact an impact set:

\begin{center}
$\vdash\; \{(\bigwedge_{t \in A} u \neq t) \land \LC(u) \land x \neq \nil\}\; x.f := y\; \{\LC(u)\}$
\end{center}

The triple says that any location $u$ that is not $A$ which satisfied $\LC$ before the mutation must continue to satisfy it after the mutation. As discussed in the main text, this validity of this triple can be check effectively by decision procedures over quantifier-free combinations of theories that are supported by SMT solvers~\cite{Z3,cvc4}.

Finally, we can compute a provably correct and minimal impact set automatically, if one exists, by considering subsets of $\mathit{ImpactableObjects}$ in turn and checking the corresponding triple as described above. However, in our experiments we compute impact sets manually and check their correctness automatically.

\section{Details for Case Studies in Section~\ref{sec:case-studies}}
\label{app:case-studies}

\iflong
In this appendix we provide further details for the various case studies discussed in the main text.
\else
In this appendix we provide further details for the various case studies discussed in the main text and detail some other case studies not featured in the main text.
\fi

\iflong\else
\subsection{Discussion on Sorted List Insertion (Section~\ref{sec:cs-slist-insert})}
\label{app:cs-slist-insert}

\fi

\iflong
\subsection{Continuing Discussion for BST Rotate (Section~\ref{sec:cs-bst-rotate})}
\label{app:cs-bst-rotate}

The following program for BST Rotate is based off of the local conditions and impact sets that appear in Section~\ref{sec:cs-bst-rotate}. 
\else

\fi

We present the fully annotated program below, with comments displaying the state of the broken set $Br$ at the corresponding point in the program.

{\footnotesize
\begin{alltt}
 \annotation{pre:} \ensuremath{\Br=\emptyset \land l(x)\neq\nil \land p(x)=xp}
 \annotation{post:} \ensuremath{\Br=\emptyset \land p(ret)=xp}
       \ensuremath{\land l(ret)=\old(l(l(x))) \land ret=\old(l(x)) \land r(ret)=x}
       \ensuremath{\land l(r(ret))=\old(r(l(x))) \land r(r(ret))=\old(r(x))}
 \funcname{bst_right_rotate}(x: \type{BST}, xp: \type{BST?}, Br: \type{Set(BST)}) 
 \keyword{returns} ret: \type{BST}, Br: \type{Set(BST)}
 \{
   \ghost{LCOutsideBr(x, Br);}
   \ghost{if (xp \ensuremath{\neq} nil) then \{}
        \ghost{LCOutsideBr(xp, Br);}
   \ghost{\}}
   \ghost{if (x.l \ensuremath{\neq} nil) then \{}
        \ghost{LCOutsideBr(x.l, Br);}
   \ghost{\}}
   \ghost{if (x.l \ensuremath{\neq} nil \ensuremath{\land} x.l.r \ensuremath{\neq} nil) then \{}
        \ghost{LCOutsideBr(x.l.r, Br);}
   \ghost{\}}
   \keyword{var} y := x.l;                \codecomment{// \{\}}
   Mut(x, l, y.r, Br);          \codecomment{// \{x, y\}}
   \keyword{if} (xp \ensuremath{\neq} nil) \keyword{then} \{
        \keyword{if} (x = xp.l) \keyword{then} \{
            Mut(xp, l, y, Br);  \codecomment{// \{xp, x, y\}}
        \}
        \keyword{else} \{
            Mut(xp, r, y, Br);  \codecomment{// \{xp, x, y\}}
        \}
   \}
   Mut(y, r, x, Br);            \codecomment{// \{xp, x, y, x.l\} (Note: x.l == old(y.r))}
   \codecomment{// (1): Repairing \texttt{x.l}}
   \ghost{if (x.l \ensuremath{\neq} nil) then \{}
        \ghost{Mut(x.l, p, x, Br);}     \codecomment{// \{xp, x, y, x.l\}}
   \ghost{\}}
   \codecomment{// (2): Repairing \texttt{x}}
   \ghost{Mut(x, p, y, Br);}            \codecomment{// \{xp, x, y, x.l\}}
   \ghost{Mut(x, min, if x = nil then x.k else x.l.min, Br);}    \codecomment{// \{xp, x, y, x.l\}}
   \codecomment{// (3): Repairing \texttt{y}}
   \ghost{Mut(y, p, xp, Br);}           \codecomment{// \{xp, x, y, x.l\}}
   \ghost{Mut(y, max, x.max, Br);}      \codecomment{// \{xp, x, y, x.l\}}
   \ghost{Mut(y, rank, if xp = nil then x.rank+1 else (xp.rank+x.rank)/2, Br);}     \codecomment{// \{xp, x, y, x.l\}}
   \ghost{AssertLCAndRemove(x.l, Br);}  \codecomment{// \{xp, x, y\}}
   \ghost{AssertLCAndRemove(x, Br);}    \codecomment{// \{xp, y\}}
   \ghost{AssertLCAndRemove(y, Br);}    \codecomment{// \{xp\}}
   \ghost{AssertLCAndRemove(xp, Br);}   \codecomment{// \{\}}
   ret := y \codecomment{// return y}
 \}
\end{alltt}
}

\subsection{Discussion on Sorted List Reversal (Section~\ref{sec:cs-slist-reverse})}
\label{app:cs-slist-reverse}

What follows are the complete local conditions and impact sets for Sorted List Reverse:

\begin{figure*}[h]
\begin{equation*}
\label{eq:slist-reverse-lc}
\begin{aligned}
\LC \equiv \forall x.\, \prev(x) \neq \nil \Rightarrow\; & \nxt(\prev(x)) = x\\
\land~ \nxt(x) \neq \nil \Rightarrow\; & \prev(\nxt(x)) = x\\
&\land~ \len(x) = \len(\nxt(x)) + 1\\
&\land~ \keys(x) = \keys(\nxt(x)) \cup \{ \key(x) \}\\
&\land~ \hslist(x) = \hslist(\nxt(x)) \uplus \{ x \}\\
&\land~ \sorted(x) \Rightarrow \key(x) \leq \key(\nxt(x))\\
&\qquad\qquad\qquad \land \sorted(x) = \sorted(\nxt(x))\\
&\land~ \revsorted(x) \Rightarrow \key(x) \geq \key(\nxt(x))\\
&\qquad\qquad\qquad \land \revsorted(x) = \revsorted(\nxt(x))\\
~~~~~\land~ (\nxt(x) = \nil \Rightarrow\; & \len(x) = 1 \land \keys(x) = \{ x \} \land \hslist(x) =\{x\})\\
\end{aligned}
\end{equation*}
\caption{Full local condition for lists for Sorted List Reverse}
\end{figure*}

\begin{table*}[h]
\begin{center}
\begin{tabular}{ |c|c| } 
 \hline
 Mutated Field $f$ & Impacted Objects $A_f$ \\ 
 \hline
 $\nxt$ & $\{x, \old(\nxt(x))\}$ \\ 
 $\key$ & $\{x, \prev(x)\}$\\
 $\prev$ & $\{x, \old(\prev(x))\}$\\
 $\len$ & $\{x, \prev(x)\}$ \\
 $\keys$& $\{x, \prev(x)\}$ \\
 $\hslist$& $\{x, \prev(x)\}$ \\
 $\sorted$& $\{x, \prev(x)\}$ \\
 $\revsorted$& $\{x, \prev(x)\}$ \\
 \hline
\end{tabular}
\end{center}
\caption{Full impact sets for lists for Sorted List Reverse}
\label{tab:slist-reverse-impact}
\end{table*}

The following program reverses a sorted list as defined by the local condition above. We annotate this program with comments on the current composition of the broken set according to the rules of Table \ref{tab:slist-reverse-impact}.

{\footnotesize
\begin{alltt}
 \annotation{pre:} \ensuremath{\Br=\emptyset \land \varphi(x) \land \sorted(x)}
 \annotation{post:} \ensuremath{\Br'=\emptyset \land \varphi(ret) \land \revsorted(ret) \land \keys(ret)=\old(\keys(x)) \land \hslist(ret)=\old(\hslist(x))}
 \funcname{sorted_list_reverse}(x: \type{C}, Br: \type{Set(C)}) 
 \keyword{returns} ret: \type{C}, Br: \type{Set(C)}
 \{
   \ghost{LCOutsideBr(x, Br);}
   \keyword{var} cur := x;
   ret := null;
   \keyword{while} (cur \ensuremath{\neq} nil)
      \keyword{invariant} \ensuremath{cur\neq\nil \Rightarrow LC(cur)\land\sorted(cur)\land\varphi(cur)}
      \keyword{invariant} \ensuremath{ret\neq\nil \Rightarrow LC(ret)\land\revsorted(ret)\land\varphi(ret)}
      \keyword{invariant} \ensuremath{cur\neq\nil \land ret\neq\nil \Rightarrow \key(ret) \leq key(cur)}
      \keyword{invariant} \ensuremath{\old(\keys(x))=\ite(cur=\nil, \emptyset, \keys(cur))\cup\ite(ret=\nil, \emptyset, \keys(ret))}
      \keyword{invariant} \ensuremath{\old(\hslist(x))=\ite(cur=\nil, \emptyset, \hslist(cur))\cup\ite(ret=\nil, \emptyset, \hslist(ret))}
      \keyword{invariant} \ensuremath{Br=\emptyset}
      \keyword{decreases} \ensuremath{\ite(cur\neq\nil, 0, \len(cur))}
   \{
      \keyword{var} tmp := cur.next;                \codecomment{// \{\}}
      \ghost{if (tmp \ensuremath{\neq} nil) then \{}            
        \ghost{LCOutsideBr(tmp, Br);}             \codecomment{// \{\}}
        \ghost{Mut(tmp, p, nil, Br);}             \codecomment{// \{cur, tmp\}}
      \ghost{\}}
      Mut(cur, next, ret, Br);            \codecomment{// \{cur, tmp\}}
      \ghost{if (ret \ensuremath{\neq} nil) then \{}
        \ghost{Mut(ret, p, cur, Br);}             \codecomment{// \{cur, tmp, ret\}}
      \}
      \ghost{Mut(cur, keys, 
        \{cur.k\} \ensuremath{\cup} (if cur.next=nil then \ensuremath{\varphi} else cur.next.keys), Br);}    \codecomment{// \{cur, tmp, ret\}}
      \ghost{Mut(cur, hslist, 
        \{cur\} \ensuremath{\cup} (if cur.next=nil then \ensuremath{\varphi} else cur.next.hslist), Br);}    \codecomment{// \{cur, tmp, ret\}}
      \ghost{if (cur.next \ensuremath{\neq} nil \ensuremath{\land} (cur.key \ensuremath{>} cur.next.key \ensuremath{\lor} \ensuremath{\neg}cur.next.sorted)) \{}
        \ghost{Mut(cur, sorted, false, Br);}      \codecomment{// \{cur, tmp, ret\}}
      \ghost{\}}
      \ghost{Mut(cur, rev\_sorted, true, Br);}     \codecomment{// \{cur, tmp, ret\}}
      \ghost{AssertLCAndRemove(cur, Br);}         \codecomment{// \{tmp, ret\}}
      \ghost{AssertLCAndRemove(ret, Br);}         \codecomment{// \{tmp\}}
      \ghost{AssertLCAndRemove(tmp, Br);}         \codecomment{// \{\}}
      ret := cur;
      cur := tmp;
   \}
   \codecomment{// The current value of ret is returned}
 \}
\end{alltt}
}

\subsection{Discussion on Circular List Insert Back (Section~\ref{sec:cs-cl-insertback})}
\label{app:cs-cl-insertback}

We first provide the local conditions and impact sets for circular lists. 

\begin{figure*}[h]
\begin{equation*}
\begin{aligned}
\LC \equiv \forall x.\, \nxt(x) \neq \nil \; & \land \; \prev(x) \neq \nil\\
\land \;~~~ \nxt(\prev(x)) \; &= \; x \; \land \; \prev(\nxt(x)) \; = \; x\\
\land~ \lst(x) = x \Rightarrow\; & \len(x) = 0 \land \revlen(x) = 0\\
&\land~ \lst(x) = \lst(\nxt(x))\\
&\land~ \nxt(x) = x \Rightarrow \keys(x) = \emptyset \land \hslist(x) = \{x\}\\
&\land~ \nxt(x) \neq x \Rightarrow \keys(x) = \keys(\nxt(x)) \hspace{8em}\text{(C1)}\\
&\qquad\qquad\qquad\quad\; \land \hslist(x) = \{x\} \cup \hslist(\nxt(x)) \hspace{4.1em}\text{(C2)}\\
\land~ \lst(x) \neq x \Rightarrow\; & \len(x) = \len(\nxt(x)) + 1\\
&\land~ \revlen(x) = \revlen(\prev(x)) + 1\\
&\land~ \nxt(x) = \lst(x) \Rightarrow \keys(x) = \{\key(x)\} \land \hslist(x) = \{x\}\\
&\land~ \nxt(x) \neq \lst(x) \Rightarrow \keys(x) = \{\key(x)\} \cup \keys(\nxt(x))\hspace{1em}\text{(C3)}\\
&\qquad\qquad\qquad\qquad\quad\; \land \hslist(x) = \{x\} \cup \hslist(\nxt(x))\hspace{2.18em}\text{(C4)}\\
&\qquad\qquad\qquad\qquad\quad\; \land x \not\in \hslist(\nxt(x))\\
&\land~ \lst(x) = \lst(\nxt(x))\\
&\land~ \lst(\lst(x)) = \lst(x)\\
&\land~ x \in \hslist(\lst(x))\\
&\land~ \prev(x) \in \hslist(\lst(x))\\
&\land~ \nxt(x) \in \hslist(\lst(x))
\end{aligned}
\end{equation*}
\caption{Full local condition for lists for Circular List Insert Back}
\label{fig:cl-insertback-lc}
\end{figure*}

The local condition $\LC$ for circular lists can be seen in Figure \ref{fig:cl-insertback-lc}. For use in loop invariants, we have defined two variants of the local condition. One of these variants is $LC_{MinusNode}(x, n)$, which can be seen as a predicate on nodes $x$ and $n$, and is formed from $\LC$ by replacing the clauses (C1), (C3), and (C4) in Figure \ref{fig:cl-insertback-lc} with the three clauses (C1'), (C3'), and (C4') in Figure \ref{fig:cl-insertback-lc-minnode}. Additionally, we have another variant: $LC_{Last}(x, n)$, which is formed from $LC$ by replacing the clause (C2) in Figure \ref{fig:cl-insertback-lc} with $\hslist(x) = \{x, n\} \cup \hslist(\nxt(x))$.

\begin{figure*}[h]
\begin{equation*}
\begin{aligned}
&(\keys(x) = \keys(\nxt(x)) \setminus \{\key(n)\} \lor \keys(x) = \keys(\nxt(x))) \hspace{2em}\text{(C1')}\\
&(\keys(x) = (\key(x) \cup \keys(\nxt(x))) \setminus \{\key(n)\} \hspace{8.75em}\text{(C3')}\\
&\quad~ \lor \keys(x) = (\key(x) \cup \keys(\nxt(x))))\\
&(\hslist(x) = (x \cup \hslist(\nxt(x))) \setminus \{n\})\hspace{12.1em}\text{(C4')}
\end{aligned}
\end{equation*}
\caption{Alterations to Figure \ref{fig:cl-insertback-lc} to form $LC_{MinusNode}$}
\label{fig:cl-insertback-lc-minnode}
\end{figure*}

Note that in this example as well as other benchmarks where we introduce scaffolding nodes, in order to prove a bound on the impact set, we require that a precondition $\phi$ holds before we mutate particular fields of nodes. The fields, preconditions, and impact sets for every node can be seen in Table \ref{tab:cl-insertback-impact}. Note that our benchmark contains another manipulation macro, \texttt{AddToLastHsList(x, n, Br)}, which, if \texttt{x} is a scaffolding node (or $\lst(x) = x$), adds the node \texttt{n} to the set $\hslist(x)$. The precondition for invoking this macro is that $\lst(x) = x$, and the only object impacted by the macro is $\{x\}$. %

\begin{table*}[h]
\begin{center}
\begin{tabular}{ |c|c|c| } 
 \hline
 Mutated Field $f$ & Mutation Precond. $\phi$ & Impacted Objects $A_f$ \\ 
 \hline
 $\nxt$ & $\top$ & $\{x, \old(\nxt(x))\}$ \\ 
 $\key$ & $\top$ & $\{x, \prev(x)\}$\\
 $\prev$ & $\top$ & $\{x, \old(\prev(x))\}$\\
 $\lst$ & $\lst(x) \neq x \lor (\lst(x) = x \land \hslist(x) = \{x\})$ & $\{x, \prev(x)\}$\\
 $\len$ & $\top$ & $\{x, \prev(x)\}$ \\
 $\revlen$ & $\top$ & $\{x, \nxt(x)\}$ \\
 $\keys$& $\top$ & $\{x, \prev(x)\}$ \\
 $\hslist$& $\lst(x) \neq x \lor (\lst(x) = x \land \hslist(x) = \{x\})$  & $\{x, \prev(x)\}$ \\
 \hline
\end{tabular}
\end{center}
\caption{Full impact sets for lists for Circular List Insert Back}
\label{tab:cl-insertback-impact}
\end{table*}

\clearpage
We give the specification and program for Circular List Insert Back below.

{\footnotesize
\begin{alltt}
 \annotation{pre:} \ensuremath{\Br=\emptyset \land \nxt(x)=\lst(x)}
 \annotation{post:} \ensuremath{\Br=\emptyset \land \nxt(ret)=\lst(ret) \land \lst(ret)=\old(\lst(x))}
       \ensuremath{\land \keys(\lst(ret))=\old(\keys(\lst(x)))\cup\{k\} \land fresh(\hslist(\lst(ret))\setminus\old(\hslist(\lst(x))))}
 \funcname{circular_list_insert_back}(x: \type{C}, k: \type{Int} Br: \type{Set(C)}) 
 \keyword{returns} ret: \type{C}, Br: \type{Set(C)}
 \{
   \ghost{LCOutsideBr(x, Br);}
   \ghost{LCOutsideBr(x.next, Br);}
   \ghost{LCOutsideBr(x.prev, Br);}
   
   \keyword{var} last: C = x.next;
   \keyword{var} node: C;
   NewObj(node, Br);
   Mut(node, key, k, Br);
   Mut(node, next, x.next, Br);
   Mut(x, next, node, Br);

   \ghost{AddToLastHsList(last, node, Br);}
   \ghost{Mut(last, prev, node, Br);}
   \ghost{Mut(node, prev, x, Br);}
   \ghost{Mut(node, length, 1, Br);}
   \ghost{Mut(node, rev_length, 1 + node.prev.rev_length, Br);}
   \ghost{Mut(node, keys, \{k\}, Br);}
   \ghost{Mut(node, hslist, \{node\}, Br);}
   \ghost{Mut(node, last, node.prev.last, Br);}
   \ghost{AssertLCAndRemove(node, Br);}

   \ghost{\keyword{ghost var} cur: C = x;}
   \ghost{\keyword{label} PreLoop:}
   \ghost{\keyword{while} (cur \ensuremath{\neq} last)}
      \ghost{\keyword{invariant} \ensuremath{cur\neq{last}\Rightarrow}
              \ensuremath{Br=\{cur, last\}}
              \ensuremath{\land LC\sb{MinusNode}(cur, node)}
              \ensuremath{\land \lst(cur) = last}
              \ensuremath{\land LC\sb{Last}(last, node)}}
      \ghost{\keyword{invariant} \ensuremath{cur=last\Rightarrow LC\sb{MinusNode}(cur, node)}}
      \ghost{\keyword{invariant} \ensuremath{node \in \hslist(\nxt(cur))}}
      \ghost{\keyword{invariant} \ensuremath{k \in \keys(\nxt(cur))}}
      \ghost{\keyword{invariant} \keyword{Unchanged}@PreLoop\ensuremath{(node)}}
      \ghost{\keyword{invariant} \keyword{Unchanged}@PreLoop\ensuremath{(last)}}
      \ghost{\keyword{invariant} \ensuremath{Br \subseteq \{cur, last\}}}
      \ghost{\keyword{decreases} \ensuremath{\revlen(cur)}}
   \ghost{\{}
      \ghost{if (cur.prev \ensuremath{\neq} last) \{
        LCOutsideBr(cur.prev, Br);
      \}
      Mut(cur, length, cur.next.length + 1, Br);
      Mut(cur, hslist, cur.next.hslist + \{node\});
      Mut(cur, keys, cur.next.keys + \{node.k\});
      AssertLCAndRemove(cur, Br);
      cur := cur.prev;}
   \ghost{\}}

   \ghost{LCOutsideBr(node, Br);}
   \ghost{Mut(cur, keys, cur.next.keys, Br);}
   \ghost{AssertLCAndRemove(cur, Br);}
   \ghost{AssertLCAndRemove(node, Br);}
   
   ret := node;
 \}
\end{alltt}
}

\iflong
\subsection{Continuing Discussion for Sorted List Merge (Section~\ref{sec:cs-slist-merge})}
\label{app:cs-slist-merge}

We now present the full local conditions and impact sets for disjoint sorted lists, which combine sorted list conditions with monadic maps $list1, list2, list3$ that ensure that lists of a particular class (represented by these monadic maps) are disjoint from lists of another class.

\else

\smallskip
\noindent
We provide below the full local conditions and impact sets.
\fi

\begin{figure}[H]
\begin{equation}
\label{eq:dslist-lc}
\begin{aligned}
\LC \equiv \forall x&. (list1(x) \lor list2(x) \lor list3(x)) \\
&\land~ \neg(list1(x) \land list2(x)) \land \neg(list2(x) \land list3(x)) \\
&\land~ \neg(list1(x) \land list3(x)) \\
&\land~ (\prev(x) \neq \nil \Rightarrow \nxt(\prev(x)) = x)\\
&\land~ (\nxt(x) \neq \nil \Rightarrow \prev(\nxt(x)) = x\\
&\qquad\qquad\qquad\quad \land~ \len(x) = \len(\nxt(x)) + 1\\
&\qquad\qquad\qquad\quad \land~ \keys(x) = \keys(\nxt(x)) \cup \{ \key(x) \}\\
&\qquad\qquad\qquad\quad \land~ \hslist(x) = \hslist(\nxt(x)) \uplus \{ x \} \text{\quad(disjoint union)}\\
&\qquad\qquad\qquad\quad \land~ \key(x) \leq \key(\nxt(x)))\\
&\land~ (\nxt(x) = \nil \Rightarrow \len(x) = 1 \land \keys(x) = \{ \key(x) \} \land \hslist(x) = \{ x \})\\
&\land~ (list1(x) \Rightarrow (\nxt(x)\neq\nil \Rightarrow list1(\nxt(x)))) \\
&\land~ (list2(x) \Rightarrow (\nxt(x)\neq\nil \Rightarrow list2(\nxt(x)))) \\
&\land~ (list3(x) \Rightarrow (\nxt(x)\neq\nil \Rightarrow list3(\nxt(x)))) \\
\end{aligned}
\end{equation}
\caption{Full local condition for lists for Sorted List Reverse}
\end{figure}

Note that we also have a variation of the local condition $LC_{NC}$, used in ghost loop invariants, which is similar to Equation \ref{eq:dslist-lc}, except the final three conjuncts (those enforcing closure on $list1, list2, list3$) are removed. This is done when converting an entire list from one class to another (i.e., converting from $list1$ to $list3$). The following are the full impact sets for all fields of this data structure.

\begin{figure}[H]
\begin{center}
\begin{tabular}{ |c|c| } 
 \hline
 Mutated Field $f$ & Impacted Objects $A_f$ \\ 
 \hline
 $\nxt$ & $\{x, \old(\nxt(x))\}$ \\ 
 $\key$ & $\{x, \prev(x)\}$\\
 $\prev$ & $\{x, \old(\prev(x))\}$\\
 $\len$ & $\{x, \prev(x)\}$ \\
 $\keys$& $\{x, \prev(x)\}$ \\
 $\hslist$& $\{x, \prev(x)\}$ \\
 $list1$& $\{x, \prev(x)\}$ \\
 $list2$& $\{x, \prev(x)\}$ \\
 $list3$& $\{x, \prev(x)\}$ \\
 \hline
\end{tabular}
\end{center}
\caption{Full impact sets for disjoint sorted lists}
\end{figure}

\end{document}